\documentclass[12pt]{article}
\usepackage{a4wide}
\usepackage{latexsym}
\usepackage{amsmath}
\usepackage{amsfonts}
\usepackage{amssymb}
\usepackage{amscd}
\usepackage{amsthm}
\usepackage{cite}
\usepackage{placeins}
\usepackage{bm}
\usepackage{axodraw2}
\usepackage{tcolorbox}
\tcbuselibrary{breakable}
\usepackage{pslatex}
\usepackage{graphicx}
\usepackage[latin1,utf8]{inputenc}
\usepackage[T1]{fontenc}
\usepackage{graphicx}
\usepackage{amsmath}
\usepackage{amsfonts}
\usepackage{acronym}
\usepackage{ifthen}

\usepackage{hyperref}

\allowdisplaybreaks

\newcommand{\bq}{\begin{eqnarray}}
\newcommand{\eq}{\end{eqnarray}}

\newcommand{\Eulerconstant}{\gamma_{\mathrm{E}}}

\newcommand{\Dint}{D_{\mathrm{int}}}

\newcommand{\eps}{\varepsilon}

\newcommand{\loopnumber}{l}

\newcommand{\nedges}{N}

\newcommand{\nexternal}{n_{\mathrm{ext}}}

\newcommand{\nexternalindependent}{e}

\newcommand{\NB}{N_B}

\newcommand{\NF}{N_F}

\newcommand{\NL}{N_L}

\newcommand{\NE}{N_E}

\newcommand{\NV}{n}

\newcommand{\ND}{N_D}

\newcommand{\laportaorder}{(a,w,o,|\mu|,\dots)}

\newcommand{\arbitraryscale}{M}

\newcommand{\Divisor}{P}
\newcommand{\divisor}{p}

\newcommand{\Baikovvariable}{\sigma}

\newcommand{\preall}{C_{\eps}}
\newcommand{\preabs}{C_{\mathrm{abs}}}
\newcommand{\prerel}{C_{\mathrm{rel}}}
\newcommand{\preclutch}{C_{\mathrm{clutch}}}
\newcommand{\prebaikov}{C_{\mathrm{Baikov}}}

\newcommand{\Fcomb}{F_{\mathrm{comb}}}
\newcommand{\Fgeom}{F_{\mathrm{geom}}}
\newcommand{\Fgen}{F}

\newcommand{\Hcomb}{H_{\mathrm{comb}}}
\newcommand{\Hgeom}{H_{\mathrm{geom}}}
\newcommand{\Hgen}{H}

\newcommand{\hcomb}{h_{\mathrm{comb}}}
\newcommand{\hgeom}{h_{\mathrm{geom}}}

\newcommand{\differentialform}{\Psi}

\newcommand{\Acomb}{\Omega_{\mathrm{comb}}}
\newcommand{\Ageom}{\Omega_{\mathrm{geom}}}
\newcommand{\Agen}{\Omega}

\newcommand{\absmu}{|\mu|}

\newcounter{algocounter}

\theoremstyle{plain}
\newtheorem{theoremcounter}{}[]

\newtheorem{examplecounter}{}[]

\newtheorem{theorem}[theoremcounter]{Theorem}
\newtheorem{corollary}[theoremcounter]{Corollary}

\newtheorem{lemma}[theoremcounter]{Lemma}

\newtheorem{myexample}[examplecounter]{Example}

\newboolean{longversion}
\setboolean{longversion}{false}

\begin{document}

\thispagestyle{empty}



\begin{center}
  {\Large\bf New algorithms for Feynman integral reduction and $\varepsilon$-factorised differential equations \\
  }
  \vspace{1cm}
  {\large The $\varepsilon$-collaboration: 
          Iris Bree${}^{a}$,
          Federico Gasparotto${}^{b}$,
          Antonela Matija\v{s}i\'c${}^{a}$,
          Pouria Mazloumi${}^{a}$,
          Dmytro Melnichenko${}^{a}$,
          Sebastian~P\"ogel${}^{c}$, 
          Toni~Teschke${}^{a}$, 
          Xing~Wang${}^{d}$,
          Stefan~Weinzierl${}^{a}$,
          Konglong~Wu${}^{e}$ and
          Xiaofeng~Xu${}^{f}$
\\
  \vspace{1cm}
      {\small \em ${}^{a}$ PRISMA Cluster of Excellence, Institut f{\"u}r Physik, Staudinger Weg 7,} \\
      {\small \em Johannes Gutenberg-Universit{\"a}t Mainz, D-55099 Mainz, Germany}\\
  \vspace{2mm}
      {\small \em ${}^{b}$ Bethe Center for Theoretical Physics, Universität Bonn, D-53115 Bonn, Germany} \\
  \vspace{2mm}
      {\small \em ${}^{c}$ Paul Scherrer Institut, CH-5232 Villigen PSI, Switzerland} \\
  \vspace{2mm}
      {\small \em ${}^{d}$ School of Science and Engineering,} \\
      {\small \em The Chinese University of Hong Kong, Shenzhen, 518172 Guangdong, China} \\
  \vspace{2mm}
      {\small \em ${}^{e}$ School of Physics and Technology,} \\
      {\small \em Wuhan University, No. 299 Bayi Road, Wuhan 430072, China} \\
  \vspace{2mm}
      {\small \em ${}^{f}$ Department of Physics, Xiamen University, Xiamen, 361005, China} 
  } 
\end{center}

\vspace{1.0cm}

\begin{abstract}\noindent
  {
In this paper, we give a detailed account of the algorithm outlined in \cite{e-collaboration:2025frv}
for Feynman integral reduction and $\varepsilon$-factorised differential equations.
The algorithm consists of two steps.
In the first step, we use a new geometric order relation in the integration-by-parts reduction to obtain a basis of master integrals,
whose differential equations on the maximal cut are of a Laurent polynomial form in the regularisation parameter $\varepsilon$ and compatible with a filtration.
This step works entirely with rational functions.
In a second step, we provide a method to $\varepsilon$-factorise the aforementioned Laurent differential equations.
The second step may introduce algebraic and transcendental functions.
We illustrate the versatility of the algorithm by applying it to different examples with a wide range of complexity.   
}
\end{abstract}

\vspace*{\fill}

\newpage

\tableofcontents

\newpage

\section{Introduction}
\label{sect:intro}
Feynman integrals play a pivotal role in  Quantum Field Theory (QFT). 
Advancements in the experimental programme at the Large Hadron Collider require constant progress in pushing the frontier of perturbative calculations further.  Besides collider phenomenology, Feynman integrals nowadays play an important role in the context of gravitational wave physics, due to the success of QFT-based methods in the computation of classical corrections in general relativity.
In both scenarios, the precision requirements of state-of-the-art computations necessitate the evaluation of Feynman integrals of increasing complexity. 
Current research focuses therefore on advancing the techniques for these calculations \cite{DHoker:2023khh,Marzucca:2023gto,delaCruz:2024xit,Baune:2024biq,Baune:2024ber,Jockers:2024uan,Gehrmann:2024tds,Pogel:2024sdi,Duhr:2024xsy,Gasparotto:2024bku,Duhr:2024uid,DHoker:2025szl,DHoker:2025dhv,Duhr:2025ppd,Duhr:2025tdf,Becchetti:2025oyb,Duhr:2025lbz,Chaubey:2025adn}.
At the level of the integrals, the complexity of Feynman integral computations is represented in the number of loops and kinematic invariants involved in the process.

Perturbative calculations are carried out with the help of computer algebra.
The available computing resources dictate which higher-order calculations can be done.
Almost all perturbative higher-order calculations rely on a few basic algorithms for Feynman integral reduction and the computation of master integrals.
These algorithms require significant computing resources in the form of memory and CPU time and are therefore often the bottleneck.
Improving the efficiency of these basic algorithms will have a direct impact on the higher-order calculations which can be done.

Feynman integral reduction is one of the basic algorithms.
Integration-by-parts identities~\cite{Tkachov:1981wb,Chetyrkin:1981qh} allow us to express a Feynman integral from a (large) set of Feynman integrals as a linear
combination of Feynman integrals from a smaller set. The Feynman integrals in the smaller set are called master integrals
and we may think of the master integrals as a basis of an (abstract) vector space.
The linear system of integration-by-parts identities is usually solved with the Laporta algorithm \cite{Laporta:2000dsw}. 
The Laporta algorithm is based on Gau{\ss} elimination in combination with an order relation.
The order relation defines which of two Feynman integrals is considered more complicated.
In each linear equation, the most complicated Feynman integral is eliminated in favour of simpler ones.
Up to now, standard Feynman integral reduction programs \cite{vonManteuffel:2012np,Smirnov:2014hma,Maierhoefer:2017hyi,Wu:2023upw,Guan:2024byi}
order Feynman integrals by sectors, but within a sector only an ad-hoc order is used, usually based on the number of dots and the number of irreducible scalar products.
While the number of master integrals is independent of the chosen order relation, the size of the reduction tables may depend on the choice.
In this paper, we propose an order relation derived from mathematics, which orders the Feynman integrals within sectors
according to their geometric properties.
Our pilot studies show that this order relation performs significantly better.
This can be traced back to the observation that 
this order relation eliminates to a large extent the occurence of spurious polynomials in the denominator.

A second basic algorithm is the computation of master integrals 
through the method of differential equations~\cite{Kotikov:1990kg,Kotikov:1991pm,Remiddi:1997ny,Gehrmann:1999as,Henn:2013pwa}.
For the master integrals, one derives (again by using integration-by-parts identities)
differential equations in the external invariants or internal masses.
In general, the right-hand side of this differential equation will involve rational functions 
in the kinematic variables $x$ and in the dimensional regularisation parameter $\eps$.
A differential equation is said to be in $\eps$-factorised form if the only explicit occurrence of the
dimensional regularisation parameter is a simple overall prefactor on the right-hand side of the differential equation \cite{Henn:2013pwa}.
An $\eps$-factorised differential equation has the advantage that a solution in terms of iterated integrals \cite{Chen} is straightforward,
provided appropriate boundary values are known.
The required boundary values are usually not the bottleneck,
in fact, they can be recursively reduced to single-mass vacuum integrals \cite{Liu:2022chg,Liu:2017jxz,Liu:2022mfb}.
The application of the above-mentioned $\eps$-factorised form was successfully employed to compute Feynman integrals 
evaluating to polylogarithmic functions depending on many kinematic scales, see e.g.~ref.~\cite{Henn:2025xrc}. 
Besides this, there have been numerous contributions investigating the role of $\eps$-factorisation in the case of Feynman integrals 
whose underlying geometry is non-trivial, e.g.~elliptic curve(s)~\cite{Adams:2018yfj,Honemann:2018mrb,Bogner:2019lfa,Muller:2022gec,Giroux:2022wav,Dlapa:2022wdu,Gorges:2023zgv,Delto:2023kqv,Jiang:2023jmk, Ahmed:2024tsg,Giroux:2024yxu,Duhr:2024bzt,Schwanemann:2024kbg,Marzucca:2025eak,Becchetti:2025oyb,Becchetti:2025rrz, Ahmed:2025osb, Chen:2025hzq,Coro:2025vgn}, Calabi--Yau manifolds~\cite{Pogel:2022vat,Pogel:2022yat,Pogel:2022ken,Duhr:2022dxb,Forner:2024ojj,Frellesvig:2024rea,Duhr:2025lbz,Maggio:2025jel,Duhr:2025kkq,Pogel:2025bca} or higher genus surfaces~\cite{Duhr:2024uid}. 
One of the distinct features of these cases is that the transformation to the $\eps-$factorised form is not restricted to algebraic functions (contrary to the polylogarithmic case). 
Going beyond case-by-case studies of specific Feynman integrals, 
there are a few methods which allow the construction of an $\eps$-factorised differential equation under specific circumstances,
for example, by restricting to the case of multiple polylogarithms \cite{Moser:1959,Lee:2014ioa,Lee:2017oca,Prausa:2017ltv,Gituliar:2017vzm,Meyer:2017joq},
by requiring a guess of a good initial basis \cite{Gorges:2023zgv}
or by requiring advance knowledge of the alphabet \cite{Dlapa:2022wdu}.
A key observation of \cite{e-collaboration:2025frv} is that the geometric order relation in the Laporta algorithm leads to a differential equation,
whose right-hand-side is on the maximal cut a Laurent polynomial in the dimensional regularisation parameter $\eps$ and for each entry of the connection matrix
the occurring powers of $\eps$ are compatible with a filtration $F^{\bullet}$.
We call such a differential equation an $F^{\bullet}$-compatible differential equation.
In this paper, we give a constructive proof that we may always transform an $F^{\bullet}$-compatible differential equation to an $\eps$-factorised form.

The improvements of this paper will be beneficial both to analytical approaches for the computation of Feynman integrals~\cite{Chicherin:2020oor,Chicherin:2021dyp,Becchetti:2025rrz,Coro:2025vgn} 
as well as to (semi-) numerical approaches~\cite{Hidding:2020ytt,Liu:2022chg,Liu:2017jxz,Liu:2022mfb,Armadillo:2022ugh,Prisco:2025wqs,PetitRosas:2025xhm}.
In addition, a general algorithm for an $\eps$-factorised differential equation offers insight into the underlying mathematical structures, 
like the singularity structure and the function space.
It also provides a playground for some mathematical conjectures~\cite{Drummond:2017ssj,Chicherin:2020umh,Pokraka:2025ali,Kristensson:2021ani,Wilhelm:2022wow}. 

In this paper, we give a detailed description of the algorithm outlined in ref.~\cite{e-collaboration:2025frv}.
Introductions to the main concepts of the method can also be found in refs.~\cite{e-collaboration:2026aup,e-collaboration:2026fbo}.
The procedure consists of two steps. 
The first step consists of Laporta reduction with a new order relation, which orders the Feynman integrals within a sector by their geometric properties.
As we do not want to rely on properties of specific geometries like elliptic curves or K3-surfaces,
we need one more layer of abstraction, which treats all these geometries in one single framework.
Twisted cohomology \cite{Yoshida:book,Aomoto:book} and ideas from classical Hodge theory \cite{Deligne:1970,Deligne:1971,Deligne:1974,Carlson,Voisin_book}
provide this framework.
The integrands of Feynman integrals can be viewed as twisted cohomology classes \cite{Mastrolia:2018uzb,Frellesvig:2019uqt}.
We study the space of differential forms associated with a given family of Feynman integrals in the Baikov representation.
For the order relation, we essentially count the number of poles and the number of consecutive non-zero residues.
We observe that this order relation always gives a basis of master integrals with an $F^{\bullet}$-compatible differential equation.
At present, we cannot prove this observation in a strict mathematical sense.
However, we do not know a counter-example either.
 
In the second step, we show how to systematically remove the unwanted terms in $\eps$ through the solution of differential constraints. 
For the second step, we give a rigorous mathematical proof.
These constraints are obtained by decomposing the Laurent polynomial differential equations into different orders of $\eps$. 
For each order, we construct a rotation matrix such that in the rotated differential equations, this order vanishes. 
The elements of these rotation matrices are determined by the solution of the aforementioned differential constraints, which are simpler than the original linearly coupled differential system.

The two steps are notably independent of prior knowledge of the underlying geometry of the Feynman integral.
We emphasise that the first step only involves rational functions, while the second step may involve algebraic and transcendental functions.
This suggests that an $F^{\bullet}$-compatible differential equation is the best that can be achieved by restricting to rational transformations.

This paper is organised as follows: 
In section~\ref{sect:definitions}, we introduce the notation used throughout the paper, as well as define relevant concepts. 
In section~\ref{sect:set_up}, we construct the mathematical setup of our methodology. 
In particular, we define pole order and residues associated with a differential form. 
We use these two objects to define filtrations of differential forms. 
Furthermore, we establish the underlying geometry corresponding to Feynman integrals as well as the sub-geometries associated with different localisations in the twist. 
In section~\ref{sect:method}, we present our algorithm. 
We first outline the steps of the algorithm that lead to a set of candidate master integrals whose differential equation is a Laurent polynomial in $\eps$. 
Thereafter, we explain the method to systematically $\eps-$factorise the Laurent series differential equations via solutions of a set of differential constraints. 
In section~\ref{sect:examples}, we showcase our method in different examples. 
These examples include a wide range of possible cases, which underlines the versatility of our algorithm. 
Finally, section~\ref{sect:conclusions} contains our conclusions and an outline towards possible future directions. 

\newpage 

\section{Notation, definitions and review}
\label{sect:definitions}


\subsection{Summary of the notation}

For the convenience of the reader, we summarise in this paragraph the notation used throughout this paper.
\begin{center}
\begin{tabular}{lll}
$\NF = N_{\mathrm{Fibre}}$: & Number of master integrals, & \\
 & master integrals are denoted by & $I = (I_1, ..., I_{\NF})^T$. \\
 & & \\
$\NB = N_{\mathrm{Base}}$: & Number of kinematic variables, & \\
 & kinematic variables are denoted by & $x=(x_1, ..., x_{\NB})^T$. \\ 
 & \\
$\NE = N_{\mathrm{Edges}}$: & Number of Baikov variables, & \\
 & Baikov variables are denoted by & $\Baikovvariable=(\Baikovvariable_1, ..., \Baikovvariable_{\NE})^T$. \\ 
 & \\
$\NV = N_{\mathrm{Variables}}$: & Number of Baikov variables on a maximal cut, & \\
 & Baikov variables on a maximal cut are denoted by & $z=(z_1, ..., z_{\NV})^T$. \\ 
 & \\
$\ND = N_{\mathrm{Divisors}}$: & Number of divisors, & \\
 & Divisors are denoted by & $\Divisor=(\Divisor_1, ..., \Divisor_{\ND})^T$. \\ 
 & \\
$\NL = N_{\mathrm{Letters}}$: & Number of letters, & \\
 & differential one-forms are denoted by & $ \omega=(\omega_1, ..., \omega_{\NL})^T$. \\
\end{tabular}
\end{center}
By default, all vectors are regarded as column vectors, hence the explicit use of the transpose.
We may view the kinematic space as the projective space ${\mathbb C}{\mathbb P}^{\NB}$.
Homogeneous coordinates are denoted by $[x_0:x_1:\dots:x_{\NB}]$, 
and $x=(x_1, ..., x_{\NB})^T$ are coordinates in the chart $x_0=1$.
In a similar way, we view the Baikov space on a maximal cut as the projective space ${\mathbb C}{\mathbb P}^{\NV}$.
Homogeneous coordinates are denoted by $[z_0:z_1:\dots:z_{\NV}]$, 
and $z=(z_1, ..., z_{\NV})^T$ are coordinates in the chart $z_0=1$.
The differential with respect to the kinematic variables is denoted by
\bq
\label{def_d_B}
 d_B & = & \sum\limits_{j=0}^{\NB} dx_j \frac{\partial}{\partial x_j}.
\eq
The differential with respect to the Baikov variables on a maximal cut is denoted by
\bq
\label{def_d_F}
 d_F & = & \sum\limits_{j=0}^{\NV} dz_j \frac{\partial}{\partial z_j}.
\eq
For monomials in several variables $z_0,z_1,\dots,z_{\NV}$ we use the following notation:
\bq
 {\bf z}^{\bf \kappa}
 & = &
 \prod\limits_{j=0}^{\NV} z_j^{\kappa_j}.
\eq
We denote polynomials in the chart $z_0=1$ of ${\mathbb C}{\mathbb P}^{\NV}$ by small letters, e.g. $p(z_1,\dots,z_{\NV})$.
We denote the corresponding homogeneous polynomials on ${\mathbb C}{\mathbb P}^{\NV}$ 
(defined as the $d$-homogenisation of $p$, where $d=\deg p$)
by capital letters, e.g. $P(z_0,z_1,\dots,z_{\NV})$.

An $\eps$-factorised differential equation is of the form
\bq
 d_B K & = & \eps A K,
\eq
where $K$ is a vector of $\NF$ master integrals and the $(\NF \times \NF)$-matrix $A$ is independent of $\eps$.
We write
\bq
 A & = &
 \sum\limits_{j=1}^{\NL} C_j \omega_j,
\eq
where the $C_j$'s are $(\NF \times \NF)$-matrices, whose entries are (rational) numbers and the $\omega_j$'s are differential one-forms, called the letters. 


\subsection{Feynman integrals}
\label{sect:feynman_integrals}

Let $G$ be a Feynman graph with $\loopnumber$ loops, $\nexternal$ legs and $\NE$ propagators.
We will consider scalar Feynman integrals associated to this graph.
We denote by $q_j$ the momentum flowing through edge $e_j$ (with respect to a chosen orientation of the edge),
by $m_j$ the mass of the particle propagating through edge $e_j$, and by $\nu_j \in {\mathbb Z}$ the power to which
this propagator occurs.
We set
\bq
 \Baikovvariable_j & = & -q_j^2 + m_j^2.
\eq
The $\Baikovvariable_j$'s are the inverse propagators.
Although we consider in this paper standard Feynman propagators with a quadratic dependence on the loop momenta, the extension 
towards propagators with a linear dependence on the loop momenta is straightforward.
The latter occur for example in the eikonal approximation.

Without loss of generality, we assume that the graph $G$ has a Baikov representation.
This means that 
\bq
 \NE & = &
 \frac{1}{2} \loopnumber \left(\loopnumber+1\right) + \nexternalindependent \loopnumber,
\eq
where $\nexternalindependent$ denotes the number of independent external momenta
and any scalar product involving a loop momentum can be written as a linear combination of the inverse propagators and a constant term.

We are interested in the family of Feynman integrals
\bq
\label{def_feynman_integral}
 I_{\nu_1 \dots \nu_{\NE}}\left(\Dint, \eps, x \right)
 & = &
 e^{\loopnumber \eps \Eulerconstant} \left(\arbitraryscale^2\right)^{\nu-\frac{\loopnumber D}{2}}
 \int \prod\limits_{r=1}^{\loopnumber} \frac{d^Dk_r}{i \pi^{\frac{D}{2}}} 
 \prod\limits_{j=1}^{\NE} \Baikovvariable_j^{-\nu_j},
\eq
where $D=\Dint-2\eps$ (with $\Dint \in {\mathbb Z}$ and $\eps \in {\mathbb C}$) 
denotes the number of space-time dimensions within dimensional regularisation,
$\eps$ is the dimensional regularisation parameter,
$\gamma_E$ denotes the Euler-Mascheroni constant, 
$\arbitraryscale$ is an arbitrary scale introduced to render the Feynman integral dimensionless.
Setting $\arbitraryscale^2$ equal to a Lorentz invariant or to a mass squared will turn the Feynman integral into a (in general multi-valued) function defined on
projective space ${\mathbb C}{\mathbb P}^{\NB}$.
We will always assume that such a choice is made.
Here, $\NB$ denotes the number of (dimensionless) kinematic variables. 
Thus we may view $I_{\nu_1 \dots \nu_{\NE}}$ as a multi-valued function
\bq
 I_{\nu_1 \dots \nu_{\NE}} & : & {\mathbb Z} \times {\mathbb C} \times {\mathbb C}{\mathbb P}^{\NB} \rightarrow {\mathbb C},
 \nonumber \\
 & & \left(\Dint, \eps, x \right) \mapsto {\mathbb C}.
\eq
The quantity $\nu$ is defined by
\bq
 \nu & = &
 \sum\limits_{j=1}^{\NE} \nu_j.
\eq
The kinematic variables are denoted by $x=(x_1,\dots,x_{\NB})$
and correspond to affine coordinates of the chart $x_0=1$ of ${\mathbb C}{\mathbb P}^{\NB}$.
To clarify our conventions we consider the example of massless $(2 \rightarrow 2)$-scattering.
In this case we have $\NB=1$ and we may choose $x_1=s/t$ as dimensionless kinematic variable.
We may view $x_1$ as the affine coordinate of the chart $x_0=1$ of ${\mathbb C}{\mathbb P}^{1}$.

For Minkowski spacetime, we have $\Dint=4$.
We stress that the algorithms we present works for any $\Dint \in {\mathbb Z}$.
As the value of $\Dint$ is often clear from the context, 
we will suppress in the following the dependence on $\Dint$ and simply write
$I_{\nu_1 \dots \nu_{\NE}}(\eps, x)$ for a Feynman integral.
For each integral $I_{\nu_1 \dots \nu_{\NE}}(\eps, x)$
we define 
\bq
 S & = & \left\{ \; j \; | \; \nu_j > 0 \; \right\}.
\eq  
and the sector id of this integral by
\bq
 N_\mathrm{id}
 & = & \sum\limits_{j\in S} 2^{j-1}
 \; = \; 
 \sum\limits_{j=1}^{\NE} 2^{j-1} \Theta\left(\nu_j\right).
\eq
Here, $\Theta(x)$ denotes the Heaviside step function, defined by $\Theta(x)=1$ for $x>0$ and $\Theta(x)=0$ otherwise.
We further define for each integral $I_{\nu_1 \dots \nu_{\NE}}$ the integers
\begin{align}
 \nedges 
 & = \sum\limits_{j=1}^{\NE} \Theta\left(\nu_j\right),
 &
 r 
 & = \sum\limits_{j=1}^{\NE} \nu_j \Theta\left(\nu_j\right),
 &
 s
 & = - \sum\limits_{j=1}^{\NE} \nu_j \Theta\left(-\nu_j\right).
\end{align}
The quantity $ \nedges $ denotes the number of propagators, the quantity $(r- \nedges )$ denotes the number of dots and
$s$ denotes the number of irreducible scalar products in the numerator.
It is well-known that the members of a family of Feynman integrals decompose into sectors.
A sector is uniquely specified either by the set $S$ or the sector id.
A sector specified by the set $S_1$ is called a sub-sector of a sector $S_2$, if
\bq
 S_1 & \subsetneq & S_2.
\eq
We will often work on the maximal cut of a sector $S$. 
This means that we work modulo sub-sectors.
In other words, on the maximal cut, we set all sub-sectors to zero.


\subsection{Integration-by-parts and the Laporta algorithm}

Integration-by-parts identities \cite{Tkachov:1981wb,Chetyrkin:1981qh}
are based on the fact that, within dimensional regularisation, the integral
of a total derivative vanishes
\bq
\label{basic_ibp_relation}
 e^{\loopnumber \eps \Eulerconstant} \left(\arbitraryscale^2\right)^{\nu-\frac{\loopnumber D}{2}}
 \int 
 \prod\limits_{r=1}^{\loopnumber} \frac{d^Dk_r}{i \pi^{\frac{D}{2}}}
 \;\;
 \frac{\partial}{\partial k_i^\mu} \left( q_{\mathrm{IBP}}^\mu
 \;\;
 \prod\limits_{j=1}^{\NE} \Baikovvariable_j^{-\nu_j} \right)
 & = & 0,
\eq
i.e. there are no boundary terms in the loop momentum representation. 
The vector $q_{\mathrm{IBP}}$ can be any linear combination of the external momenta and the loop momenta.
Working out the derivatives leads to linear relations among integrals
with different sets of indices $(\nu_1, \dots, \nu_{\NE})$.
The Laporta algorithm \cite{Laporta:2000dsw}
introduces an order relation among the members of a family of Feynman integrals and eliminates
in each linear relation the most complicated Feynman integral.
Typical order criteria are the tuples
\bq
\label{chapter_iterated_integrals:isp_basis}
 \left( \nedges, N_{\mathrm{id}}, r, s, \dots \right),
 & \mbox{or} &
 \left( \nedges, N_{\mathrm{id}}, s, r, \dots \right),
\eq
together with the lexicographical order.
The dots stand for further criteria needed to distinguish inequivalent integrals.
The order criteria in eq.~(\ref{chapter_iterated_integrals:isp_basis}) sort Feynman integrals by sectors, but within a sector, the criteria are ad-hoc.

Integration-by-parts identities allow us to express any member from the family $I_{\nu_1 \dots \nu_{\NE}}$ as a 
finite linear combination of master integrals.
We denote the master integrals by $I=(I_1, \dots, I_{\NF})$
with
\bq
 I_i\left(\eps, x\right) & = & I_{\nu_{i 1} \dots \nu_{i \NE} }\left(\eps, x\right), 
 \;\;\;\;\;\; \nu_{i j} \; \in \; {\mathbb Z}.
\eq
It is important to note that the order relation defines which Feynman integrals are chosen as master integrals, i.e. different choices for the order relation lead to different sets of master integrals.


\subsection{The method of differential equations}

The method of differential equations \cite{Kotikov:1990kg,Kotikov:1991pm,Remiddi:1997ny,Gehrmann:1999as,Henn:2013pwa} is a standard tool for the computation of Feynman integrals.
The master integrals $I$ satisfy a first-order system of differential equations
\bq
\label{differential_equation}
 d_B I\left(\eps, x\right)
 & = &
 A\left(\eps,x\right) I\left(\eps, x\right),
\eq
where $A(\eps,x)$ is a $\NF \times \NF$-matrix, whose entries are differential one-forms, rational
in $\eps$ and $x$.
The differential $d_B$ is defined in eq.~(\ref{def_d_B}) and denotes the differential with respect to the kinematic variables. 
The differential equation is obtained by carrying out the derivative under the integral sign and by using integration-by-parts
identities to reexpress the resulting expression in terms of master integrals.

We may change the basis of master integrals by a rotation
\bq
 I\left(\eps, x\right) & = & R\left(\eps, x\right) \tilde{I}\left(\eps, x\right).
\eq
In the new basis, we have again a differential equation of the form
\bq
 d_B \tilde{I}\left(\eps, x\right)
 & = &
 \tilde{A}\left(\eps,x\right) \tilde{I}\left(\eps, x\right),
\eq
where $\tilde{A}(\eps,x)$ is related to $A(\eps,x)$ by
\bq
 \tilde{A}\left(\eps,x\right)
 & = &
 R^{-1}\left(\eps, x\right) A\left(\eps,x\right) R\left(\eps, x\right) - R^{-1}\left(\eps, x\right) d_B R\left(\eps, x\right).
\eq
If $R(\eps, x)$ is rational in $\eps$ and $x$, then $\tilde{A}(\eps,x)$ is rational in $\eps$ and $x$ as well.


\subsection{The Baikov representation}

Feynman integrals have a Baikov representation \cite{Baikov:1996iu}.
We may either use the democratic Baikov representation or a loop-by-loop Baikov representation \cite{Frellesvig:2017aai}.
We require the Baikov representation to be a ``good'' representation (to be defined below).
It will be advantageous to work with the simplest good representation.
This is usually a loop-by-loop Baikov representation.

\subsubsection{The democratic Baikov representation}

Let us start with the democratic Baikov representation. 
The democratic Baikov representation is, from a notational perspective, the simplest (there is only one polynomial), but it is not the simplest from a computational point of view.
We assume that we have $\NE$ Baikov variables and $\nedges$ internal edges.
Without loss of generality, we may assume that the first $\nedges$ Baikov variables correspond to the internal edges.
The integral measure can be written as
\bq
 e^{\loopnumber \eps \Eulerconstant} 
 \left(\arbitraryscale^2\right)^{-\frac{\loopnumber D}{2}}
 \prod\limits_{r=1}^{\loopnumber} \frac{d^Dk_r}{i \pi^{\frac{D}{2}}}
 & = &
 C'
 \; b\left(\Baikovvariable_1,...,\Baikovvariable_{\NE}\right)^{\frac{D-\loopnumber-{\nexternalindependent}-1}{2}} 
 \; 
 \frac{d^{\NE}\Baikovvariable}{\left(2\pi i\right)^{\NE}},
\eq
where $C'$ is an irrelevant prefactor,
$\nexternalindependent$ is the number of independent external momenta
and $b\left(\sigma\right)$ the Baikov polynomial.
The Baikov polynomial is given as the Gram determinant of the independent loop momenta and the independent external momenta.
Let ${\mathcal C}_{\mathrm{maxcut}}$ be a contour which corresponds for the first $\nedges$ variables to small circles
around $\sigma_i=0$ (for $1 \le i \le \nedges$).
This amounts to taking the maximal cut.
On the maximal cut, we have
\bq
 e^{\loopnumber \eps \Eulerconstant} 
 \left(\arbitraryscale^2\right)^{-\frac{\loopnumber D}{2}}
 \int\limits_{{\mathcal C}_{\mathrm{maxcut}}}
 \prod\limits_{r=1}^{\loopnumber} \frac{d^Dk_r}{i \pi^{\frac{D}{2}}}
 \frac{\left(\arbitraryscale^2\right)^{\nedges}}{\prod\limits_{j=1}^{\nedges} \sigma_j}
 =
 C''
 \int
 \frac{d^{\NE-\nedges}\Baikovvariable}{\left(2\pi i\right)^{\NE-\nedges}}
 \;
 b_{\mathrm{maxcut}}\left(\Baikovvariable_{\nedges+1},...,\Baikovvariable_{\NE}\right)^{\alpha},
\eq
with
\bq
 b_{\mathrm{maxcut}}\left(\Baikovvariable_{\nedges+1},...,\Baikovvariable_{\NE}\right)
 & = &
 b\left(0,...,0,\Baikovvariable_{\nedges+1},...,\Baikovvariable_{\NE}\right)
\eq
and
\bq
 C'' \; = \; \left(\arbitraryscale^2\right)^{\nedges} C',
 & &
 \alpha \; = \; \frac{D-\loopnumber-{\nexternalindependent}-1}{2}.
\eq
Note that by discussing the maximal cut, we already performed $\nedges$ integrations with the help of the residue theorem.
We introduce dimensionless variables $(z_1,\dots,z_{\NE-\nedges})$
for the remaining Baikov variables, i.e. $z_j=\sigma_{\nedges+j}/\arbitraryscale^2$.

The exponent of the Baikov polynomial $b_{\mathrm{maxcut}}$ is of the form
\bq
\label{def_form_exponent}
 \alpha \; = \; 
 \frac{1}{2} \left( a + b \eps \right),
 & \mbox{with} &
 a,b \; \in \; {\mathbb Z}.
\eq
In the democratic Baikov representation, we always have $b=-2$ for $D=\Dint-2\eps$.

\subsubsection{The loop-by-loop Baikov representation}

Let us now turn to the loop-by-loop Baikov representation \cite{Frellesvig:2017aai}.
We set $\NV$ to be the number of Baikov variables on the maximal cut.
We denote these (dimensionless) Baikov variables by $z=(z_1,\dots,z_{\NV})$.
We have $\NV \le \NE - \nedges$.
On the maximal cut, we have
\bq
\label{loop_by_loop_Baikov_measure}
 e^{\loopnumber \eps \Eulerconstant} 
 \left(\arbitraryscale^2\right)^{-\frac{\loopnumber D}{2}}
 \int\limits_{{\mathcal C}_{\mathrm{maxcut}}}
 \prod\limits_{r=1}^{\loopnumber} \frac{d^Dk_r}{i \pi^{\frac{D}{2}}}
 \frac{\left(\arbitraryscale^2\right)^{\nedges}}{\prod\limits_{j=1}^{\nedges} \sigma_j}
 & = &
 C
 \int
 \frac{d^{\NV}z}{\left(2\pi i\right)^{\NV}} \; \prod\limits_{i \in I_{\mathrm{all}}} \left[ \divisor_i\left(z\right) \right]^{\alpha_i},
\eq
where the $\divisor_i$'s are irreducible polynomials in the $\NV$ remaining Baikov variables.
$I_{\mathrm{all}}$ is a finite index set.

The exponents $\alpha_i$ in eq.~(\ref{loop_by_loop_Baikov_measure})
are of the form as in eq.~(\ref{def_form_exponent}):
\bq
\label{def_form_exponent_loop_by_loop}
 \alpha_i \; = \; 
 \frac{1}{2} \left( a_i + b_i \eps \right),
 & \mbox{with} &
 a_i,b_i \; \in \; {\mathbb Z}.
\eq
We define $I_{\mathrm{odd}}$ as the set of indices for which $a_i$ is odd
and $I_{\mathrm{even}}$ as the set of indices for which $a_i$ is even.
Clearly, we have
\bq
 I_{\mathrm{odd}} \cup I_{\mathrm{even}} \; = \; I_{\mathrm{all}},
 & &
 I_{\mathrm{odd}} \cap I_{\mathrm{even}} \; = \; \emptyset.
\eq
We are in particular interested in the ``minimal'' case, where $a_i \in \{-1,0\}$ for all $i$. 
Let $r(\sigma)$ be an $\eps$-independent function, 
rational in $\sigma$ and of mass dimension zero, leading to a minimal Baikov representation
\bq
\label{Baikov_representation}
\lefteqn{
 e^{\loopnumber \eps \Eulerconstant} 
 \left(\arbitraryscale^2\right)^{-\frac{\loopnumber D}{2}}
 \int\limits_{{\mathcal C}_{\mathrm{maxcut}}} \prod\limits_{r=1}^{\loopnumber} \frac{d^Dk_r}{i \pi^{\frac{D}{2}}} 
 r\left(\sigma\right)
 = 
} & &
 \\
 & &
 \hspace*{30mm}
 \prebaikov
 \int \frac{d^{\NV}z}{\left(2\pi i \right)^{\NV}} \;
 \prod\limits_{i \in I_{\mathrm{odd}}} \left[ \divisor_i\left(z\right) \right]^{-\frac{1}{2} + \frac{1}{2} b_i \eps}
 \prod\limits_{i \in I_{\mathrm{even}}} \left[ \divisor_i\left(z\right) \right]^{\frac{1}{2} b_i \eps}.
 \nonumber
\eq
The prefactor $\prebaikov$ is defined by this equation.

\subsubsection{Good Baikov representations}
\label{sect:good_Baikov_representation}

We call a Baikov representation a ``good'' Baikov representation for a particular sector,
if all singularities are regulated by the twist.
This implies that all exponents $\alpha_i$ in eq.~(\ref{def_form_exponent})
or eq.~(\ref{def_form_exponent_loop_by_loop})
satisfy $\alpha_i \notin {\mathbb Z}$.
This ensures that on the maximal cut, we may always work with twisted cohomology and are not forced to use relative
twisted cohomology.
We note that the democratic Baikov representation is always a good representation; however, not every loop-by-loop representation is good.
\begin{myexample}
A counter-example is given by a two-loop four-point integral.
We consider sector $93$ for the inverse propagators defined by
\begin{align}
 \sigma_1 & = -\left(k_1-p_1\right)^2 +m^2,
 &
 \sigma_2 & = -\left(k_1-p_{12}\right)^2,
 &
 \sigma_3 & = -k_1^2 + m^2,
 \nonumber \\
 \sigma_4 & = -\left(k_1+k_2\right)^2,
 &
 \sigma_5 & = -\left(k_2+p_{12}\right)^2 + m^2,
 &
 \sigma_6 & = -k_2^2,
 \nonumber \\
 \sigma_7 & = -\left(k_2+p_{123}\right)^2 + m^2,
 & 
 \sigma_8 & = -\left(k_1-p_{13}\right)^2,
 &
 \sigma_9 & = -\left(k_2+p_{13}\right)^2.
\end{align}
The external particles are massless $p_1^2=p_2^2=p_3^2=p_4^2=0$.
The Mandelstam variables are denoted by
\begin{align}
 s & = \left(p_1+p_2\right)^2,
 &
 t & = \left(p_2+p_3\right)^2,
\end{align}
and we use the notation $p_{ij}=p_i+p_j$, $p_{ijk}=p_i+p_j+p_k$.
The Feynman graph for sector $93$ is shown in fig.~\ref{fig:sector93} 
and corresponds to the propagator set $S=\{1,3,4,5,7\}$.
\begin{figure}
\begin{center}
\includegraphics[scale=1.0]{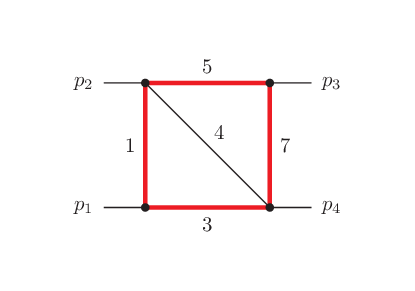}
\end{center}
\caption{
The Feynman graph for sector $93$.
Massive propagators are indicated by red lines.
}
\label{fig:sector93}
\end{figure}
We consider the loop-by-loop Baikov representation, where we take as the inner loop the one with inverse
propagators $\sigma_4, \sigma_5, \sigma_7$.
This will lead to a Baikov polynomial (or a divisor in the mathematical language)
\bq
 m^2-t-\sigma_8,
\eq
which is not regulated by the twist.
In this case, we have to rely on the democratic Baikov representation or an alternative loop-by-loop Baikov representation.
\end{myexample}
We conclude this subsection with a side remark:
If one uses a loop-by-loop Baikov representation, it is not directly evident that 
there is a basis of master integrals for this sector, such
that the integrand of each master integral can be written in this Baikov representation.
In particular one might be worried about Feynman integrals involving irreducible scalar products in the numerator, which do not appear
in the chosen loop-by-loop Baikov representation (``missing ISPs'').
One has to show that any Feynman integral involving a 
product of irreducible scalar products
$\sigma_j$ with $j \in \{\nedges+\NV+1,\dots,\NE\}$ in the numerator can be expressed in terms
of Feynman integrals involving only the Baikov variables $\sigma_j$ with $j \in \{1,\dots,\nedges+\NV\}$.
This is indeed the case, and ref.~\cite{Frellesvig:2024ymq} provides an algorithmic solution for this task.


\subsection{Projective space}
\label{sect:projective_space}

In order to capture possible singularities at infinity, we extend the affine space with coordinates $(z_1,\dots,z_{\NV})$
to projective space ${\mathbb C}{\mathbb P}^{\NV}$ with homogeneous coordinates $[z_0:z_1:\dots:z_{\NV}]$.
Let $d_i$ be the degree of $\divisor_i$ and denote by $\Divisor_i$ the $d_i$-homogenisation
\bq
 \Divisor_i\left(z_0,z_1,\dots,z_{\NV}\right)
 & = & 
 z_0^{d_i}
 \divisor_i\left(\frac{z_1}{z_0},\dots,\frac{z_{\NV}}{z_0}\right).
\eq
We further set $\Divisor_0(z_0,z_1,\dots,z_{\NV})=z_0$
and define the exponent
\bq
\label{def_alpha_0}
 \alpha_0 & = & \frac{1}{2} \left( a_0 + b_0 \eps \right)
\eq
by
\bq
\label{def_a_0_b_0}
 a_0 & = &
 \left\{\begin{array}{ll}
 0 & \mbox{if} \;\; \sum\limits_{i \in I_{\mathrm{odd}}} d_i \;\; \mbox{even}, \\
 -1 & \mbox{if} \;\; \sum\limits_{i \in I_{\mathrm{odd}}} d_i \;\; \mbox{odd}, \\
 \end{array}
 \right.
 \nonumber \\
 b_0 & = & 
 - \sum\limits_{i \in I_{\mathrm{all}}} b_i d_i.
\eq
We extend the definition of a ``good'' Baikov representation to include the requirement $\alpha_0 \neq {\mathbb Z}$
for $\NV>0$.
The case $\NV=0$ is special and discussed in detail in section~\ref{sect:localisation_on_a_point}.
In the case $\NV=0$ we have $\alpha_0=0$, as $I_{\mathrm{all}}$ is the empty set.

We can unify the notation by including the index $0$ in $I_{\mathrm{even}}$ or $I_{\mathrm{odd}}$, 
depending on $a_0$ being zero or $(-1)$, respectively.
We denote the resulting index sets by $I_{\mathrm{even}}^0$, $I_{\mathrm{odd}}^0$ and $I_{\mathrm{all}}^0$.
The definition of $a_0$ ensures that the sum of the degrees of the polynomials in $I_{\mathrm{odd}}^0$ is even.
The definition of $b_0$ follows from homogeneity requirements.
We set 
\bq
  \ND & = & \left| I_{\mathrm{all}} \right|.
\eq
\begin{myexample}
As an example, we consider a two-loop four-point integral from ref.~\cite{Muller:2022gec}.
We consider sector $79$ for the inverse propagators defined by
\begin{align}
 \sigma_1 & = -\left(k_1+p_2\right)^2 + m^2,
 &
 \sigma_2 & = -k_1^2 + m^2,
 &
 \sigma_3 & = -\left(k_1+p_1+p_2\right)^2 + m^2,
 \nonumber \\
 \sigma_4 & = -\left(k_1+k_2\right)^2 + m^2,
 &
 \sigma_5 & = -k_2^2,
 &
 \sigma_6 & = -\left(k_2+p_3+p_4\right)^2,
 \nonumber \\
 \sigma_7 & = -\left(k_2+p_3\right)^2 + m^2,
 &
 \sigma_8 & = -\left(k_1+p_2-p_3\right)^2 + m^2,
 &
 \sigma_9 & = -\left(k_2-p_2+p_3\right)^2.
\end{align}
Sector $79$ corresponds to the propagator set $S=\{1,2,3,4,7\}$.
The Feynman graph for sector $79$ is shown in fig.~\ref{fig:sector79}. 
\begin{figure}
\begin{center}
\includegraphics[scale=1.0]{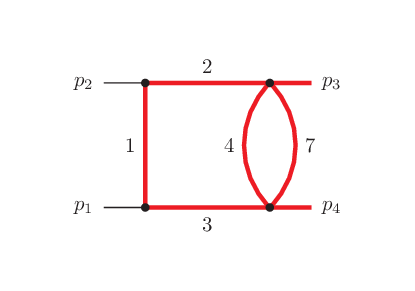}
\end{center}
\caption{
The Feynman graph for sector $79$.
Massive propagators are indicated by red lines.
}
\label{fig:sector79}
\end{figure}
We set $\arbitraryscale=m$, $x_1=s/m^2$ and $x_2=t/m^2$.
On the maximal cut we obtain, using the loop-by-loop approach, 
a minimal one-dimensional Baikov representation from an integrand with a dot on either propagator $4$ or $7$.
Eq.~(\ref{Baikov_representation}) specialises with $D=4-2\eps$ and $z_1=\sigma_8/m^2$ to
\bq
 e^{2 \eps \Eulerconstant} 
 \int\limits_{{\mathcal C}_{\mathrm{maxcut}}} \prod\limits_{r=1}^{2} \frac{d^Dk_r}{i \pi^{\frac{D}{2}}} 
 \frac{\left(m^2\right)^{2+2\eps}}{\sigma_1 \sigma_2 \sigma_3 \sigma_4^2 \sigma_7}
 & = &
 \prebaikov
 \int \frac{dz_1}{2\pi i} \;
 \left[ \divisor_1\left(z\right) \right]^{-\frac{1}{2}}
 \left[ \divisor_2\left(z\right) \right]^{-\frac{1}{2}-\eps}
 \left[ \divisor_3\left(z\right) \right]^{-\frac{1}{2}-\eps}.
 \;\;
 \nonumber
\eq
The prefactor is given by
\bq
\label{def_prebaikov_example_2}
 \prebaikov
 =  
 \frac{2^{4+4\eps} \pi^4 e^{2 \eps \Eulerconstant}}{\left[\Gamma\left(\frac{1}{2}-\eps\right)\right]^2 x_1^{1+\eps} }
 \left[\left(1-x_2\right)^2+x_1x_2\right]^\eps.
\eq
The Baikov polynomials read
\begin{align}
 p_1 & = z_1-x_2,
 &
 p_2 & = z_1+ 4-x_2,
 &
 p_3 & = \left( z_1 + 1 \right)^2 - 4 \left[ x_2 +\frac{\left(1-x_2\right)^2}{x_1} \right]. 
\end{align}
We have $I_{\mathrm{even}}=\emptyset$ and 
$I_{\mathrm{odd}}=\{1,2,3\}$.
The equation 
\bq
 y^2 = p_1\left(z_1\right) p_2\left(z_1\right) p_3\left(z_1\right)
\eq
defines an elliptic curve, where the four branch points are at finite distance.
We expect two master integrals associated with this elliptic curve. The sector $79$ has three master integrals.
To see the third master integral, we need to go to projective space and include the point at infinity.
The homogenisations of $p_1$, $p_2$ and $p_3$ are 
\begin{align}
 P_1 & = z_1-x_2z_0, 
 &
 P_2 & = z_1+(4-x_2)z_0,
 &
 P_3 & = \left( z_1 + z_0 \right)^2 - 4 \left[ x_2 +\frac{\left(1-x_2\right)^2}{x_1} \right] z_0^2.
\end{align}
We further introduce $P_0=z_0$ with $a_0=0$ and $b_0=6\eps$, according to eq.~(\ref{def_a_0_b_0}).
We have $I^0_{\mathrm{even}}=\{0\}$ and $I^0_{\mathrm{odd}}=\{1,2,3\}$.
The third master integral is associated with the residue at $z_0=0$.
\end{myexample}


\subsection{Twisted cohomology}
\label{sect:twisted_cohomology}

We define the minimal twist $U$ as the homogenisation of the twist in eq.~(\ref{Baikov_representation})
by
\bq
\label{def_twist}
 U\left(z_0,z_1,\dots,z_{\NV}\right)
 & = &
 \prod\limits_{i \in I_{\mathrm{odd}}^0} \Divisor_i^{-\frac{1}{2}+\frac{1}{2} b_i \eps}
 \prod\limits_{j \in I_{\mathrm{even}}^0} \Divisor_j^{\frac{1}{2} b_j \eps}.
\eq
One can see that $U(z)$ is homogeneous of degree
\bq
 d_U & = & - \frac{1}{2} \sum\limits_{i \in I_{\mathrm{odd}}^0} d_i.
\eq
The quantity $b_0$ has been defined in eq.~(\ref{def_a_0_b_0}) such that $d_U$ is independent of $\eps$. 

Each polynomial in the twist defines a hypersurface
\bq
 D_i \; = \; \{ \left[z_0:z_1:\dots:z_n\right] \in \mathbb{C} \mathbb{P}^{\NV}\,\, | \,\,P_i\left(z_0,z_1,\dots,z_n\right) = 0 \},
\eq
and we set $D$ to be the union of these hypersurfaces:
\bq
\label{def_divisior_D}
 D \; = \; 
 \bigcup\limits_{i \in I^0_{\mathrm{all}}} D_i.
\eq
It is customary to set
\bq 
 \omega & = & d_F \ln U\left(z_0,z_1,\dots,z_{\NV}\right),
\eq
where $d_F$ is defined in eq.~(\ref{def_d_F}) and denotes the differential
with respect to the Baikov variables $[z_0:z_1:\dots:z_\NV]$.
We further introduce the covariant derivative
\bq
\label{def_nabla_F}
 \nabla_F & = & d_F + \omega.
\eq
As we are only interested in differential forms, which depend on the holomorphic coordinates $[z_0:\dots:z_{\NV}]$, but
not on the anti-holomorphic coordinates $[\bar{z}_0:\dots:\bar{z}_{\NV}]$, we may restrict to the holomorphic derivatives 
appearing in $d_F$ and ignore the anti-holomorphic derivatives.
We are interested in the twisted cohomology group
\bq
 H^\NV_\omega,
\eq
which is defined as the set of equivalence classes of the $\nabla_F$-closed differential $(\NV,0)$-forms holomorphic
on ${\mathbb C}{\mathbb P}^\NV-D$ modulo the $\nabla_F$-exact ones.

For later purposes we also define the $(\eps=0)$-part $U_0$ of the minimal twist by
\bq
 U_0\left(z_0,z_1,\dots,z_{\NV}\right)
 & = &
 \prod\limits_{i \in I_{\mathrm{odd}}^0} \Divisor_i^{-\frac{1}{2}}.
\eq


\section{The set-up}
\label{sect:set_up}

In this section, we introduce the framework for our approach. 
Our aim is to order the Feynman integrals within a sector by their geometric properties.
We therefore focus on the maximal cut\footnote{Subtleties related to the Feynman integrals within one sector, Feynman integrals on the maximal cut and magic relations will be discussed in a separate forthcoming publication.}.
We start with the definition of the objects, which we will be studying.
Instead of integrals, we study integrands, and more precisely differential $\NV$-forms on projective space ${\mathbb C} {\mathbb P}^{\NV}$.
These will be defined in section~\ref{sect:objects}.
In section~\ref{sect:vector_spaces} we discuss the relation between the vector space of integrals and the vector space of integrands.
We are in particular interested in linear relations between the integrands.
These are discussed in section~\ref{sect:linear_relations}.
For the order relation, we count poles and residues. These concepts are introduced in section~\ref{sect:poles_and_residues}.
In section~\ref{sect:filtrations}, we define three filtrations, based on the number of residues, the pole order and the sum of the indices.
The algorithm proceeds recursively by considering all possible localisations (i.e. places where we can take a residue).
Polynomials of higher degree might introduce algebraic extensions. In section~\ref{sect:localisations} we show how this can be avoided.
We present a method which works entirely with rational functions.
In section~\ref{sect:geometry} , we discuss the specific geometries associated with the maximal cut of Feynman integrals.


\subsection{Objects}
\label{sect:objects}

For our purposes, it is slightly more convenient to include the twist in the differential forms, as the twist
affects the definition of a prefactor.
The central objects of our study are differential forms, which can be written as
\bq
\label{def_input_data}
 \differentialform_{\mu_0 \dots \mu_{\ND}}\left[Q\right]
 & = &
 \prebaikov \;
 \preall \;
 U\left(z\right) \hat{\Phi}_{\mu_0 \dots \mu_{\ND}}\left[Q\right]
 \eta.
\eq
The integrands in eq.~(\ref{def_input_data}) are all what we need to study. Up to prefactors, any integrand can be written in this form.
We explain the ingredients right-to-left.
$\eta$ is the standard $\NV$-form on ${\mathbb C}{\mathbb P}^{\NV}$ defined by
\bq
\label{def_eta}
 \eta
 & = &
 \sum\limits_{j=0}^{\NV} (-1)^{j} \; z_j \; dz_0 \wedge ... \wedge \widehat{dz_j} \wedge ... \wedge dz_{\NV},
\eq
where the hat indicates that the corresponding term is omitted.
$\hat{\Phi}_{\mu_0 \dots \mu_{\ND}}[Q]$ is an $\eps$-independent meromorphic function in $z$, 
holomorphic on ${\mathbb C}{\mathbb P}^\NV-D$
and given by 
\bq
 \hat{\Phi}_{\mu_0 \dots \mu_{\ND}}\left[Q\right]
 & = &
 \frac{Q}{\prod\limits_{i \in I_{\mathrm{all}}^0} \Divisor_i^{\mu_i}},
\eq
where $\mu_j \in {\mathbb N}_0$ and $Q$ is a homogeneous polynomial of degree
\bq
\label{def_d_Q}
 d_Q & = &
 \sum\limits_{i \in I_{\mathrm{all}}^0} \mu_i d_i - d_U - \NV - 1.
\eq
We set
\bq
\label{def_mu}
 \left| \mu \right|
 & =
 \sum\limits_{i \in I_{\mathrm{all}}^0} \mu_i.
\eq
The twist function $U(z)$ has been defined in eq.~(\ref{def_twist}).

The $\eps$-dependent prefactor $\preall$ is independent of $z$ (but may depend on $x$) and 
given as a product 
\bq
 \preall
 & = &
 \preabs \cdot \prerel \cdot \preclutch.
\eq
The prefactor $\prebaikov$ has been defined in eq.~(\ref{Baikov_representation}).
The factor $\preabs$ is defined such that~$\preabs \cdot \prebaikov$ is pure of transcendental weight zero.

The relative prefactor $\prerel$ depends on $(\mu_0,\mu_1,\dots)$ and is given by
\bq
 \prerel
 =
 \prod\limits_{i \in I_{\mathrm{all}}^0} 
 \left( \alpha_i \right)_{\mu_i},
\eq
with $(a)_n=\Gamma(a+1)/\Gamma(a+1-n)$ being the falling factorial.
The $\alpha_i$'s are the exponents appearing in eq.~(\ref{def_form_exponent_loop_by_loop}) and eq.~(\ref{def_alpha_0}).
For the minimal twist function $U$ defined by eq.~(\ref{def_twist}) the relative prefactor reduces to
\bq
 \prerel
 =
 \prod\limits_{i \in I_{\mathrm{odd}}^0} 
 \left(-\frac{1}{2}+\frac{1}{2} b_i \eps\right)_{\mu_i}
 \;\;
 \prod\limits_{i \in I_{\mathrm{even}}^0} 
 \left(\frac{1}{2} b_i \eps \right)_{\mu_i}.
\eq
The inclusion of the relative prefactor is motivated by the structure of the integration-by-parts identities
within twisted cohomology.
The relative prefactor is chosen such that it trivialises the $\eps$-dependence of the integration-by-parts identities.
The integration-by-parts identities within twisted cohomology will be derived in section~\ref{sect:ibp}.

The clutch prefactor $\preclutch$ is given by
\bq
 \preclutch
 & = &
 \eps^{-\absmu},
\eq
where $\absmu$ is defined in eq.~(\ref{def_mu}).
The clutch factor follows from known examples involving non-trivial geometries, where an 
$\eps$-factorised differential equation has been constructed.
In the case of the equal-mass $l$-loop banana integrals \cite{Adams:2018yfj,Pogel:2022vat},
the first master integral starts at order $\eps^l$ and this factor is contributed by $\preabs$. The following
master integrals start each one power of $\eps$ earlier compared to the previous one.
These powers of $\eps$ are provided by $\preclutch$.

\begin{myexample}
We continue with example 2:
The prefactor $\prebaikov$ was given in eq.~(\ref{def_prebaikov_example_2}) by
\bq
 \prebaikov
 & = & 
 \frac{2^{4+4\eps} \pi^4 e^{2 \eps \Eulerconstant}}{\left[\Gamma\left(\frac{1}{2}-\eps\right)\right]^2 x_1^{1+\eps} }
 \left[\left(1-x_2\right)^2+x_1x_2\right]^\eps.
\eq
It is easy to check that with $\preabs=\eps^3 x_1$
the product $\preabs \prebaikov$ is pure of transcendental weight zero:
\bq
 \preabs \cdot \prebaikov
 & = &
 16 \pi^3 \eps^3
 - 16 \pi^3 \left( L_1 - L_2 \right) \eps^4
 + 8 \pi^3 \left[ \left(L_1-L_2\right)^2 - \pi^2 \right] \eps^5
 + {\mathcal O}\left(\eps^6\right),
\eq
where $L_1=\ln(x_1)$ and $L_2=\ln((1-x_2)^2+x_1x_2)$.
\end{myexample}


\subsection{Vector spaces}
\label{sect:vector_spaces}

We denote the vector space spanned by the Feynman integrals on the maximal cut (viewed as abstract symbols)
by $A^{\NV}$.
This is an infinite-dimensional vector space of countable dimension.
We denote the vector space spanned by the differential forms of eq.~(\ref{def_input_data}) by $\Agen^{\NV}_\omega$.
This is again an infinite-dimensional vector space of countable dimension.
For a good Baikov representation we denote the linear map, which associates to each Feynman integral its integrand on the maximal cut by
\bq
\label{def_iota_A_to_Omega}
 \hat{\iota} & : & A^{\NV} \rightarrow \Agen^{\NV}_\omega.
\eq
We are interested in finite-dimensional vector spaces.
We denote the vector space of Feynman integrals 
on the maximal cut modulo linear relations by $V^{\NV}$
and the twisted cohomology group related to the integrands in eq.~(\ref{def_input_data})
by $\Hgen^{\NV}_{\omega}$.
Both are finite-dimensional vector spaces.
Their dimensions can be computed with the Laporta algorithm 
and this will be the method used throughout this paper. 
In addition, we will sometimes also report the number based on the counting of critical points \cite{Lee:2013hzt}. 

There is an injective linear map
\bq
\label{def_iota}
 \iota & : & V^{\NV} \rightarrowtail \Hgen^{\NV}_{\omega},
\eq
such that the following diagram is commutative
\bq
\label{commutative_diagram}
\begin{CD}
 A^{\NV} @>{\hat{\iota}}>> \Agen^{\NV}_{\omega} \\
@A{\sigma}AA @VV{\pi_{\omega}}V \\
 V^{\NV} @>{\iota}>> \Hgen^{\NV}_\omega \\
\end{CD},
\eq
where the map $\sigma : V^{\NV} \rightarrow A^{\NV}$ is trivial, once a basis of $V^{\NV}$ has been chosen 
(it sends each master integral to itself).
The map $\pi_{\omega} : \Agen^{\NV}_{\omega} \rightarrow \Hgen^{\NV}_\omega$
sends each differential form to its cohomology class.
In other words, $\iota$ is defined by $\iota = \pi_\omega \circ \hat{\iota} \circ \sigma$.
The map $\iota$ is injective: If the integrand of the Baikov representation is zero in $H^{\NV}_{\omega}$, 
then the maximal cut integral is zero in $V^{\NV}$.
Hence, we have
\bq
 \dim V^{\NV} & \le & \dim H^{\NV}_{\omega}.
\eq
The map is surjective, if $\dim V^{\NV} = \dim H^{\NV}_{\omega}$.
In general, the map will not be surjective.
There are two reasons for this: symmetries and super-sectors.

\subsubsection{Symmetries}

Integration can lead to symmetries among Feynman integrals (elements in $V^n$), which are not symmetries of the integrands (elements in $\Hgen^n_{\omega}$).
As a very simple example, we have the equality of the integrals
\bq
 \int\limits_{[0,1]^2} z_1 dz_1 \wedge dz_2 & = & \int\limits_{[0,1]^2} z_2 dz_1 \wedge dz_2,
\eq
however, this does not imply the equality of the integrands
\bq
 z_1 dz_1 \wedge dz_2 & \neq & z_2 dz_1 \wedge dz_2.
\eq
An example from the context of Feynman integrals is the following:
\begin{myexample}
We consider the equal-mass sunrise integral in $D=2-2\eps$ dimensions.
The inverse propagators are defined by
\begin{align}
 \sigma_1 & = -k_1^2 +m^2,
 &
 \sigma_2 & = -k_2^2 +m^2,
 &
 \sigma_3 & = -\left(k_1+k_2-p\right)^2 + m^2,
 \nonumber \\
 \sigma_4 & = -\left(k_1-p\right)^2,
 &
 \sigma_5 & = -\left(k_1+k_2\right)^2.
\end{align}
The sunrise integral has sector id $7$, corresponding to the propagator set $S=\{1,2,3\}$.
We consider the loop-by-loop Baikov representation, where we take as the inner loop the one with inverse
propagators $\sigma_2$ and $\sigma_3$. 
This gives us a one-dimensional Baikov representation with $\sigma_4$ as the remaining Baikov variable.
With $z_1=\sigma_4/(-p^2)$ and $x=m^2/(-p^2)$ the twist function reads
\bq
\label{def_twist_equal_mass_sunrise}
 U\left(z\right)
 & = & 
 z_0^{3 \eps}
 z_1^{-\frac{1}{2}} 
 \left\{ \left(z_1+4xz_0\right)
         \left[ z_1^2 + 2 \left(x-1\right) z_0 z_1 + \left(x+1\right)^2 z_0^2 \right] \right\}^{-\frac{1}{2}-\eps}.
\eq
On the Feynman integral side, we have two master integrals for sector $7$, hence $\dim V^1 = 2$.
However, in twisted cohomology, we find for the twist function of eq.~(\ref{def_twist_equal_mass_sunrise}) that
$\dim H_\omega^1 = 3$.
This is an example where the map from $V^1$ to $H^1_\omega$ is not surjective.
We have two symmetry relations on the Feynman integral side, which can be taken as
\bq
 I_{12100} - I_{11200} & = & 0,
 \nonumber \\
 2 I_{21100} - I_{12100} - I_{11200} & = & 0.
\eq
Our loop-by-loop Baikov representation is shown pictorially in fig.~\ref{fig:sector15}
and treats the inverse propagators $\sigma_2$ and $\sigma_3$ symmetrically,
hence, the left-hand side of the first symmetry relation maps to zero in twisted cohomology.
\begin{figure}
\begin{center}
\includegraphics[scale=1.0]{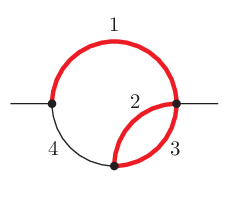}
\end{center}
\caption{
The Feynman graph for sector $15$.
Massive propagators are indicated by red lines.
}
\label{fig:sector15}
\end{figure}
However, the left-hand side of the second symmetry relation does not map to zero in twisted cohomology.
Moding out the image 
\bq
 \hat{\iota}\left( 2 I_{21100} - I_{12100} - I_{11200} \right)
\eq
in $H^1_\omega$ gives a two-dimensional space.
This corresponds to the two master integrals in the top sector.
\end{myexample}

\subsubsection{Super-sectors}

A sector $S_1$ is called a super-sector of the sector $S_2$, if $S_2$ can be obtained from $S_1$ 
by contracting some of the edges of $S_1$.
 
In order to motivate the relevance of super-sectors for the twisted cohomology gropu $\Hgen^\NV_\omega$, consider the situation where
a polynomial $\divisor_j(z)$ with $j \in I_{\mathrm{even}}$ is simply a factor $z_r=\Baikovvariable_l/\arbitraryscale^2$,
where $\Baikovvariable_l$ is an uncut inverse propagator.
In this case, $\Hgen^{\NV}_{\omega}$ will also contain the integrands of the sector 
where the exponent of this inverse propagator is positive.
If this sector has additional master integrals, they will also appear in $\Hgen^{\NV}_{\omega}$.
Such a sector is a super-sector, and we include it in the analysis.
\begin{myexample}
An example is given by the three-loop unequal-mass banana integral \cite{Pogel:2025bca}.
The inverse propagators are given by
\begin{align}
 \sigma_1 & = -k_1^2 + m_1^2,
 &
 \sigma_2 & = -k_2^2 + m_2^2,
 &
 \sigma_3 & = -k_3^2 + m_3^2,
 \nonumber \\
 \sigma_4 & = -\left(k_1+k_2+k_3-p\right)^2 + m_4^2,
 &
 \sigma_5 & = -\left(k_1+k_2-p\right)^2,
 &
 \sigma_6 & = -\left(k_1-p\right)^2,
 \nonumber \\
 \sigma_7 & = -\left(k_1+k_2\right)^2,
 & 
 \sigma_8 & = -\left(k_1+k_3\right)^2,
 &
 \sigma_9 & = -\left(k_2+k_3\right)^2.
\end{align}
and we consider the loop-by-loop Baikov representation, 
where the innermost loop is formed by the inverse propagators $\sigma_3$ and $\sigma_4$,
followed by the loop formed by the inverse propagators $\sigma_2$ and $\sigma_5$,
and finally the loop formed by the inverse propagators $\sigma_1$ and $\sigma_6$.
The three-loop banana integral has sector id $15$, corresponding to the propagator set $S=\{1,2,3,4\}$.
With $z_1=\sigma_5/(-p^2)$ and $z_2=\sigma_6/(-p^2)$ the twist function is given by
\bq
 U
 & = &
 z_0^{4\eps}
 z_1^{\eps}
 z_2^{\eps}
 P_3^{-\frac{1}{2}-\eps}
 P_4^{-\frac{1}{2}-\eps}
 P_5^{-\frac{1}{2}-\eps}.
\eq
The definition of the polynomials $P_3-P_5$ is not relevant here (the definition can be found in ref.~\cite{Pogel:2025bca}),
what is relevant is the fact that $z_1$ and $z_2$ appear as even polynomials in the twist function.
Hence, we have to include the super-sectors $31$ (with propagator set $S=\{1,2,3,4,5\}$), 
$47$ (with propagator set $S=\{1,2,3,4,6\}$)
and $63$ (with propagator set $S=\{1,2,3,4,5,6\}$).
These are shown in fig.~\ref{fig:super_sectors}.
\begin{figure}
\begin{center}
\includegraphics[scale=1.0]{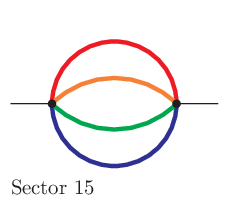}
\includegraphics[scale=1.0]{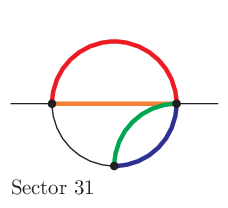}
\includegraphics[scale=1.0]{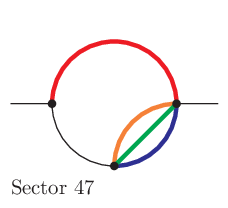}
\includegraphics[scale=1.0]{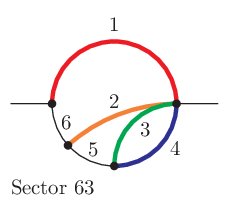}
\end{center}
\caption{
The sector of the banana integral together with the relevant super-sectors.
}
\label{fig:super_sectors}
\end{figure}
The super-sectors $31$ and $47$ have one master integral each. A possible basis for these super-sectors on the Feynman integral side is
\bq
 I_{111110000},
 \;\;\;
 I_{111101000}.
\eq
The super-sector $63$ is reducible and does not have any master integral.
These two additional master integrals make up for the difference in dimensions:
We have 
\bq
 \dim V^2 \; = \; 11,
 & &
 \dim H_\omega^2 \; = \; 13.
\eq
\end{myexample}
Our starting point is always a graph $G$, which has the property that any scalar product involving a loop
momenta can be expressed as a linear combination of inverse propagators (see section~\ref{sect:feynman_integrals}).
This does not imply that all relevant super-sectors are sub-graphs of the graph $G$.
We provide a counterexample.
\begin{myexample}
We start from a non-planar double box integral with inverse propagators
\begin{align}
\label{inverse_propagators_moeller}
 \sigma_1 & = -\left(k_1-p_1\right)^2 +m^2,
 &
 \sigma_2 & = -\left(k_1-p_{12}\right)^2,
 &
 \sigma_3 & = -k_1^2,
 \nonumber \\
 \sigma_4 & = -k_{12}^2 + m^2,
 &
 \sigma_5 & = -\left(k_{12}+p_3\right)^2,
 &
 \sigma_6 & = -k_2^2,
 \nonumber \\
 \sigma_7 & = -\left(k_2+p_{123}\right)^2 + m^2,
 & 
 \sigma_8 & = -\left(k_1-p_{13}\right)^2,
 &
 \sigma_9 & = -\left(k_2+p_{13}\right)^2.
\end{align}
The auxiliary graph $G$ with nine propagators has sector id $511$,
the non-planar double box integral has sector id $127$.
We are interested in sector $123$ with propagator set $S=\{1,2,4,5,6,7\}$.
The sectors are shown in fig.~\ref{fig:moeller}.
\begin{figure}
\begin{center}
\includegraphics[scale=1.0]{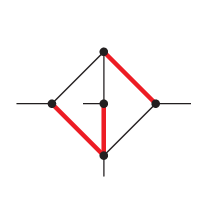}
\includegraphics[scale=1.0]{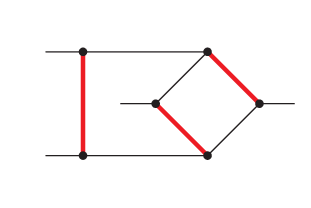}
\includegraphics[scale=1.0]{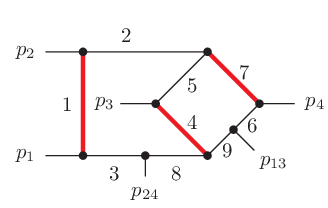}
\end{center}
\caption{
Sector $123$ (left), sector $127$ (middle) and sector $511$ (right).
}
\label{fig:moeller}
\end{figure}
For sector $123$ we 
consider the loop-by-loop Baikov representation, where the inner loop is given by $\sigma_4,\sigma_5,\sigma_6,\sigma_7$.
We obtain a two-dimensional Baikov representation with twist
\bq
 U
 & = &
 z_0^{4\eps} P_1^\eps P_2^\eps P_3^{-\frac{1}{2}-\eps} P_4^{-\frac{1}{2}-\eps},
\eq
with
\bq
 P_1 \; = \; \frac{\sigma_8+ \left(m^2-t\right)z_0}{s},
 & &
 P_2 \; = \; \frac{\sigma_8-\sigma_3+ \left(m^2-t-s\right)z_0}{s}.
\eq
The explicit expressions for $P_3$ and $P_4$ are not relevant here.
$P_1$ and $P_2$ are not proportional to any of the inverse propagators appearing in eq.~(\ref{inverse_propagators_moeller}).
 We have
\bq
\label{diff_dim_sector_123}
 \dim V^2 \; = \; 8,
 & &
 \dim H_\omega^2 \; = \; 12.
\eq
The difference in dimensions is due to super-sectors, where $P_1$ or $P_2$ appear in the denominator.
These super-sectors are not sub-sectors of sector $511$ of the graph $G$.
We may however construct a new auxiliary graph $\tilde{G}$, with inverse propagators
\begin{align}
\label{inverse_propagators_moeller_alt}
 \sigma_1 & = -\left(k_1-p_1\right)^2 +m^2,
 &
 \sigma_2 & = -\left(k_1-p_{12}\right)^2,
 &
 \tilde{\sigma}_3 & = -\left(k_1-p_{123}\right)^2,
 \nonumber \\
 \sigma_4 & = -k_{12}^2 + m^2,
 &
 \sigma_5 & = -\left(k_{12}+p_3\right)^2,
 &
 \sigma_6 & = -k_2^2,
 \nonumber \\
 \sigma_7 & = -\left(k_2+p_{123}\right)^2 + m^2,
 & 
 \tilde{\sigma}_8 & = -\left(k_1-2p_{12}-p_3\right)^2,
 &
 \sigma_9 & = -\left(k_2+p_{13}\right)^2,
\end{align}
such that sector $123$ and the super-sectors are sub-graphs of $\tilde{G}$.
The super-sectors $127$ and $251$ with respect to the inverse propagators defined in eq.~(\ref{inverse_propagators_moeller_alt}) contribute two master integrals each, explaining the difference in dimensions
in eq.~(\ref{diff_dim_sector_123}).
\end{myexample}

\subsubsection{Conversion between $V^{\NV}$ and $H^{\NV}_{\omega}$}
\label{sect:conversion}

As already mentioned at the beginning of section~\ref{sect:vector_spaces}, the map $\iota : V^{\NV} \rightarrowtail H^{\NV}_{\omega}$ goes from the vector space of Feynman integrals on the maximal cut modulo linear relations to the twisted cohomology classes.
We write each master integral on the maximal cut in the Baikov representation and take the integrand in this representation as a representative of a twisted cohomology class.
However, this assumes that we already have a basis of master integrals on $V^{\NV}$, which is required to define the map
$\sigma : V^{\NV} \rightarrow A^{\NV}$ in eq.~(\ref{commutative_diagram}).

What we would like to have is a map $j : H^{\NV}_{\omega} \twoheadrightarrow V^{\NV}$ and use this map to define
master integrals on $V^{\NV}$ from the master integrands on $H^{\NV}_{\omega}$.
This can be done as follows:
We denote by $A^{\NV}_{\mathrm{all}}$ the infinite-dimensional vector space of Feynman integrals on the maximal cut including the super-sectors,
and by $A^{\NV}_{\mathrm{super}}$ the infinite-dimensional vector space of Feynman integrals of the super-sectors only.
We denote by $V^{\NV}_{\mathrm{all}}$ and $V^{\NV}_{\mathrm{super}}$ the corresponding finite-dimensional vector spaces,
obtained by moding out linear relations.
We have
\bq
 V^{\NV}_{\mathrm{all}} & = & V^{\NV} + V^{\NV}_{\mathrm{super}}.
\eq
The map $\hat{\iota}$ extends trivially to $\hat{\iota} : A^{\NV}_{\mathrm{all}} \rightarrow \Agen^{\NV}_{\omega}$.
Note that the map $\hat{\iota}$ is not necessarily injective, due to the ``missing ISPs'' discussed at the end
of section \ref{sect:good_Baikov_representation}.
There is a surjective map
\bq
 j & : & H^{\NV}_{\omega} \twoheadrightarrow V^{\NV},
\eq
which is obtained from the following sequence of maps
\bq
 H^{\NV}_{\omega}
 \stackrel{\pi_{\mathrm{integration}}}{\longrightarrow}
 V^{\NV}_{\mathrm{all}}
 \stackrel{\pi_{\mathrm{sector}}}{\longrightarrow}
 V^{\NV}.
\eq
The first map projects out symmetry relations due to integration,
the second map subtracts out contributions from super-sectors.
This allows us to define master integrals on $V^{\NV}$ and to decompose $H^{\NV}_{\omega}$ as
\bq
 H^{\NV}_{\omega}
 & = &
 \mathrm{Im}\left(\iota\right) + \mathrm{Ker}\left(j\right)
\eq
In detail, we proceed as follows: 
A symmetry relation is a linear relation among Feynman integrals, i.e.
\bq
 \sum\limits_{i} c_i I_i & = & 0,
\eq
A non-trivial symmetry relation defines a linear relation
\bq
\label{symmetry_relation_H_twist}
 \hat{\iota}\left(\sum\limits_{i} c_i I_i\right) & = & 0
\eq
in $\Agen^{\NV}_{\omega}$. 
We denote by $\Hgen^{\NV}_{\mathrm{symm}}$ the vector space $\Hgen^{\NV}_{\omega}$ modulo these relations.
The cohomology classes defined by the left-hand side of eq.~(\ref{symmetry_relation_H_twist}) span the kernel of $\pi_{\mathrm{integration}}$.
A basis element of $\Hgen^{\NV}_{\mathrm{symm}}$ is represented by $\differentialform \in \Agen^{\NV}_{\omega}$.
This differential form has one or more pre-images in $A^{\NV}_{\mathrm{all}}$, which can be found by linear algebra.
We define an order relation on the space of Feynman integrals and hence on
$A^{\NV}_{\mathrm{all}}$.
We select the simplest pre-image according to this order relation.
As a side-remark we point out that this order relation will appear 
again in eq.~(\ref{full_order_relation}) as the lowest order criterion.
Note that the order relation on $A^{\NV}_{\mathrm{all}}$ is only used to select among a set of possible pre-images 
the simplest pre-image. All possible pre-images are equivalent, once integration-by-parts identities 
and symmetry relations have been taken into account.
Note further, that we do not need to know all pre-images, we are only interested in the simplest one.
In practice, this involves therefore only a rather small linear system of equations.
This procedure defines a basis on $V^{\NV}_{\mathrm{all}}$.
We may now define the map $\sigma : V^{\NV}_{\mathrm{all}} \rightarrow A^{\NV}_{\mathrm{all}}$
and in turn $\iota = \pi_\omega \circ \hat{\iota} \circ \sigma$.
With a slight abuse of notation we use $\sigma$ to denote the maps $\sigma : V^{\NV} \rightarrow A^{\NV}$ and
$\sigma : V^{\NV}_{\mathrm{all}} \rightarrow A^{\NV}_{\mathrm{all}}$.
In the same spirit, we use $\iota$ to denote the maps
$V^{\NV} \rightarrow \Hgen^{\NV}_\omega$ and $V^{\NV}_{\mathrm{all}} \rightarrow \Hgen^{\NV}_\omega$.

Among the basis elements of $V^{\NV}_{\mathrm{all}}$ we have the ones which span $V^{\NV}_{\mathrm{super}}$.
Let us assume that these are labelled $I_{\NF+1},I_{\NF+2},\dots$.
The kernel of $\pi_{\mathrm{sector}}$ is given by these integrals.
The pre-images of $I_{\NF+1},I_{\NF+2},\dots$ in $H^{\NV}_{\omega}$ are given by $\iota(I_{\NF+1}),\iota(I_{\NF+2}),\dots$.

We may decompose $H^{\NV}_{\omega}$ into $\mathrm{Im}(\iota)$, the kernel of $\pi_{\mathrm{integration}}$ and the pre-image of the kernel of $\pi_{\mathrm{sector}}$.
The latter two span the kernel of $j$.
There is an isomorphism between 
$V^{\NV}$ and $\mathrm{Im}(\iota)$.


\subsection{Linear relations}
\label{sect:linear_relations}

For the integrands defined in eq.~(\ref{def_input_data}) we have three types of linear relations:
Integration-by-parts identities, distribution identities and cancellation identities.

\subsubsection{Integration-by-parts identities}
\label{sect:ibp}

The integration-by-parts identities read
\bq
\label{eq_ibp}
 \frac{1}{\eps}
 \differentialform_{\mu_0 \dots \mu_i \dots \mu_{\ND}}\left[\partial_{z_j} Q_+\right]
 +
 \sum\limits_{i \in I_{\mathrm{all}}^0} 
 \differentialform_{\mu_0 \dots (\mu_i+1) \dots \mu_{\ND}}\left[Q_+ \cdot \left( \partial_{z_j} P_i \right) \right]
 & = & 0,
\eq
where $Q_+$ is an $\eps$-independent homogeneous polynomial of degree $\deg Q_+ = d_Q+1$.
The integration-by-parts identities are derived as follows:
We start from the fact that within twisted cohomology, we may always add a $\nabla_F$-exact form to a twisted cohomology class:
\bq
 \nabla_F \Xi & = & 0,
\eq
where $\Xi$ is a $(\NV-1)$-form, holomorphic on ${\mathbb C} {\mathbb P}^{\NV}-D$.
We call $\Xi$ the seed.
$\Xi$ may have singularities on $D$.
Following Griffiths \cite{Griffiths:1969} we may write $\Xi$ as
\bq
\label{seed}
 \Xi & = & 
 \frac{Q_+}{\prod\limits_{i \in I_{\mathrm{all}}^0} \Divisor_i^{\mu_i}}
 \left( 
       \sum\limits_{k<j} \left(-1\right)^{k} z_k
                        dz_0 \wedge ... \wedge \widehat{dz_k} \wedge ... \wedge \widehat{dz_j} \wedge ... \wedge dz_{\NV}
 \right. \nonumber \\
 & & \left.
       -
       \sum\limits_{k>j} \left(-1\right)^{k} z_k
                        dz_0 \wedge ... \wedge \widehat{dz_j} \wedge ... \wedge \widehat{dz_k} \wedge ... \wedge dz_{\NV}
\right),
\eq
with
\bq
 \mu_i \in {\mathbb N}_0,
 && i \in \{0,1,\dots,\ND \},
\eq
subject to the homogeneity condition
\bq
\label{constraint_seeds}
 \deg Q_+
 & = &
 \sum\limits_{i \in I_{\mathrm{all}}^0} \mu_i d_i
 - d_U - \NV.
\eq
Note that the degree of $Q_+$ is one higher than the degree of $Q$ given in eq.~(\ref{def_d_Q})
(hence the notation ``$Q_+$'').
The possible seeds are specified by the $(\ND+2)$-tuple
$(\mu_0,\mu_1,\dots,\mu_{\ND},j)$ and the polynomial $Q_+$,
subject to the constraint given in eq.~(\ref{constraint_seeds}).
We have
\bq
 \nabla_F \Xi & = & 
 \eta \nabla_{F,z_j} \left( \frac{Q_+}{\prod\limits_{i \in I_{\mathrm{all}}^0} \Divisor_i^{\mu_i}}
 \right)
 \; = \; 0,
\eq
where $\nabla_{F,z_j}$ is defined by
\bq
 \nabla_F & = & \sum\limits_{j=0}^{\NV} dz_j \nabla_{F,z_j}.
\eq
Thus if we strip off $\eta$, the integration-by-parts identities are
\bq
 \nabla_{F,z_j} \left( \frac{Q_+}{\prod\limits_{i \in I_{\mathrm{all}}^0} \Divisor_i^{\mu_i}}
 \right)
 & = & 0.
\eq
In this form, each integration-by-parts identity is linear in $\eps$, where the part proportional to $\eps$ comes from the twist.
Including the twist and the prefactors, we have
\bq
 \partial_{z_j} \left( \prebaikov \preall U \frac{Q_+}{\prod\limits_{i \in I_{\mathrm{all}}^0} \Divisor_i^{\mu_i}}
 \right)
 & = & 0.
\eq
Working this out yields eq.~(\ref{eq_ibp}).

\subsubsection{Distribution identities}
\label{sect:distribution}

The distribution identities are rather trivial and originate from writing a polynomial $Q=Q_1+Q_2$ as a sum of 
two other polynomials:
\bq
\label{eq_distribution}
 \differentialform_{\mu_0 \dots \mu_{\ND}}\left[Q\right]
 & = &
 \differentialform_{\mu_0 \dots \mu_{\ND}}\left[Q_1\right]
 +
 \differentialform_{\mu_0 \dots \mu_{\ND}}\left[Q_2\right].
 \;
\eq

\subsubsection{Cancellation identities}
\label{sect:cancellation}

The cancellation identities originate from a cancellation of $\Divisor_j$ in the numerator and the denominator. 
\bq
 \frac{\Divisor_j Q}{\Divisor_j^{\mu_j+1} \prod\limits_{\substack{i \in I_{\mathrm{all}}^0 \\ i \neq j }} \Divisor_i^{\mu_i}}
 & = & 
 \frac{Q}{\prod\limits_{i \in I_{\mathrm{all}}^0} \Divisor_i^{\mu_i}}.
\eq
They read
\bq
\label{eq_cancellation}
 \differentialform_{\mu_0 \dots (\mu_j+1) \dots \mu_{\ND}}\left[\Divisor_j \cdot Q\right]
 & = &
 \frac{1}{\eps}
 \frac{\prerel^{(j)}}{\prerel}
 \differentialform_{\mu_0 \dots \mu_j \dots \mu_{\ND}}\left[Q\right],
 \;
\eq
where
$\prerel$ denotes the relative prefactor of $\differentialform_{\dots \mu_j \dots}[Q]$ and
$\prerel^{(j)}$ denotes the relative prefactor of $\differentialform_{\dots (\mu_j+1) \dots}[\Divisor_j \cdot Q]$.
The ratio of the relative prefactors depends on $a_j, b_j, \mu_j$ and $\eps$ and is always linear in $\eps$:
\bq
\label{linear_factor_in_eps}
 \frac{\prerel^{(j)}}{\prerel}
 & = &
 \frac{1}{2} a_j - \mu_j + \frac{b_j}{2} \eps.
\eq


\subsection{Poles and residues}
\label{sect:poles_and_residues}

We will count the pole order and the number of non-zero consecutive residues 
of the integrands $\differentialform_{\mu_0 \dots \mu_{\ND}}[Q]$ defined in eq.~(\ref{def_input_data}).
The regulator $\eps$ in the twist function will not be relevant for this counting and
we define $\differentialform^0_{\mu_0 \dots \mu_{\ND}}[Q]$ by replacing $U(z)$ by $U_0(z)$:
\bq
 \differentialform^0_{\mu_0 \dots \mu_{\ND}}\left[Q\right]
 = 
 \prebaikov \;
 \preall \;
 U_0\left(z\right) \hat{\Phi}_{\mu_0 \dots \mu_{\ND}}\left[Q\right]
 \eta.
\eq
The differential form $\differentialform^0_{\mu_0 \dots \mu_{\ND}}[Q]$ is algebraic in $z$, the algebraic part comes entirely from the square roots of the odd polynomials in $U_0$.
The differential forms $\differentialform_{\mu_0 \dots \mu_{\ND}}[Q]$ and $\differentialform^0_{\mu_0 \dots \mu_{\ND}}[Q]$
are holomorphic on ${\mathbb C} {\mathbb P}^{\NV}-D$, where the divisor $D$ is defined in eq.~(\ref{def_divisior_D}).
They may have singularities on the divisor $D$.

For a differential form $\differentialform_{\mu_0 \dots \mu_{\ND}}[Q]$ (and $\differentialform^0_{\mu_0 \dots \mu_{\ND}}[Q]$)
we define the active divisor as follows: We set
\bq
 D_{\mathrm{active}} \; = \; \bigcup\limits_{i \in I_{\mathrm{active}}^0} D_i,
 & \mbox{where} &
 I^0_{\mathrm{active}} \; = \; \left\{ \; i \; | \; i \in I^0_{\mathrm{even}} \; \mathrm{and} \; \mu_i > 0 \; \right\} \; \cup \; I^0_{\mathrm{odd}}.
\eq
In section~\ref{sect:definitions} we divided the polynomials defining the divisors into even and odd, 
depending on whether the $\eps^0$-term of the corresponding exponent in the twist function is zero or $(-\frac{1}{2})$.
The even and the odd polynomials will play different roles in the following.
If an even polynomial is present in the denominator, we may take a residue and reduce to a simpler problem.
The odd polynomials define a geometry, which -- to a first approximation -- is associated with the Feynman integral.
We will elaborate on this point in section~\ref{sect:geometry}.

\subsubsection{Blow-ups}

Before we start to count the pole order, we need to discuss a technical detail:
In general, the active divisor $D_{\mathrm{active}}$ (and the divisor $D$) will have self-intersections.
We distinguish normal crossings and non-normal crossings.
\begin{figure}
\begin{center}
\includegraphics[scale=1.0]{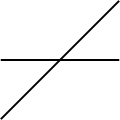}
\hspace*{20mm}
\includegraphics[scale=1.0]{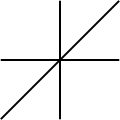}
\hspace*{20mm}
\includegraphics[scale=1.0]{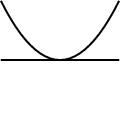}
\end{center}
\caption{
Crossings in a plane: The left figure shows a normal crossing, the middle figure and the right figure show non-normal crossings.
}
\label{fig:crossings}
\end{figure}
Fig.~\ref{fig:crossings} shows an example of a normal crossing in a plane and two examples of non-normal crossings in a plane.
In order to determine the pole order correctly, we would like to have normal crossings.
A theorem by Hironaka \cite{Hironaka:1964} guarantees that by a (finite) sequence of blow-ups we may always convert
a non-normal crossing to a normal crossing.
Within the physics community, this procedure is known as sector decomposition \cite{Binoth:2000ps,Bogner:2007cr,Smirnov:2008aw,Kaneko:2009qx}.
Whenever we encounter a non-normal crossing, we will always perform a blow-up.
\begin{myexample}
\label{example_blow_up}
Consider the twist function (which corresponds to a Feynman integral contributing to the three-loop electron self-energy, which is shown in fig. \ref{fig:SE_electron})
\bq 
 U\left(z_0,z_1,z_2\right)
 & = &
 P_0^{4\eps}
 P_1^\eps 
 P_2^{-2\eps}
 P_3^{-\frac{1}{2}}
 P_4^{-\frac{1}{2}-\eps}
 P_5^{-\frac{1}{2}-\eps}
\eq
with
\begin{align}
 P_0 & = z_0, & P_3 & = z_1, \nonumber \\
 P_1 & = z_2, & P_4 & = z_1 + 4 x z_0, \nonumber \\
 P_2 &= z_2+z_0, & P_5 & = \left(z_2-z_1\right)^2 +2 x z_0 \left( z_1+z_2\right) + x^2 z_0^2.
\end{align}
$P_3$ and $P_5$ have a non-normal intersection at $[z_0:z_1:z_2]=[1:0:-x]$.
In the chart $z_0=1$ we blow-up the point $(z_1,z_2)=(0,-x)$.
The linear transformation $\tilde{z}_2=z_2-z_1+x$ shifts the point to $(z_1,\tilde{z}_2)=(0,0)$.
The blow-up gives us three charts:
\bq
\begin{array}{lll}
\mathrm{Chart} \; 1: & z_1 = z_1', & \tilde{z}_2 = z_1' z_2', \\
\mathrm{Chart} \; 2: & z_1 = z_1'' (z_2'')^2, & \tilde{z}_2 = z_2'', \\
\mathrm{Chart} \; 3: & z_1 = (z_1''')^2 z_2''', & \tilde{z}_2 = z_1''' z_2'''. \\
\end{array}
\eq
\end{myexample}

\subsubsection{Poles}
\label{sect:poles}

The differential forms $\differentialform_{\mu_0 \dots \mu_{\ND}}[Q]$ and $\differentialform^0_{\mu_0 \dots \mu_{\ND}}[Q]$
may have singularities on $D$.
We now define an integer $o$, called the pole order of the differential form, which measures how singular the differential form is.
We define the pole order $o$ of $\differentialform^0_{\mu_0 \dots \mu_{\ND}}[Q]$ as follows:
If the active divisors have non-normal crossings, we first perform a blow-up. 
After this step, we may assume that we only have normal crossings.
The pole order is the maximum of the pole orders at individual points.
In the univariate case we define the pole order of $z^{-\alpha} dz$ at $z=0$ for $\alpha>0$ to be 
$\lfloor \alpha \rfloor$, where $\lfloor x \rfloor$ denotes the floor function.
Hence the pole order of $dz/z$ at $z=0$ is $1$,
the pole order of $dz/\sqrt{z}$ at $z=0$ is $0$.
In the multivariate case, we define the pole order for normal-crossing singularities to be additive,
i.e. the pole order of $dz_1/z_1 \wedge dz_2/z_2^2$ at
$(z_1,z_2)=(0,0)$ is $3$.
We define the pole order of $\differentialform_{\mu_0 \dots \mu_{\ND}}[Q]$ to be the pole order of $\differentialform^0_{\mu_0 \dots \mu_{\ND}}[Q]$.
\begin{myexample}
The differential form
\bq
 \frac{z_0 dz_1 \wedge dz_2 - z_1 dz_0 \wedge dz_2 + z_2 dz_0 \wedge dz_1}{z_1 z_2 \sqrt{\left(z_2-z_1\right)^2 +2 x z_0 \left( z_1+z_2\right) + x^2 z_0^2}}
\eq
has pole order $o=2$.
The pole order $2$ is attained at the points
$[1:0:0]$, $[1:0:-x]$ and $[1:-x:0]$.
The differential form
\bq
 \frac{z_0 dz_1 \wedge dz_2 - z_1 dz_0 \wedge dz_2 + z_2 dz_0 \wedge dz_1}{z_1 \left[\left(z_2-z_1\right)^2 +2 x z_0 \left( z_1+z_2\right) + x^2 z_0^2\right]}
\eq
has pole order $o=3$. The pole order $3$ is attained at the point
$[z_0:z_1:z_2]=[1:0:-x]$.
At the points $[1:0:-x]$ and $[1:-x:0]$, we have non-normal crossings and a blow-up, as in example~\ref{example_blow_up} is required.
\end{myexample}

\subsubsection{Residues}
\label{sect:residues}

If an even polynomial appears in the denominator of a differential form, we may take a residue.
If more than one even polynomial appears in the denominator, we may iterate this procedure.
We denote by $r$ the largest number such that the $r$-fold residue of $\differentialform^0_{\mu_0 \dots \mu_{\ND}}[Q]$ is non-zero.
We define the number of non-zero consecutive residues of $\differentialform_{\mu_0 \dots \mu_{\ND}}[Q]$ to be $r$.
\begin{myexample}
The differential form
\bq
 \frac{z_0 dz_1 \wedge dz_2 - z_1 dz_0 \wedge dz_2 + z_2 dz_0 \wedge dz_1}{z_1 z_2 \sqrt{\left(z_2-z_1\right)^2 +2 x z_0 \left( z_1+z_2\right) + x^2 z_0^2}}
\eq
has $r=2$.
We may take a two-fold residue at $[z_0:z_1:z_2]=[1:0:0]$, at $[z_0:z_1:z_2]=[1:0:-x]$ and at $[z_0:z_1:z_2]=[1:-x:0]$. 
\end{myexample}


\subsection{Filtrations}
\label{sect:filtrations}

For each object $\differentialform_{\mu_0 \dots \mu_{\ND}}[Q] \in \Agen^{\NV}_\omega$ we defined three integers $r, o, |\mu|$.
The number $r$ of non-zero consecutive residues has been defined in section~\ref{sect:residues},
the pole order $o$ has been defined in section~\ref{sect:poles} and
the quantity $|\mu|$ has been defined in eq.~(\ref{def_mu}).
These numbers define three filtrations $W_\bullet$, $\Fgeom^\bullet$ and $\Fcomb^\bullet$ on the space $\Agen^{\NV}_\omega$.
\begin{description}

\item{(i)} The weight filtration $W_\bullet$ is defined by
\begin{alignat}{2}
 \differentialform_{\mu_0 \dots \mu_{\ND}}[Q] & \in W_w \Agen^{\NV}_\omega & \quad \mbox{if} & \quad \NV + r \le w.
\end{alignat}

\item{(ii)} The filtration $\Fgeom^\bullet$ is defined by
\begin{alignat}{2}
 \differentialform_{\mu_0 \dots \mu_{\ND}}[Q] & \in \Fgeom^{p} \Agen^{\NV}_\omega & \quad \mbox{if} & \quad \NV+r-o \ge p.
\end{alignat}

\item{(iii)} The filtration $\Fcomb^\bullet$ is defined by
\begin{alignat}{2}
 \differentialform_{\mu_0 \dots \mu_{\ND}}[Q] & \in \Fcomb^{p'} \Agen^{\NV}_\omega & \quad \mbox{if} & \quad \NV-\absmu \ge p'.
\end{alignat}

\end{description}
These filtrations are inspired by Hodge theory.
The weight filtration is the standard weight filtration from Hodge theory. As $\NV$ is fixed, it is essentially 
a filtration by the number of non-zero consecutive residues.
At fixed weight $w$, the filtration $\Fgeom^\bullet$ is a filtration by the pole order $o$.
The combinatorial filtration $\Fcomb^\bullet$ is a filtration by the quantity $\absmu$.
The general idea is that we always work modulo simpler terms, i.e. modulo terms with fewer residues, lower pole order
or a smaller sum of indices $\absmu$. 

With the filtrations $\Fcomb^\bullet$, $\Fgeom^\bullet$ and $W_\bullet$ 
we may in principle decompose $\Agen^{\NV}_\omega$ into subspaces $\Agen_\omega^{|\mu|,o,r}$.
A projection onto two indices is easier to display and
we define two coarse-grained decompositions:
We set 
\bq
 \Ageom^{p,q} \; = \; \mathrm{Gr}^{p}_{\Fgeom} \mathrm{Gr}^W_{p+q} \Agen^{\NV}_\omega
 & \mbox{and} &
 \Acomb^{p',q'} \; = \; \mathrm{Gr}^{p'}_{\Fcomb} \mathrm{Gr}^W_{p'+q'} \Agen^{\NV}_\omega,
\eq
where the graded parts are defined by 
\bq
 \mathrm{Gr}^W_{w} X \; = \; W_{w} X / W_{w-1} X
 & \mbox{and} &
 \mathrm{Gr}^{p}_{F} X \; = \; F^p X / F^{p+1} X.
\eq
We carry over the filtrations to $\Hgen^{\NV}_\omega$ as follows:
Cohomology classes in $\Hgen^{\NV}_\omega$ are cosets, and different elements of a coset may have different values of $(|\mu|,o,r)$.
In section~\ref{sect:method}, we define a total order on the elements $\differentialform_{\mu_0 \dots \mu_{\ND}}[Q]$ of $\Agen^{\NV}_\omega$
based on the values $(|\mu|,o,r)$.
We define the standard representative of a cohomology class 
to be the object $\differentialform_{\mu_0 \dots \mu_{\ND}}[Q]$ in this cohomology class, which is minimal
with respect to the order relation.
For a fixed total order, this element is unique.
We call this standard representative a master integrand.
We let $\Hgeom^{p,q}$ be generated by all master integrands $\differentialform \in \Ageom^{p,q}$
and $\Hcomb^{p',q'}$ be generated by all master integrands $\differentialform \in \Acomb^{p',q'}$.
We set 
\bq
 \hgeom^{p,q} \; = \; \dim \Hgeom^{p,q}
 & \mbox{and} &
 \hcomb^{p',q'} \; = \; \dim \Hcomb^{p',q'}.
\eq
We display this information in the form of a Hodge-like diagram, for example, for $\NV=2$
\begin{center}
\begin{axopicture}(280,140)(0,0)
\Text(70,60)[c]{$\hgeom^{2,0}$}
\Text(110,60)[c]{$\hgeom^{1,1}$}
\Text(150,60)[c]{$\hgeom^{0,2}$}
\Text(90,80)[c]{$\hgeom^{2,1}$}
\Text(130,80)[c]{$\hgeom^{1,2}$}
\Text(110,100)[c]{$\hgeom^{2,2}$}
\DashLine(30,70)(260,70){6}
\DashLine(30,90)(260,90){6}
\DashLine(30,110)(260,110){6}
\Text(253,60)[c]{$W_2$}
\Text(253,80)[c]{$W_3$}
\Text(253,100)[c]{$W_4$}
\Line(260,70)(260,64)
\Line(260,90)(260,84)
\Line(260,110)(260,104)
\DashLine(120,10)(240,130){3}
\DashLine(80,10)(200,130){3}
\DashLine(40,10)(160,130){3}
\Text(110,20)[c]{$\Fgeom^0$}
\Text(70,20)[c]{$\Fgeom^1$}
\Text(30,20)[c]{$\Fgeom^2$}
\Line(120,10)(117,10)
\Line(80,10)(77,10)
\Line(40,10)(37,10)
\end{axopicture}
\end{center}


\subsection{Localisations}
\label{sect:localisations}

If a polynomial $\Divisor_i$ with $i \in I_{\mathrm{even}}^0$
appears in the denominator, we may take a residue, or in other words, localise on $\Divisor_i=0$.
By taking a residue we go from a variety of dimension $\NV$ to a sub-variety of dimension $(\NV-1)$.
The equation $\Divisor_i=0$ with $i \in I_{\mathrm{even}}^0$ defines a distinguished sub-variety, which plays an important role in our algorithm.
Clearly, this process can be iterated. On a sub-variety we may take a further residue and obtain a sub-sub-variety of codimension two.
In general, we obtain a set of nested sub-varieties as shown in fig.~\ref{fig:mixed_geometry}.
\begin{figure}
\begin{center}
\includegraphics[scale=1.0]{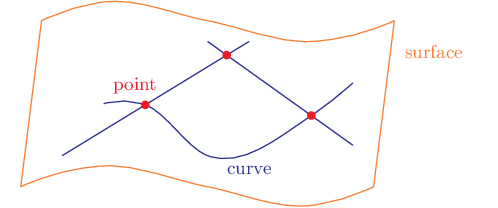}
\end{center}
\caption{
A mixed geometry:
Inside a surface of dimension two there can be curves of dimension one and points of dimension zero.
}
\label{fig:mixed_geometry}
\end{figure}
Ultimately, these distinguished sub-varieties determine the decomposition of the vector space $V^\NV$ of Feynman integrals on the maximal cut into
pieces associated with specific geometries.

Our algorithm will first recursively treat all possible localisations. 
In the simplest case, where $\Divisor$ equals one Baikov variable (e.g. $\Divisor_i=z_r$), the localisation 
eliminates this Baikov variable, and we are left with a sub-problem with one Baikov variable less.
Localisations are rather straightforward if each divisor we localise on is linear in at least one Baikov variable.
Below, we present a method which avoids algebraic extensions even if this does not hold.
The philosophy is not to eliminate variables, but to work with the original variables modulo an ideal on a restricted
set of objects with a restricted set of linear relations.
As a localisation defines a sub-variety of ${\mathbb C}{\mathbb P}^{\NV}$, we inevitably need to work with charts.
\begin{myexample}
\label{example_charts}
In general, we need to consider all charts. Integration-by-parts identities in a single chart might not give the full
set of linear relations.
To see this, consider in ${\mathbb C}{\mathbb P}^2$ the twist 
\bq
 U & = &
 P_0^{2\eps} P_1^{2\eps} P_2^{-\eps} P_3^{-\frac{1}{2}-\eps}.
\eq
where the polynomials are given by
\begin{align}
 P_0
 & =
 z_0,
 &
 P_2
 & =
 z_1 z_2 + x_1 z_0 z_2 - z_1^2,
 \nonumber \\
 P_1
 & =
 z_1,
 &
 P_3
 & = 
 \left[ x_2 z_1 - \left(1+x_2\right) z_2 + \left(x_1-x_2\right) z_0\right]^2 
 + 4 x_1 \left(1+x_2\right) z_0 z_2.
\end{align}
This example has six master integrands.
However, if we only use the integration-by-parts identities in the chart $z_2=1$ 
(by excluding $j=2$ in eq.~(\ref{basic_ibp_relation})), we will only be able to reduce the integrands to seven free integrands. 
This example will be continued in example~\ref{example_square_free_sub_geometry}
and discussed in detail in section~\ref{sect:sector_93_moeller}.
\end{myexample}
In general, we denote localised integrands by
\bq
\label{notation_localised}
 \mathrm{Res}_{{\mathcal P}_1, \dots, {\mathcal P}_r } \differentialform_{\mu_0,\dots,\mu_{\ND}}\left[Q, I, {\mathcal P}_s \right],
\eq
where ${\mathcal P}_1, \dots, {\mathcal P}_r$ denotes an ordered sequence, specifying the order in which the residues are taken
(the first residue is taken with respect to ${\mathcal P}_1$, the last residue is taken with respect to ${\mathcal P}_r$).
$I = \langle {\mathcal P}_1, \dots, {\mathcal P}_r \rangle$ denotes the ideal we localise on and 
${\mathcal P}_s \notin I$ denotes a scale polynomial, which we need to introduce in order to keep the integrand homogeneous.
We will only be considering differential forms $\differentialform_{\mu_0,\dots,\mu_{\ND}}$ which have simple poles along $I$.
Let $G$ be a Gr\"obner basis of the ideal $I$. With the help of the Gr\"obner basis, we may reduce the numerator
polynomial $Q$ modulo the ideal $I$.
If 
\bq
 Q_1 \mod I & = & Q_2 \mod I
\eq
we have
\bq
 \mathrm{Res}_{{\mathcal P}_1, \dots, {\mathcal P}_r } \differentialform_{\mu_0,\dots,\mu_{\ND}}\left[Q_1, I, {\mathcal P}_s \right]
 & = &
 \mathrm{Res}_{{\mathcal P}_1, \dots, {\mathcal P}_r } \differentialform_{\mu_0,\dots,\mu_{\ND}}\left[Q_2, I, {\mathcal P}_s \right],
\eq
since we assumed simple poles and terms regular on $I$ will vanish under the residue map.
Among the pre-images of the residue map, we may single out the one where the numerator polynomial $Q$ is reduced
modulo $I$:
\bq
\label{notation_preimage}
 \differentialform_{\mu_0,\dots,\mu_{\ND}}\left[Q\right],
\eq
with $Q = Q \mod I$.
There is a one-to-one map between the residues and these special pre-images.
When there is no risk of confusion, we will use the shorter notation of the pre-image of eq.~(\ref{notation_preimage}) instead of the longer notation of eq.~(\ref{notation_localised}).

We treat cases with increasing complexity.
We start in section~\ref{sect:divisor_equal_Baikov_var} with the simplest case, where $\Divisor$ equals one Baikov variable.
In section~\ref{sect:localisation_on_single_divisor} we discuss the case of a single divisor, not necessarily equal to one Baikov variable.
In section~\ref{sect:localisation_on_multiple_divisors} we treat the general case, where we take multiple residues and localise on the intersection of multiple divisors.
In section~\ref{sect:localisation_on_a_point}, we discuss the special case where the localisation on multiple divisors leads
to a localisation on a point.

\subsubsection{Localisation on $z_{\NV}=0$}
\label{sect:divisor_equal_Baikov_var}

We start from the twist function
\bq
 U & = & \prod\limits_{k=0}^{\ND} \Divisor_k^{\alpha_k}.
\eq
We first consider the case where the divisor defining the localisation equals one Baikov variable.
Without loss of generality, we may assume
\bq
 \Divisor_{\ND} & = & z_{\NV}.
\eq
We assume that $\Divisor_{\ND}$ is an even polynomial, hence 
\bq
 \alpha_{\ND} & = & \frac{b_{\ND} \eps}{2}.
\eq
We denote by $V(\langle z_{\NV} \rangle)$ the sub-variety in ${\mathbb C}{\mathbb P}^{\NV}$ defined by $z_{\NV}=0$.
Although $V(\langle z_{\NV} \rangle) \sim {\mathbb C}{\mathbb P}^{\NV-1}$,
the situation for the differential forms is a little bit trickier:
We would like to study the ``residue'' of $\differentialform_{\mu_0,\dots,\mu_{(\ND-1)},1}[Q]$ on $V(\langle z_{\NV} \rangle)$.
We first have to remove 
\bq
 z_{\NV}^{\frac{b_{\ND} \eps}{2}}
\eq
from the twist function in order to be able to take the residue.
However, we can not just remove this factor alone, as the remainder would no longer be homogeneous and well-defined on ${\mathbb C}{\mathbb P}^{\NV}$.
We choose a homogeneous polynomial $\Divisor_s$ of degree $\deg \Divisor_s = d_s$, which we call the scale polynomial,
and define $\differentialform_{\mu_0,\dots,\mu_{(\ND-1)},1}[Q,\langle z_{\NV} \rangle, \Divisor_s]$
by
\bq
 \differentialform_{\mu_0,\dots,\mu_{(\ND-1)},1}\left[Q\right]
 & = & 
 \frac{\prerel}{\prerel^{\mathrm{loc}}}
 \left( \frac{z_{\NV}^{d_s}}{\Divisor_s}\right)^{\frac{b_{\ND} \eps}{2 d_s}}
 \differentialform_{\mu_0,\dots,\mu_{(\ND-1)},1}\left[Q,\langle z_{\NV} \rangle, \Divisor_s \right],
\eq
where
$\prerel$ is the relative prefactor of $\differentialform_{\mu_0,\dots,\mu_{(\ND-1)},1}[Q]$ 
and $\prerel^{\mathrm{loc}}$ is the relative prefactor of $\differentialform_{\mu_0,\dots,\mu_{(\ND-1)},1}[Q,\langle z_{\NV} \rangle, \Divisor_s]$,
computed from the twist function
\bq
 U^{\mathrm{loc}} & = & 
 P_s^{\left(\alpha_s+\frac{b_{\ND} \eps}{2 d_s}\right)}
 \prod\limits_{\substack{k=0 \\ k \neq s}}^{\ND-1} P_k^{\alpha_k}.
\eq
The exponents of the divisors in the original twist function $U$ and in the modified twist function $U^{\mathrm{loc}}$ 
differ for $k \in \{s,\ND\}$.
By construction, the exponent of $P_{\ND}$ in $U^{\mathrm{loc}}$ will be zero.
In order to compute the relative prefactor $\prerel^{\mathrm{loc}}$ we have to define the falling factorial for this case.
As $\mu_{\ND}=1$ is constant, the falling factorial will also be constant for all cases of interest.
We can define it to be any non-zero finite value.
In accordance with section~\ref{sect:localisation_on_a_point} we will use the convention that
\bq 
 (0)_1 & = & -\eps.
\eq
In order to avoid introducing new spurious singularities, it is advantageous to choose
\bq
\label{choice_scale_polynomial}
 \Divisor_s & \in & \{ \Divisor_0, \Divisor_1, \dots, \Divisor_{\ND-1} \}.
\eq
$\differentialform_{\mu_0,\dots,\mu_{(\ND-1)},1}[Q,\langle z_{\NV} \rangle, \Divisor_s]$ is a well-defined differential $\NV$-form on 
${\mathbb C}{\mathbb P}^{\NV}$ with a simple pole along $z_{\NV}=0$.
Hence, we may take the residue at $z_{\NV}=0$, which we denote by
\bq
 \mathrm{Res}_{z_{\NV}} \differentialform_{\mu_0,\dots,\mu_{(\ND-1)},1}[Q,\langle z_{\NV} \rangle, \Divisor_s].
\eq
$\mathrm{Res}_{z_{\NV}} \differentialform_{\mu_0,\dots,\mu_{(\ND-1)},1}[Q,\langle z_{\NV} \rangle, \Divisor_s]$ 
is a differential $(\NV-1)$-form on $V(\langle z_{\NV} \rangle) \sim {\mathbb C}{\mathbb P}^{\NV-1}$.
Let us define $\tilde{Q}$ by
\bq
\label{def_Qtilde}
 \tilde{Q} & = & \left. Q \right|_{z_{\NV}=0} \; = \; Q \mod \langle z_{\NV} \rangle.
\eq
It is clear that 
\bq
 \mathrm{Res}_{z_{\NV}} \differentialform_{\mu_0,\dots,\mu_{(\ND-1)},1}[\tilde{Q},\langle z_{\NV} \rangle, \Divisor_s]
 & = &
 \mathrm{Res}_{z_{\NV}} \differentialform_{\mu_0,\dots,\mu_{(\ND-1)},1}[Q,\langle z_{\NV} \rangle, \Divisor_s],
\eq
as terms regular on $z_{\NV}=0$ have no residue.
From now on, we will always assume that the numerator polynomial is reduced to $\tilde{Q}$ according to eq.~(\ref{def_Qtilde}).

We are interested in the integration-by-parts identities of
$\mathrm{Res}_{z_{\NV}} \differentialform_{\mu_0,\dots,\mu_{(\ND-1)},1}[\tilde{Q},\langle z_{\NV} \rangle, \Divisor_s]$
on $V(\langle z_{\NV} \rangle) \sim {\mathbb C}{\mathbb P}^{\NV-1}$.
From now on, we will assume that we always choose the scale polynomial from the set of the remaining divisors, 
i.e. according to eq.~(\ref{choice_scale_polynomial}).
We will further assume that for $\dim V(\langle z_{\NV} \rangle) > 0$
the choice is generic, i.e. the exponent of $\tilde{P}_s$ in $U^{\mathrm{loc}}$ is not an integer.
The case $\dim V(\langle z_{\NV} \rangle) = 0$ is special and will be treated in section~\ref{sect:localisation_on_a_point}.
With this assumption
the integration-by-parts identities
are then given for $j \in \{0,\dots,\NV-1\}$ by
\bq
\label{ibp_zn_eq_0_localised}
 0 & = &
 \frac{1}{\eps}
 \mathrm{Res}_{z_{\NV}} \differentialform_{\mu_0,\dots,\mu_i,\dots,\mu_{(\ND-1)},1}\left[\partial_{z_j} \tilde{Q}_+,\langle z_{\NV} \rangle, \Divisor_s \right]
 \nonumber \\
 & &
 +
 \sum\limits_{i=0}^{\ND-1} 
 \mathrm{Res}_{z_{\NV}} \differentialform_{\mu_0,\dots,\mu_i+1,\dots,\mu_{(\ND-1)},1}\left[\tilde{Q}_+ \cdot \left( \partial_{z_j} \tilde{P}_i \right),\langle z_{\NV} \rangle, \Divisor_s  \right],
\eq
where $\tilde{Q}_+=\left. Q_+ \right|_{z_{\NV}=0}$.
A few remarks are in order:
\begin{enumerate}

\item Note that we exclude $j=\NV$ in eq.~(\ref{ibp_zn_eq_0_localised}).
We are only interested in derivatives within the hyperplane $V(\langle z_{\NV} \rangle)$, not in derivatives normal to it.

In the full system (without localisation)
we have additional integration-by-parts identities (namely the ones with $j=\NV$), for example
\bq
\lefteqn{
 \differentialform_{\mu_0,\dots,\mu_{(\ND-1)},1}\left[Q_+ \right]
 = } & &
 \nonumber \\
 & &
 - \frac{1}{\eps}
 \differentialform_{\mu_0,\dots,\mu_i,\dots,\mu_{(\ND-1)},0}\left[\partial_{z_{\NV}} Q_+\right]
 -
 \sum\limits_{i=0}^{\ND-1} 
 \differentialform_{\mu_0,\dots,\mu_i+1,\dots,\mu_{(\ND-1)},0}\left[Q_+ \cdot \left( \partial_{z_{\NV}} P_i \right) \right].
\eq

\item Note that even if we start with a minimal twist function $U$ and irreducible polynomials $\Divisor_k$, 
the modified twist function $U^{\mathrm{loc}}$ need not be minimal and the polynomials $\tilde{P}_k$ need not be irreducible.

\item The integral reduction tables 
for $\mathrm{Res}_{z_{\NV}} \differentialform_{\mu_0,\dots,\mu_{(\ND-1)},1}[\tilde{Q},\langle z_{\NV} \rangle, \Divisor_s]$ 
on $V(\langle z_{\NV} \rangle)$ depend on the choice for the scale polynomial $\Divisor_s$.

\item In general, the integral reductions 
for $\mathrm{Res}_{z_{\NV}} \differentialform_{\mu_0,\dots,\mu_{(\ND-1)},1}[\tilde{Q},\langle z_{\NV} \rangle, \Divisor_s]$ 
on $V(\langle z_{\NV} \rangle)$
are not related to the ones for $\differentialform_{\mu_0,\dots,\mu_{(\ND-1)},1}[Q]$ on ${\mathbb C}{\mathbb P}^{\NV}$, 
as they correspond to different twist functions $U$ and $U^{\mathrm{loc}}$.
The only information, which we carry from $V(\langle z_{\NV} \rangle)$ to ${\mathbb C}{\mathbb P}^{\NV}$ is the following:
If $\mathrm{Res}_{z_{\NV}} \differentialform_{\mu_0,\dots,\mu_{(\ND-1)},1}[\tilde{Q},\langle z_{\NV} \rangle, \Divisor_s]$ is a master integrand
on $V(\langle z_{\NV} \rangle)$, then we take
$\differentialform_{\mu_0,\dots,\mu_{(\ND-1)},1}[\tilde{Q}]$ as a preferred master candidate on ${\mathbb C}{\mathbb P}^{\NV}$.

\item On the localisation $z_{\NV}=0$, we only consider differential forms which have a simple pole at $z_{\NV}=0$,
i.e. the differential forms with $\mu_{\ND}=1$.
We do not consider differential forms, which have higher poles at $z_{\NV}=0$,
i.e. we exclude differential forms with $\mu_{\ND}>1$.
Note that differential forms with higher poles are not generated by the integration-by-parts identities
given in eq.~(\ref{ibp_zn_eq_0_localised}).

\item On the localisation $z_{\NV}=0$, differential forms with $\mu_{\ND}=0$ are automatically zero due to the residue operation.
This implies that we may restrict all numerator polynomials $Q$ and all seeds $Q_+$ to
\begin{align}
 \tilde{Q} & = Q \mod \langle z_{\NV} \rangle,
 &
 \tilde{Q}_+ & = Q_+ \mod \langle z_{\NV} \rangle.
\end{align}

\end{enumerate}

\subsubsection{Localisation on a single divisor}
\label{sect:localisation_on_single_divisor}

We now consider the case that we localise on $\Divisor_{\ND}=0$, but do not assume that $\Divisor_{\ND}$ equals one Baikov variable.
For the derivation of the integration-by-parts identities on the divisor $\Divisor_{\ND}=0$, we will assume that
\bq
 \frac{\partial \Divisor_{\ND}}{\partial z_{\NV}} & \neq & 0,
\eq
if $\Divisor_{\ND}$ does not depend on $z_{\NV}$, we replace $z_{\NV}$ by a Baikov variable it depends on.
We consider again the objects
\bq
 \differentialform_{\mu_0, \dots, \mu_{\ND-1}, 1}\left[Q,\langle \Divisor_{\ND} \rangle, \Divisor_s \right],
\eq
now defined by
\bq
 \differentialform_{\mu_0,\dots,\mu_{(\ND-1)},1}\left[Q\right]
 & = & 
 \frac{\prerel}{\prerel^{\mathrm{loc}}}
 \left( \frac{\Divisor_{\ND}^{d_s}}{\Divisor_s^{d_{\ND}}}\right)^{\frac{b_{\ND} \eps}{2 d_s}}
 \differentialform_{\mu_0,\dots,\mu_{(\ND-1)},1}\left[Q,\langle \Divisor_{\ND} \rangle, \Divisor_s \right].
\eq
The modified twist function $U^{\mathrm{loc}}$ is given by
\bq
 U^{\mathrm{loc}} & = & 
 P_s^{\left(\alpha_s+\frac{d_{\ND} b_{\ND} \eps}{2 d_s}\right)}
 \prod\limits_{\substack{k=0 \\ k \neq \ND}}^{\ND-1} P_k^{\alpha_k}.
\eq
We need to work out the integration-by-parts identities on the localisation.
These are the ones where we take a derivative within the hypersurface defined by $\Divisor_{\ND}=0$, but not the one transversal to it.
We work temporarily in the chart $z_0=1$.
The equation $\Divisor_{\ND}=0$ defines implicitly a function $z_\NV=f(z_1,\dots,z_{\NV-1})$ and we use $(z_1,\dots,z_{\NV-1})$ as coordinates
on the hypersurface.
With the help of the implicit function theorem, we obtain for the derivatives within the hypersurface
\bq
\label{derivative_in_hyperplane}
 \frac{\partial}{\partial z_j} + \frac{\partial f}{\partial z_j} \frac{\partial}{\partial z_{\NV}}
 & = &
 \frac{\partial}{\partial z_j} - \left( \frac{\partial \Divisor_{\ND}}{\partial z_{\NV}} \right)^{-1} \frac{\partial \Divisor_{\ND}}{\partial z_j} \frac{\partial}{\partial z_{\NV}},
 \;\;\;\;\;\;
 j \; \in \; \{1,\dots,\NV-1\}.
\eq
Hence, the integration-by-parts identities in the chart $z_0=1$ of the localisation $\Divisor_{\ND}=0$ read
\bq
\label{ibp_localised_Pn}
 0 & = &
 \frac{1}{\eps}
 \mathrm{Res}_{\Divisor_{\ND}} \differentialform_{\mu_0 \dots \mu_i \dots \mu_{\ND-1}, 1}\left[\left(\partial_{z_{\NV}} \Divisor_{\ND}\right) \left(\partial_{z_j} Q_+\right)-\left(\partial_{z_j} \Divisor_{\ND}\right) \left(\partial_{z_{\NV}} Q_+\right),\langle \Divisor_{\ND} \rangle, \Divisor_s \right]
 \\
 & &
 +
 \sum\limits_{i = 0}^{\ND-1} 
 \mathrm{Res}_{\Divisor_{\ND}} \differentialform_{\mu_0 \dots (\mu_i+1) \dots \mu_{\ND-1}, 1}\left[Q_+ \left( \left(\partial_{z_{\NV}} \Divisor_{\ND}\right) \left( \partial_{z_j} P_i \right) - \left(\partial_{z_j} \Divisor_{\ND}\right) \left( \partial_{z_{\NV}} P_i \right) \right),\langle \Divisor_{\ND} \rangle, \Divisor_s \right],
 \nonumber
\eq
where
\bq
\label{degree_Q_plus_localised}
 \deg Q_+
 & = &
 \sum\limits_{i = 0}^{\ND-1} \mu_i d_i - d_U - \NV + 1
\eq
and $j \in \{1,\dots,\NV-1\}$.
Integration-by-parts identities in other charts may be worked out along the same lines.

Again, we may reduce all numerator polynomials $Q$ modulo $\Divisor_{\ND}$
and it is sufficient to restrict the numerators $Q_+$ of the seeds to polynomials, which are not
reducible by $\Divisor_{\ND}$.
The cancellation identities take the following form: 
\bq
\label{eq_cancellation_localised_Pn}
\lefteqn{
 \mathrm{Res}_{\Divisor_{\ND}} \differentialform_{\mu_0, \dots, (\mu_j+1), \dots, \mu_{\ND-1}, 1}\left[\left( \Divisor_j \cdot Q \right) \mod \langle \Divisor_{\ND} \rangle,\langle \Divisor_{\ND} \rangle, \Divisor_s   \right]
 = } & &
 \nonumber \\
 & &
 \frac{1}{\eps}
 \frac{\prerel^{\mathrm{loc},(j)}}{\prerel^{\mathrm{loc}}}
 \mathrm{Res}_{\Divisor_{\ND}} 
 \differentialform_{\mu_0, \dots, \mu_j, \dots, \mu_{\ND-1}, 1}\left[Q\mod \langle \Divisor_{\ND} \rangle,\langle \Divisor_{\ND} \rangle, \Divisor_s  \right].
\eq
We may look at eq.~(\ref{derivative_in_hyperplane}) also from a slightly different perspective:
This equation actually defines a logarithmic vector field \cite{Bohm:2017qme} along $\Divisor_{\ND}=0$, since
\begin{align}
\left( \frac{\partial \Divisor_{\ND}}{\partial z_{\NV}} \frac{\partial }{\partial z_j} -\frac{\partial \Divisor_{\ND}}{\partial z_{j}} \frac{\partial }{\partial z_{\NV}} \right) \Divisor_{\ND}=0 \, .
\end{align}
Therefore, we can define the following $(\NV-1)$-form
\bq
 \Xi
 & = &
 \frac{Q_+}{\left(\NV-1\right)! \; \Divisor_{\ND} \prod\limits_{k=1}^{\ND-1} \Divisor_k^{\mu_k}} 
 \;\;
 \sum\limits_{j_0,j_1,\dots,j_{\NV}}
 z_{j_0} \frac{\partial \Divisor_{\ND} }{\partial z_{j_1}} \eps_{j_0 j_1 \dots, j_{\NV}} dz_{j_2} \wedge \dots \wedge dz_{j_{\NV}}.
\eq
The covariant derivative of $\Xi$ gives the integration-by-parts identities of eq.~(\ref{ibp_localised_Pn}).

\subsubsection{Localisation on multiple divisors}
\label{sect:localisation_on_multiple_divisors}

We now consider the general case, where we take $r$ consecutive residues.
We follow the same ideas as in the previous two subsections.
The case of $r$ consecutive residues is more cumbersome only from a notational perspective.
After we have taken the first residue, an odd polynomial may become a perfect square,
thus turning into an even polynomial.
For this reason, we rename the divisors where we take the residues to ${\mathcal P}_1, \dots, {\mathcal P}_r$.
We localise on the ideal
\bq
 {\mathcal I}
 & = &
 \langle {\mathcal P}_1, \dots, {\mathcal P}_r \rangle.
\eq
We temporarily work in a chart and without loss of generality we consider the chart $z_0=1$.
We will assume that the determinant of the $(r \times r)$-matrix
\bq
 {\mathcal J}_{ij}
 & = &
 \frac{\partial {\mathcal P}_i}{\partial z_{\NV-r+j}},
 \;\;\;\;\;\;
 i,j \in \{1,\dots,r\},
\eq
is non-zero.
If this is not the case, we consider a different chart.
With this assumption, we may again use the implicit function theorem and obtain for the derivatives within the variety 
defined by the ideal ${\mathcal I}$
\bq
 \det\left({\mathcal J}\right) \frac{\partial}{\partial z_j} 
 - \sum\limits_{k,l=1}^r \left( \mathrm{adj} {\mathcal J} \right)_{kl} \frac{\partial {\mathcal P}_l}{\partial z_j} \frac{\partial}{\partial z_{n-r+k}},
 \;\;\;\;\;\;
 j \in \{1,\dots,\NV-r\},
\eq
where $\mathrm{adj} {\mathcal J}$ denotes the adjoint matrix of ${\mathcal J}$.
Applying these differential operators to seeds gives us the integration-by-parts identities 
in the chart $z_0=1$ of the localisation defined by the ideal ${\mathcal I}$:
\bq
\label{ibp_localised}
 0
 & = &
 \frac{1}{\eps}
 \mathrm{Res}_{{\mathcal P}_1, \dots, {\mathcal P}_r} \differentialform_{\mu_0 \dots \mu_i \dots \mu_{\ND}}\left[\hat{Q},I,\Divisor_s\right]
 +
 \sum\limits_{i=0}^{\ND} 
 \mathrm{Res}_{{\mathcal P}_1, \dots, {\mathcal P}_r} \differentialform_{\mu_0 \dots (\mu_i+1) \dots \mu_{\ND}}\left[\hat{Q}_i,I,\Divisor_s\right],
\eq
with
\bq
 \hat{Q}
 & = &  
 \det\left({\mathcal J}\right) \left( \partial_{z_j} Q_+ \right)
 - \sum\limits_{k,l=1}^r \left( \mathrm{adj} {\mathcal J} \right)_{kl} \left( \partial_{z_j} {\mathcal P}_l \right) \left( \partial_{z_{n-r+k}} Q_+ \right),
 \nonumber \\
 \hat{Q}_i
 & = &
 Q_+ \left[
 \det\left({\mathcal J}\right) \left( \partial_{z_j} P_i \right)
 - \sum\limits_{k,l=1}^r \left( \mathrm{adj} {\mathcal J} \right)_{kl} \left( \partial_{z_j} {\mathcal P}_l \right) \left( \partial_{z_{n-r+k}} P_i \right) 
 \right],
\eq
where
\bq
 \deg Q_+
 & = &
 \sum\limits_{i = 0}^{\ND} \mu_i d_i - d_U - \NV + r - \sum\limits_{j=1}^r \deg {\mathcal P}_j
\eq
and $j \in \{1,\dots,\NV-r\}$.
Integration-by-parts identities in other charts may be worked out along the same lines.

We restrict to differential forms with simple poles along ${\mathcal P}_1, \dots, {\mathcal P}_r$,
i.e. in any chart we have that
\bq 
\label{localisation_identity}
 {\mathcal P}_1 \cdot \dots \cdot {\mathcal P}_r
 \differentialform_{\mu_0 \dots \mu_{\ND}}\left[Q\right]
\eq
has no residue on the $r$-fold localisation defined by the ideal ${\mathcal I}$.

In addition, we may reduce the polynomials in the numerator modulo the ideal ${\mathcal I}$.
For this task, we compute a Gr\"obner basis of ${\mathcal I}$.
The cancellation identities then take the following form: 
\bq
\label{eq_cancellation_localised}
\lefteqn{
 \mathrm{Res}_{{\mathcal P}_1, \dots, {\mathcal P}_r} \differentialform_{\mu_0 \dots (\mu_j+1) \dots \mu_{\ND}}\left[\left( \Divisor_j \cdot Q \right) \mod I,I,\Divisor_s \right]
 = } & &
 \nonumber \\
 & &
 \frac{1}{\eps}
 \frac{\prerel^{\mathrm{loc},(j)}}{\prerel^{\mathrm{loc}}}
 \mathrm{Res}_{{\mathcal P}_1, \dots, {\mathcal P}_r} \differentialform_{\mu_0 \dots \mu_j \dots \mu_{\ND}}\left[Q \mod I,I,\Divisor_s\right].
\eq
 
\subsubsection{Localisation on a point}
\label{sect:localisation_on_a_point}

The recursion starts with integrands localised on a point.
It is worth examining this case in detail.

We first investigate the case where we start from a zero-dimensional Baikov representation.
This corresponds to ${\mathbb C}{\mathbb P}^0$, which is a point.
In homogeneous coordinates, we denote this point by $[z_0]$ with $z_0 \neq 0$.
The $0$-form $\eta$ is given by
\bq
 \eta & = & z_0.
\eq
There is only one non-trivial irreducible homogeneous polynomial
\bq
 \Divisor_0 & = & z_0.
\eq
The twist is simply
\bq
 U & = & 1.
\eq
We may write the twist as
\bq
 U \; = \; P_0^{\alpha_0}
 & \mbox{with} &
 \alpha_0 \; = \; 0.
\eq
As already briefly mentioned in section~\ref{sect:projective_space}, in zero dimensions we no longer have
$\alpha_0 \neq {\mathbb Z}$.
In this subsection, we carefully investigate the implications of this fact.
 
The rational function $\hat{\Phi}$ has to be of degree $(-1)$.
The differential $0$-forms are of the form
\bq
 \differentialform_{\mu_0}\left[Q\right]
 & = &
 \prebaikov \; \preall \;
 \frac{Q}{\Divisor^{\mu_0}}
 \eta,
\eq
with $\deg Q = \mu_0-1$ and $\mu_0 \in {\mathbb N}$.
Note that $\mu_0 = 0$ is not allowed, since $\deg \hat{\Phi} = -1$.
The clutch prefactor is given by
\bq
 \preclutch & = & \frac{1}{\eps^{\mu_0}}.
\eq
The relative prefactor is a little bit trickier: 
We cannot simply take the limit $\alpha \rightarrow 0$ of the standard definition, as this would give zero.
We determine the relative prefactor from the limit
\bq
 \lim\limits_{\eps \rightarrow 0} \left(-\frac{1}{2} b_0 \eps \right)^{-1} \prerel
 & = &
 \left(-1\right)^{\mu_0} \left(\mu_0-1\right)!.
\eq
The factor $-\frac{1}{2} b_0$ is not essential and we set
\bq
 \prerel
 & = & 
 \eps \left(-1\right)^{\mu_0} \left(\mu_0-1\right)!.
\eq
With this definition, 
the integration-by-parts identities and the cancellation identities coincide and read
\bq
\label{ibp_point}
 \frac{\mu_0}{\eps} \differentialform_{\mu_0}\left[z_0^{\mu_0-1}\right] + \differentialform_{\mu_0+1}\left[z_0^{\mu_0}\right] & = & 0.
\eq
Thus, there is one master integrand, the one with minimal $|\mu|$, which is given by
\bq
 \differentialform_{1}\left[1\right].
\eq
Note that although ${\mathbb C}{\mathbb P}^0$ is a point, there is still an integration-by-parts identity, the one given
by eq.~(\ref{ibp_point}).
This can be understood as follows: ${\mathbb C}{\mathbb P}^0$ equals ${\mathbb C}^\times / \sim$, where $z \sim z'$ if $z = \lambda z'$.
The integration-by-parts identity corresponds to the direction which is moded out by the equivalence relation $\sim$.

We now turn back to a general $\NV$-dimensional Baikov representation.
Localising on an ideal
\bq
 {\mathcal I}
 & = &
 \langle {\mathcal P}_1, \dots, {\mathcal P}_{\NV} \rangle
\eq
corresponds to localisation on $N$ points. From the remarks above, it is clear that in the case $N=1$, i.e. 
on the localisation of a single point,
we have one master integrand.
We take this integrand to be the simplest among the ones which have a non-zero residue at this point.
For the localisation on a zero-dimensional variety consisting of $N$ points, we take the $N$ simplest integrands such that the $(N\times N)$-period matrix (i.e. the matrix of residues) is non-degenerate.
\begin{myexample}
We illustrate this with a simple example.
We consider in ${\mathbb C}{\mathbb P}^1$ the twist
\bq
 U\left(z_0,z_1\right)
 & = &
 z_0^{-2\eps} \left( z_1^2 + z_0 z_1 + z_0^2 \right)^\eps,
\eq
together with $\prebaikov=\preabs=1$.
We have
\begin{align}
 P_0 & = z_0,
 &
 P_1 & = z_1^2 + z_0 z_1 + z_0^2.
\end{align}
The objects on ${\mathbb C}{\mathbb P}^1$ are
\bq
 \differentialform_{\mu_0 \mu_1}\left[Q\right]
 & = &
 \preclutch \;
 \prerel \; 
 U\left(z_0,z_1\right) \hat{\Phi}_{\mu_0 \mu_1}\left[Q\right]
 \eta,
 \;\;\;\;\;\;
 \hat{\Phi}_{\mu_0 \mu_1}\left[Q\right] \; = \; \frac{Q}{P_0^{\mu_0} P_1^{\mu_1}},
\eq
with $\deg Q = \mu_0+2\mu_1-2$.
From an analysis of the critical points, we expect $\dim H_\omega^1 = 1$.

Let us first look at the localisation $P_0=0$. We consider the differential one-forms
$\differentialform_{1\mu_1}[Q]$ and the corresponding zero-forms $\mathrm{Res}_{z_0}\differentialform_{1\mu_1}[\tilde{Q},\langle z_0\rangle,\Divisor_1]$.
From the homogeneity condition, we need $\mu_1 \ge 1$.
The simplest integrand with a non-vanishing residue at $P_0=0$ is
\bq
\label{masters_point_1}
 \differentialform_{1 1}\left[z_1,\langle z_0\rangle,\Divisor_1 \right].
\eq
Let us now look at the localisation $P_1=0$.
We consider the differential one-forms $\differentialform_{\mu_0 1}[Q]$
and the corresponding zero-forms $\mathrm{Res}_{P_1}\differentialform_{\mu_0 1}[\tilde{Q},\langle \Divisor_1 \rangle,z_0]$.
The simplest integrands with a non-degenerate residue matrix at the two points of $P_1=0$ are
\bq
\label{masters_point_2}
 \differentialform_{0 1}\left[1,\langle \Divisor_1 \rangle,z_0\right], 
 \;\;\;
 \differentialform_{1 1}\left[z_1,\langle \Divisor_1 \rangle,z_0 \right].
\eq
After having done all sub-problems, we go back to the original problem and consider the integrands on ${\mathbb C}{\mathbb P}^1$.
From eq.~(\ref{masters_point_1}) and eq.(\ref{masters_point_2}) we infer that 
\bq
 \differentialform_{0 1}\left[1\right], 
 \;\;\;
 \differentialform_{1 1}\left[z_1\right]
\eq
are preferred candidates for master integrands on ${\mathbb C}{\mathbb P}^1$.
We express this preference by assigning a number $a=-w<0$ to the preferred candidates and $a=0$ to all
other objects.
In this case we assign $a=-2$ to $\differentialform_{1 1}[z_1]$ and $\differentialform_{0 1}[1]$.

Among the set of integration-by-parts identities on ${\mathbb C}{\mathbb P}^1$, we have
\bq
 \frac{1}{\eps}
 \differentialform_{0 1}\left[1\right]
 +
 \differentialform_{1 1}\left[ 2 z_1 + z_0 \right]
 +
 \differentialform_{0 2}\left[\left( 2 z_0 + z_1 \right)  \left( 2 z_1 + z_0 \right) \right]
 & = & 0,
 \nonumber \\
 \frac{1}{\eps}
 \differentialform_{0 1}\left[1\right]
 +
 \differentialform_{0 2}\left[\left( 2 z_1 + z_0 \right) \left( 2 z_0 + z_1 \right) \right]
 & = & 0.
\eq
These two equations imply
\bq
 \differentialform_{1 1}\left[ 2 z_1 + z_0 \right]
 & = & 0,
\eq
and together with the cancellation identity
\bq
 \differentialform_{1 1}\left[z_0 \right]
 & = &
 -2
 \differentialform_{0 1}\left[1\right]
\eq
it follows that $\differentialform_{1 1}[z_1]=\differentialform_{0 1}[1]$.
The algorithm will eliminate $\differentialform_{1 1}[z_1]$ (with $|\mu|=2$) 
in favour of $\differentialform_{0 1}[1]$ (with $|\mu|=1$) 
and we obtain one master integrand given by
\bq 
 \differentialform_{0 1}\left[1\right].
\eq
Remark: This example also shows that twisted cohomology ``sees'' the internal structure of $z_1^2+z_1+1$.
The dlog-form given in the chart $z_0=1$ by
\bq
 \frac{\left(2z_1+1\right)dz_1}{z_1^2+z_1+1},
\eq
is in the cohomology class of zero.
\end{myexample}


\subsection{Geometries}
\label{sect:geometry}

The algorithms we present are agnostic to the specific geometry associated with the Feynman integrals under consideration.
This does not mean that the algorithms are not based on geometric properties.
By using methods from twisted cohomology and Hodge theory, we are using methods which treat all specific geometries
in an (abstract) unified setting.
Although not strictly needed, it can be helpful to distil the specific geometries associated with Feynman integrals.
We are mainly interested in the geometries associated with the maximal cut of Feynman integrals.
We distinguish the twisted geometry, the nodal geometry and the square-free geometry.
The twisted geometry contains the full information and is determined by the twist function $U$.
For the nodal geometry and the square-free geometry only the $(\eps=0)$-part of the twist function is relevant.
The nodal geometry is determined by all polynomials entering the twist function, even and odd.
The square-free geometry is defined by the odd polynomials only.
A statement that a certain Feynman integral corresponds to a certain geometry usually refers to the square-free geometry.
We would like to stress that this neglects sub-geometries due to localisations on even polynomials as well as the $\eps$-dependent part in the twist function.

\subsubsection{The twisted geometry}

The twisted geometry carries the most information.
It is defined as in section~\ref{sect:twisted_cohomology} by the Baikov representation.
We consider the space
\bq
 {\mathbb C}{\mathbb P}^\NV-D,
\eq
where the divisor $D$ is defined in eq.~(\ref{def_divisior_D}),
together with the connection
\bq
 \nabla_F & = & d_F + \omega,
\eq
defined in eq.~(\ref{def_nabla_F}). 
We are interested in the twisted cohomology group
\bq
 H^\NV_\omega,
\eq
which is defined as the set of equivalence classes of the $\nabla_F$-closed differential $(\NV,0)$-forms holomorphic
on ${\mathbb C}{\mathbb P}^\NV-D$ modulo the $\nabla_F$-exact ones.

\subsubsection{The nodal geometry}

The nodal geometry carries less information. We remove the multivaluedness entering through the twist function $U$ and
the dependence on the dimensional regulator $\eps$.
However, we keep the information on all polynomials entering the twist function, even and odd.

We define the nodal geometry $Y_{\mathrm{nodal}}$ as a hypersurface
in a weighted projective space $X_{\mathrm{nodal}} = {\mathbb C}{\mathbb P}^{\NV+1}_{1,\dots,1,h}$
by the equation
\bq
\label{def_hypersurface}
 y^2 - F\left(z_0,z_1,...,z_{\NV}\right) & = & 0,
\eq
where
\bq
 F\left(z_0,z_1,...,z_{\NV}\right)
 & = &
 \prod\limits_{i \in I^0_{\mathrm{even}}} \left[ \Divisor_i\left(z\right) \right]^2
 \prod\limits_{i \in I^0_{\mathrm{odd}}} \Divisor_i\left(z\right).
\eq
The weight $h$ of the coordinate $y$ is given by
\bq
 h & = & 
 \prod\limits_{i \in I^0_{\mathrm{even}}} d_i
 +
 \frac{1}{2} \prod\limits_{i \in I^0_{\mathrm{odd}}} d_i.
\eq
Note that $\prod\limits_{i \in I^0_{\mathrm{odd}}} d_i$ is always even, this follows from the definition of $a_0$ in eq.~(\ref{def_a_0_b_0}).
The nodal geometry $Y_{\mathrm{nodal}}$ carries a mixed Hodge structure.
\begin{myexample}
We consider again example 2:
In ${\mathbb C}{\mathbb P}^1$ we had the polynomials
\begin{align}
 P_0 & = z_0,
 &
 P_1 & = z_1-x_2z_0, 
 &
 P_2 & = z_1+(4-x_2)z_0,
 &
 P_3 & = \left( z_1 + z_0 \right)^2 - 4 \left[ x_2 +\frac{\left(1-x_2\right)^2}{x_1} \right] z_0^2
\end{align}
together with $I^0_{\mathrm{even}}=\{0\}$ and $I^0_{\mathrm{odd}}=\{1,2,3\}$.
The nodal geometry $Y_{\mathrm{nodal}}$ is given by the hypersurface
\bq
 y^2 & = & P_0^2 P_1 P_2 P_3
\eq
in ${\mathbb C}{\mathbb P}^2_{113}$.
This is a nodal curve of arithmetic genus $2$ and geometric genus $1$.
\end{myexample}

\subsubsection{The square-free geometry}

The square-free geometry carries even less information. We remove the information on the even polynomials.
We define the square-free geometry $Y_{\mathrm{square-free}}$ as a hypersurface
in a weighted projective space $X_{\mathrm{square-free}} = {\mathbb C}{\mathbb P}^{\NV+1}_{1,\dots,1,h}$
by the equation
\bq
\label{def_hypersurface_square_free}
 y^2 - F\left(z_0,z_1,...,z_{\NV}\right) & = & 0,
\eq
where
\bq
 F\left(z_0,z_1,...,z_{\NV}\right)
 & = &
 \prod\limits_{i \in I^0_{\mathrm{odd}}} \Divisor_i\left(z\right).
\eq
The weight $h$ of the coordinate $y$ is now given by
\bq
 h & = & 
 \frac{1}{2} \prod\limits_{i \in I^0_{\mathrm{odd}}} d_i.
\eq
\begin{myexample}
We continue with example $2$ and example $11$: The square-free geometry $Y_{\mathrm{square-free}}$ is given by the hypersurface
\bq
 y^2 & = & P_1 P_2 P_3
\eq
in ${\mathbb C}{\mathbb P}^2_{112}$.
This is an elliptic curve.
\end{myexample}

\subsubsection{Square-free sub-geometries}

By first considering all possible localisations and then for each localisation, the square-free geometry on this localisation, we obtain the square-free sub-geometries.
These can be non-trivial, as the following example shows.
\begin{myexample}
\label{example_square_free_sub_geometry}
We continue with example~\ref{example_charts}.
We consider in ${\mathbb C}{\mathbb P}^2$ the polynomials
\begin{align}
 P_0
 & =
 z_0,
 &
 P_2
 & =
 z_1 z_2 + x_1 z_0 z_2 - z_1^2,
 \nonumber \\
 P_1
 & =
 z_1,
 &
 P_3
 & = 
 \left[ x_2 z_1 - \left(1+x_2\right) z_2 + \left(x_1-x_2\right) z_0\right]^2 
 + 4 x_1 \left(1+x_2\right) z_0 z_2
\end{align}
together with $I_{\mathrm{even}}^0=\{0,1,2\}$ and $I_{\mathrm{odd}}^0=\{3\}$.
On the localisation $P_2=0$, we obtain the square-free sub-geometry defined by the hypersurface
\bq
 y^2 & = & 
 z_1^4
 +2 (x_1+x_2+x_1 x_2) z_0 z_1^3
 + (3 x_1^2+x_2^2+6 x_2 x_1^2+x_2^2 x_1^2-2 x_2^2 x_1) z_0^2 z_1^2
 \nonumber \\
 & &
 +2 x_1 (x_1+x_1 x_2-x_2) (x_1-x_2) z_0^3 z_1 
 +x_1^2 (x_1-x_2)^2 z_0^4
\eq
in ${\mathbb C}{\mathbb P}^2_{112}$.
This is an elliptic curve.
This example will be discussed in more detail in section~\ref{sect:sector_93_moeller}.
\end{myexample}


\section{The method}
\label{sect:method}

In this section, we give a detailed description of the algorithm.
We start in section~\ref{sect:order_relation} with the modifications to the Laporta algorithm.
The two main modifications are that we use a new geometric order relation and that we first recursively consider all possible localisations,
in order to get preferred candidates for the master integrands.
In section~\ref{sect:step_1}, we turn to the system of differential equations. 
We introduce the concept of an $\Fgen^\bullet$-compatible differential equation.
We observe that the reduction algorithm from section~\ref{sect:order_relation} gives us an intermediate basis $J$
with an $\Fgen^\bullet$-compatible differential equation.
In section~\ref{sect:step_2}, we prove that we can always transform an $\Fgen^\bullet$-compatible differential equation
to an $\eps$-factorised form.


\subsection{The Laporta algorithm and the order relation}
\label{sect:order_relation}

We first construct a basis of master integrands for the vector space $\Hgen^{\NV}_{\omega}$.
We do this by first considering all possible localisations, 
starting with localisations on points, followed by localisations on curves, followed by surfaces, etc., until we reach the full system.
The essential idea is that master integrands on a localisation will define preferred candidates for master integrands in the next higher dimension.
 
Technically, we do this as follows:
Each localisation is defined by an ordered sequence ${\mathcal P}_1, \dots, {\mathcal P}_{r'}$, which in turn defines the ideal
\bq
 {\mathcal I}
 & = &
 \langle {\mathcal P}_1, \dots, {\mathcal P}_{r'} \rangle.
\eq
We denote by $V({\mathcal I})$ the variety defined by the ideal ${\mathcal I}$.
The dimension of the variety is given by
\bq
 \dim V({\mathcal I}) & = & \NV - r'.
\eq
We define the $\NV'$-dimensional skeleton of $V({\mathcal I})$ (or $\NV'$-skeleton for short)
to be the union of the varieties $V({\mathcal I}')$ with ${\mathcal I} \subset {\mathcal I}'$ 
(where ${\mathcal I}'$ defines a sub-localisation)
and dimension $\dim V({\mathcal I}') \le \NV'$.
\begin{figure}
\begin{center}
\includegraphics[scale=1.0]{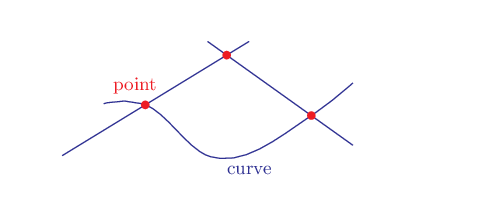}
\end{center}
\caption{
The $1$-skeleton of the example from fig.~\ref{fig:mixed_geometry}.
}
\label{fig:mixed_geometry_skeleton}
\end{figure}
As an example, we show in fig.~\ref{fig:mixed_geometry_skeleton} the $1$-skeleton of the geometry shown in fig.~\ref{fig:mixed_geometry}.

For each localisation, we consider integrands 
$\differentialform_{\mu_0 \dots \mu_{\ND}}[\tilde{Q},I,\Divisor_s]$,
which have a non-zero residue on $V({\mathcal I})$ and simple poles along ${\mathcal P}_1, \dots, {\mathcal P}_{r'}$,
together with the linear relations given by the integration-by-parts identities given for each chart of $V({\mathcal I})$ 
in eq. (\ref{ibp_localised}),
the distribution identities eq.~(\ref{eq_distribution}) and the cancellation identities eq.~(\ref{eq_cancellation_localised}).
Note that we are considering fewer integrands (as the integrands are required to have simple poles and 
a non-zero residue on $V({\mathcal I})$)
and fewer relations (as we only consider the integration-by-parts identities in $V({\mathcal I})$, but not the ones
orthogonal to it).
We may map the differential form $\differentialform_{\mu_0 \dots \mu_{\ND}}[\tilde{Q},I,\Divisor_s]$ (associated with the modified twist function $U^{\mathrm{loc}}$) to the differential form $\differentialform_{\mu_0 \dots \mu_{\ND}}[\tilde{Q}]$ (associated to the original twist function $U$)
by simply forgetting about $I$ and $\Divisor_s$ .
As before, we denote by $r$ the number of residues of $\differentialform_{\mu_0 \dots \mu_{\ND}}[Q]$,
and by $o$ its pole order in the full space ${\mathbb C}{\mathbb P}^{\NV}$.

For the Laporta algorithm, on each localisation, we use the order relation 
\bq
\label{def_laporta_order}
 \laportaorder,
\eq
where the dots stand for further criteria needed to distinguish inequivalent integrands.
In practice, the dots stand for an order on the indices $(\mu_0 \dots \mu_{\ND})$ 
(where the sum of the indices is fixed to $\absmu$)
followed by an order on the numerator polynomial $Q$ (where the degree of the polynomial $Q$ is fixed according to eq.~(\ref{def_d_Q})).
We use the string of integers in eq.~(\ref{def_laporta_order}) together with lexicographic ordering,
i.e. the relation $a_1 < a_2$ implies $\differentialform_1 < \differentialform_2$, with ties broken by $w$, etc.. 
On the localisation $V(I)$ with $\dim V(I)=\NV'$
the variable $a$ is given by
\bq
 a & = & \left\{
 \begin{array}{rl}
 -r-\NV, & \differentialform_{\mu_0 \dots \mu_{\ND}}[Q] \; \mbox{is a preferred candidate from the $(\NV'-1)$-skeleton of $V(I)$},\\
 0, & \mbox{otherwise.} \\
 \end{array}
 \right.
 \nonumber \\
\eq
The purpose of the variable $a$ is to give preference to the master integrands from the 
sub-problems as master integrands of the current problem.
This does not imply that all master integrands from the sub-problems will be master integrands
of the current problem.
As we have more integration-by-parts relations in the current problem, it may occur that some of the master integrands
from the sub-problems are eliminated.

The variable $w$ is defined by $w=\NV+r$. It corresponds to the Hodge weight of the differential form.
The purpose of the variable $w$ is the following: If we run the Laporta algorithm at weight $w$, we already have a basis of master integrands at weight $w'>w$.
These are the ones to which we have assigned $a<0$. Therefore, the basis at higher weights $w'$ is already complete, and we don't want to add any new master integrands with weights $w'>w$.
The use of the variable $w$ as the second order criterion prohibits this and forces the algorithm to choose the additional master integrands from the integrands of weight $w$.

As the third order criterion, we use the pole order $o$. 
We are now at a fixed weight $w$, and we are considering a square-free geometry. 
The algorithm will start to fill up the remaining master integrands by increasing pole order, i.e. the algorithm will first try to pick holomorphic differential forms, followed by the ones with pole order $1$, etc..
This pattern is compatible with the previously known example of the equal-mass $l$-loop banana integral \cite{Pogel:2022vat}.
In this example, the square-free geometry is a Calabi-Yau $(l-1)$-fold. 
The algorithm will start with picking the holomorphic integrand as a master integrand. 
There will be one master integrand for each pole order, up to pole order $(l-1)$.

As the fourth order criterion, we use $\absmu$, giving preference to objects with smaller $\absmu$.

The top-level algorithm to determine the master integrands on a localisation specified by the ideal ${\mathcal I}$
is as follows\footnote{The case ${\mathcal I}=\emptyset$ corresponds to the full system on the maximal cut.}:

\begin{tcolorbox}[breakable]
\refstepcounter{algocounter}
\label{algo:construct_masters}
{\bf Algorithm \thealgocounter: Construction of the masters on a localisation}

\begin{description}

\item{\bf Input:}
The ordered sequence ${\mathcal P}_1, \dots, {\mathcal P}_{r'}$.

\item{\bf Output:}
Assign a value $a<0$ to all master integrands.

\item{\bf Implementation:}
\begin{enumerate}
\item If $\dim V({\mathcal I}) = 0$ return the master integrands as described in section~\ref{sect:localisation_on_a_point} and assign $a=-2\NV$ to those.
\item If $\dim V({\mathcal I}) > 0$: Let $\NV'=\dim V({\mathcal I})$ and $w'=\NV+r'=2\NV-\NV'$.
\begin{enumerate}
\item Recursion: Consider first all sub-localisations of $V({\mathcal I})$ of dimension $(\NV'-1)$.
\item Merge procedure: Merge the master integrands from the $(\NV'-1)$-skeleton of $V({\mathcal I})$. 
\item Find additional masters at weight $w'$.
\end{enumerate}
\end{enumerate}
\end{description}
\end{tcolorbox}

In general, the variety $V({\mathcal I})$ of dimension $\NV'$ may have more than one sub-localisation of dimension $(\NV'-1)$.
In step 2 (b) we merge the master integrands from the sub-localisations of dimension $(\NV'-1)$.
This is necessary because we cannot simply take the union of the master integrands from the sub-localisations of dimension $(\NV'-1)$.
To see this, consider the following situation: Assume that in one sub-localisation we get one master integrand $\differentialform_1$, 
and that $\differentialform_2$ can be reduced to $\differentialform_1$ on this sub-localisation.
Assume further that in a second sub-localisation we obtain $\differentialform_2$ as a master integrand.
In this situation, the merge of the two sub-localisations should give us only one master integrand $\differentialform_1$,
as $\differentialform_1$ is simpler than $\differentialform_2$ according to the ordering criteria, and $\differentialform_2$ can be expressed in terms of 
$\differentialform_1$ on the first sub-localisation.
Note that in this situation, $\differentialform_2$ can be localised on the intersection of the two sub-localisations.

Technically, we do this as follows:
Let us assume that the $(\NV'-1)$-skeleton of $V({\mathcal I})$ is given by
\bq
 V({\mathcal I_1}) \cup V({\mathcal I_2}) \cup \dots \cup V({\mathcal I_t}),
\eq
where each ideal $I_i$ is defined by an ordered sequence ${\mathcal P}_1^{(i)}, \dots, {\mathcal P}_{r'}^{(i)}, {\mathcal P}_{r'+1}^{(i)}$ of $(r'+1)$ polynomials 
\bq
 I_i & = &
 \left\langle {\mathcal P}_1^{(i)}, \dots, {\mathcal P}_{r'}^{(i)}, {\mathcal P}_{r'+1}^{(i)} \right\rangle.
\eq
It may happen that different ordered sequences define the same ideal. 
As an example we have
\bq
 \left\langle z_1, z_2 \right\rangle
 \; = \;
 \left\langle z_2, z_1 \right\rangle
 \; = \;
 \left\langle z_1, z_1+z_2 \right\rangle
 \; = \;
 \left\langle z_1+z_2, z_1 \right\rangle
 \; = \;
 \left\langle z_2, z_1+z_2 \right\rangle
 \; = \;
 \left\langle z_1+z_2, z_2 \right\rangle.
\eq
If the intersection of $V({\mathcal I_{i}})$ and $V({\mathcal I_{j}})$ is non-empty and of dimension $(\NV'-2)$, we denote this intersection by
\bq
 V({\mathcal I_{ij}}) & = & V({\mathcal I_i}) \cap V({\mathcal I_j}).
\eq
The $(\NV'-2)$-skeleton of $V({\mathcal I})$ contains
the union of these intersections, which we write as a union of connected irreducible components $C_k$
\bq
 \bigcup\limits_{(ij)} V({\mathcal I_{ij}})
 & = &
 \bigcup\limits_{k} C_k.
\eq
We set
\bq
 S_k & = & \left\{ \; i \; | \; C_k \subset V\left(I_i\right) \; \right\}.
\eq
$S_k$ is the set of indices $i$ of varieties $V(I_i)$, which intersect at $C_k$.
As at least two varieties intersect at $C_k$ we have $|S_k| \ge 2$. 
We emphasise that we can have $|S_k| > 2$.
At this stage, we have already determined a basis of $C_k$,
coming from a localisation specified by an ordered sequence
${\mathcal P}_1^{(S_k)}, \dots, {\mathcal P}_{r'}^{(S_k)}, {\mathcal P}_{r'+1}^{(S_k)}, {\mathcal P}_{r'+2}^{(S_k)}$ of $(r'+2)$ polynomials.
We set
\bq
 I_{S_k} & = &
 \left\langle {\mathcal P}_1^{(S_k)}, \dots, {\mathcal P}_{r'}^{(S_k)}, {\mathcal P}_{r'+1}^{(S_k)}, {\mathcal P}_{r'+2}^{(S_k)} \right\rangle.
\eq
A basis element $\differentialform$
of $C_k$ will define an integrand on $V({\mathcal I_i})$
for at least one index $i \in S_k$, but it may happen that $\differentialform$ does not define an integrand
for all $i \in S_k$.
It may happen that there are $j \in S_k$ such that
$\differentialform$ does not have a residue on $V({\mathcal I_j})$. 
We will later see examples where this happens in section~\ref{sect:sector_93_moeller} and section~\ref{sect:electron_self_energy}.

The master integrands of $C_k$ provide for $i,j \in S_k$
the communication between $V({\mathcal I_i})$ and $V({\mathcal I_j})$, provided they have non-zero
residues on $V({\mathcal I_i})$ and $V({\mathcal I_j})$.
We obtain the following algorithm:
\begin{tcolorbox}[breakable]
\refstepcounter{algocounter}
{\bf Algorithm \thealgocounter: Merge procedure}

\begin{description}

\item{\bf Input:}
The ordered sequence ${\mathcal P}_1, \dots, {\mathcal P}_{r'}$.

\item{\bf Output:}
Let $w'=\NV+r'$.
Assign a value $a=0$ to all integrands of weight $w''>w'$, which are not masters at weight $w'$.

\item{\bf Implementation:}
\begin{enumerate}
\item Loop over all connected components $C_k$ of dimension $(\NV'-2)$ (as defined above).
\begin{enumerate}
\item If $\mathrm{Res}_{{\mathcal P}_1^{(S_k)}, \dots, {\mathcal P}_{r'+2}^{(S_k)}} \differentialform_{\mu_0 \dots \mu_{\ND}}[\tilde{Q},{\mathcal I_{S_k}},\Divisor_{s_{S_k}}]$
is a master integrand of $C_k$, $i,j \in S_k$ and both
$\mathrm{Res}_{{\mathcal P}_1^{(i)}, \dots, {\mathcal P}_{r'+1}^{(i)}} \differentialform_{\mu_0 \dots \mu_{\ND}}[\tilde{Q},{\mathcal I_{i}},\Divisor_{s_i}]$
and 
$\mathrm{Res}_{{\mathcal P}_1^{(j)}, \dots, {\mathcal P}_{r'+1}^{(j)}} \differentialform_{\mu_0 \dots \mu_{\ND}}[\tilde{Q},{\mathcal I_{j}},\Divisor_{s_j}]$
are non-zero,
add the equation
\bq
 \mathrm{Res}_{{\mathcal P}_1^{(i)}, \dots, {\mathcal P}_{r'+1}^{(i)}} \differentialform_{\mu_0 \dots \mu_{\ND}}[\tilde{Q},{\mathcal I_{i}},\Divisor_{s_i}]
 & = & 
 \mathrm{Res}_{{\mathcal P}_1^{(j)}, \dots, {\mathcal P}_{r'+1}^{(j)}} \differentialform_{\mu_0 \dots \mu_{\ND}}[\tilde{Q},{\mathcal I_{j}},\Divisor_{s_j}].
 \nonumber \\
\eq
\item If $\mathrm{Res}_{{\mathcal P}_1^{(i)}, \dots, {\mathcal P}_{r'+1}^{(i)}} \differentialform_{\mu_0 \dots \mu_{\ND}}[\tilde{Q},{\mathcal I_{i}},\Divisor_{s_i}]$ is not a master integrand of 
$V({\mathcal I_{i}})$, add the equation from the reductions on $V({\mathcal I_{i}})$, which expresses this differential form in terms of the masters
of $V({\mathcal I_{i}})$.
\end{enumerate}

\item Run the Laporta algorithm on this linear system. 

\item Substitute $\mathrm{Res}_{{\mathcal P}_1^{(i)}, \dots, {\mathcal P}_{r'+1}^{(i)}} \differentialform_{\mu_0 \dots \mu_{\ND}}[\tilde{Q},{\mathcal I_{i}},\Divisor_{s_i}]$ by
$\differentialform_{\mu_0 \dots \mu_{\ND}}[\tilde{Q}]$.

\item Set $a=0$ for all non-masters.

\end{enumerate}
\end{description}
\end{tcolorbox}
An implicit assumption for this algorithm is that the equations added in step 1 (b) are generic, i.e. the system has full rank.
A rank drop should throw an exception.

Remark: If $\NV'=\dim V(I)=1$, the $(\NV'-2)$-skeleton of $V(I)$ is empty, and the merge procedure does nothing in this case.

Finally, we specify the sub-routine ``Find additional masters'' appearing as step 2 (c) in algorithm~\ref{algo:construct_masters}:
\begin{tcolorbox}[breakable]
\refstepcounter{algocounter}
{\bf Algorithm \thealgocounter: Find additional masters}

\begin{description}

\item{\bf Input:}
The ordered sequence ${\mathcal P}_1, \dots, {\mathcal P}_{r'}$.

\item{\bf Output:}
Let $w'=\NV+r'$.
Assign the value $a=-w'$ to the additional masters at weight $w'$.

\item{\bf Implementation:}
\begin{enumerate}
\item Generate the linear relations for this localisation.
\item Perform the Gau{\ss} elimination.
\item Set $a=-w'$ for all masters, which have $a=0$.
\item Set $a=0$ for all non-masters.
\end{enumerate}
\end{description}
\end{tcolorbox}

Running algorithm~\ref{algo:construct_masters} with the order criterion of eq.~(\ref{def_laporta_order}) will give a basis of 
the vector space~$\Hgen^{\NV}_{\omega}$.
Using the homomorphism between the vector spaces $V^{\NV}$ and $\Hgen^{\NV}_{\omega}$ from section~\ref{sect:conversion}
we then obtain
a basis of master integrals on the maximal cut.
We may do this for all sectors on the maximal cut bottom-up.

At the level of Feynman integrals and including sub-sectors, we then run for the reduction to master integrals the Laporta algorithm on the full system
with the order relation
\bq
\label{full_order_relation}
 \left( \hat{\nedges}, \hat{N}_{\mathrm{id}}, a,w,o,|\mu|,\dots \right).
\eq
where
\bq
 \left( \hat{\nedges}, \hat{N}_{\mathrm{id}} \right)
 & = &
 \left\{ \begin{array}{ll}
  \left( \nedges', N_{\mathrm{id}}' \right) & \mbox{if $N_{\mathrm{id}}'$ is the smallest sector such that $N_{\mathrm{id}}$ is a super-sector of $N_{\mathrm{id}}'$.} \\
  \left( \nedges, N_{\mathrm{id}} \right) & \mbox{otherwise}. \\
 \end{array}
 \right.
 \nonumber 
\eq
This will yield a basis of master integrals $J=(J_1,J_2,\dots)^T$.
The integers $a,w,o,|\mu|$ for Feynman integrals are determined from the images 
under the homomorphism $\hat{\iota} : A^{\NV}_{\mathrm{all}} \rightarrow \Agen^{\NV}_{\omega}$.
This treatment completely avoids relative twisted cohomology.
Note that the determination of the integers $a,w,o,|\mu|$ needs only to be done once at initialisation
time and only on the maximal cut. No computation beyond the maximal cut is required.
The dots stand for further criteria needed to distinguish inequivalent integrals in the full system.
More concretely, the dots stand for an order on the indices $(\mu_0 \dots \mu_{\ND})$,
followed by an order on the numerator polynomial $Q$,
followed by an order relation on $I_{\nu_1 \dots \nu_{\NE}}$.
With the exception of the last criterion, all other order variables are obtained from
the differential form $\differentialform=\hat{\iota}(I_{\nu_1 \dots \nu_{\NE}})$.
The last order criterion is needed to distinguish inequivalent integrals, which have the same image
$\differentialform=\hat{\iota}(I_{\nu_1 \dots \nu_{\NE}})$. This can be due to sub-sectors (which are not seen on the maximal cut)
or ``missing ISPs''.
The last order criterion enters also as an order criterion on $A^{\NV}_{\mathrm{all}}$ in section~\ref{sect:conversion}.
For the last criterion a standard order relation on Feynman integrals can be used.

This order relation performs significantly better than the standard order relation.
We observe that it eliminates to a large extent the occurrence of spurious polynomials in the denominator.
This holds beyond the maximal cut, and the most drastic improvements are in the coefficients 
of the simplest master integrals
in the reduction of the most complicated sectors.


\subsection{Differential equations, step 1: The basis $J$}
\label{sect:step_1}

In this section we consider a sector of the basis $J$, obtained as described in section~\ref{sect:order_relation},
on the maximal cut.
We are in particular interested in the differential equation on the maximal cut.
We denote by $\differentialform=(\differentialform_1,\dots,\differentialform_{\NF})^T$
the basis of master integrands in $\Hgen^{\NV}_{\omega}$.
In all examples we tested, we observed that if $\differentialform_i \in \mathrm{Gr}_{\Fcomb}^{\NV-\absmu_i} \Agen^{\NV}_\omega$
and $\differentialform_j \in \mathrm{Gr}_{\Fcomb}^{\NV-\absmu_j} \Agen^{\NV}_\omega$, then
\bq
 d_B J\left(\eps,x\right) & = & A\left(\eps,x\right) J\left(\eps,x\right),
\eq
with
\bq
\label{refined_statement}
 A_{ij}\left(\eps,x\right)
 & = &
 \sum\limits_{k=-(\absmu_i-\absmu_j)}^1
 \eps^k A^{(k)}_{ij}\left(x\right).
\eq
We call a differential equation which satisfies eq.~(\ref{refined_statement}) an $\Fgen^\bullet$-compatible differential equation for the filtration $\Fcomb^\bullet$.
An $\Fgen^\bullet$-compatible differential equation implies Griffiths transversality. 
It is, however, a stronger statement, as it requires the differential equation to be in Laurent polynomial form with restrictions on the occurring powers of $\eps$.
We recall that Griffiths transversality \cite{Griffiths:1968i,Griffiths:1968ii} is the statement that $A_{ij}=0$ for $\absmu_j > \absmu_i + 1$.

It is worth elaborating on this observation:
We may split the linear relations among the objects $\differentialform_{\mu_0 \dots \mu_{\ND}}[Q]$
into a set $A_1$, consisting of the integration-by-parts identities and the distribution identities 
and a set $A_2$, consisting of the cancellation identities.
We may run the Laporta algorithm first on the set $A_1$.
For this task, we may set $\eps=1$. 
This is a significant efficiency improvement, as we have one variable less.
This is possible, because in the integration-by-parts identities eq.~(\ref{eq_ibp}) 
and the distribution identities eq.~(\ref{eq_distribution})
the explicit $\eps$-factors are synchronised with $\absmu$.
At the end, we may restore the $\eps$-dependence from $\absmu$ as follows:
Suppose we get for $\eps=1$ a relation like
\bq
 \differentialform_{\mu_0 \dots \mu_{\ND}}[Q]
 & = &
 c_1 \differentialform_{\mu_0' \dots \mu_{\ND}'}[Q']
 +
 c_2 \differentialform_{\mu_0'' \dots \mu_{\ND}''}[Q''].
\eq
Restoring $\eps$ we obtain
\bq
 \differentialform_{\mu_0 \dots \mu_{\ND}}[Q]
 & = &
 c_1
 \eps^{\left|\mu'\right|-\left|\mu\right|}
 \differentialform_{\mu_0' \dots \mu_{\ND}'}[Q']
 +
 c_2
 \eps^{\left|\mu''\right|-\left|\mu\right|}
 \differentialform_{\mu_0'' \dots \mu_{\ND}''}[Q''].
\eq
The same argument also shows that in reducing the system $A_1$, the $\eps$-dependence of the coefficients
is always monomial.

We then merge the reduced system $A_1$ with the system $A_2$.
The system $A_2$ consists entirely of cancellation identities, i.e. reduction identities for
objects $\differentialform_{\mu_0 \dots (\mu_j+1) \dots \mu_{\ND}}[\Divisor_j \cdot Q]$.
Some of these objects are masters in the reduced system $A_1$.
Substituting the cancellation identities for these objects will reduce them. 
The $\eps$-dependence of the coefficients in the linear system will no longer be monomial,
but a Laurent polynomial in $\eps$.
This is due to the factor in eq.~(\ref{linear_factor_in_eps}).

Other objects $\differentialform_{\mu_0 \dots (\mu_j+1) \dots \mu_{\ND}}[\Divisor_j \cdot Q]$ will not be masters
in the reduced system $A_1$.
In this case, there are two possibilities:
Either the integration-by-parts identities reduce the entire cancellation identity for
$\differentialform_{\mu_0 \dots (\mu_j+1) \dots \mu_{\ND}}[\Divisor_j \cdot Q]$ to $0=0$, this is unproblematic.
The second possibility is that this does not happen.
In this case, the cancellation identity will eliminate a lower object.
However, as the lower objects are multiplied by the linear factor of eq.~(\ref{linear_factor_in_eps}), 
we no longer
can guarantee the Laurent polynomial form of the $\eps$-dependence of the coefficients
and in general the $\eps$-dependence of the coefficients will be rational in $\eps$.

For the differential equation related to eq.~(\ref{refined_statement})
we observe that all reductions required to express $d_B J$ in terms of the masters $J$ 
are Laurent polynomials in $\eps$.
The reductions of other objects may (and actually do) involve rational functions in $\eps$.


\subsection{Differential equations, step 2: The basis $K$}
\label{sect:step_2}

In this section, we show that we may always construct an $\eps$-factorised differential equation 
from an $\Fgen^\bullet$-compatible differential equation.
We remark that in this section, we only require the differential equation to be compatible with some filtration $\Fgen^\bullet$, it does
not have to be the $\Fcomb^\bullet$-filtration. 

Let $J=(J_1,\dots,J_{\NF})^T$ be a basis with an $\Fgen^\bullet$-compatible differential equation
and assume that $J$ is ordered according to the $\Fgen^\bullet$-filtration, i.e.
$J_1 \in \Fgen^{p_{\max}} V^{\NV}$ and $J_{\NF} \in \Fgen^{p_{\min}} V^{\NV}$.
The assumptions imply that the differential equation for the basis $J$ reads
\bq
 d_B J\left(\eps,x\right) & = & A\left(\eps,x\right) J\left(\eps,x\right),
 \;\;\;\;\;\;
 A\left(\eps,x\right) \; = \;
 \sum\limits_{k=-\NV}^1
 \eps^k
 A^{(k)}\left(x\right).
\eq
The matrix $A$ has then a block structure induced by the $\Fgen^\bullet$-filtration.
It will be convenient to organise the matrix $A$ as 
\bq
\label{A_reorganised}
 A & = & 
 \sum\limits_{k=-\NV}^1
 B^{(k)}\left(x\right),
\eq
with $B^{(1)}(x)=\eps A^{(1)}(x)$. 
For $k<1$ the matrices $B^{(k)}(x)$ are lower block-triangular.
The blocks on the lower $j$-th block sub-diagonal are given by the terms of order $\eps^{\NV+k-j}$ of 
the corresponding blocks of $A$.
We say that a term is of $B$-order $k$ if the term appears in $B^{(k)}$.

To give an example, consider the matrix $A$ with three $1\times 1$ blocks:
\bq
 A
 & = &
 \left( \begin{array}{rrr}
 A^{(0)}_{11} + \eps A^{(1)}_{11} & \eps A^{(1)}_{12} & 0 \\ 
 \frac{1}{\eps} A^{(-1)}_{21} + A^{(0)}_{21} + \eps A^{(1)}_{21} & A^{(0)}_{22} + \eps A^{(1)}_{22} & \eps A^{(1)}_{23}  \\ 
 \frac{1}{\eps^2} A^{(-2)}_{31} + \frac{1}{\eps} A^{(-1)}_{31} + A^{(0)}_{31} + \eps A^{(1)}_{31} & \frac{1}{\eps} A^{(-1)}_{32} + A^{(0)}_{32} + \eps A^{(1)}_{32} & A^{(0)}_{33} + \eps A^{(1)}_{33}  \\ 
 \end{array} \right).
\eq
Then
\bq
 B^{(1)}
 &= &
 \left( \begin{array}{rrr}
 \eps A^{(1)}_{11} & \eps A^{(1)}_{12} & 0 \\ 
 \eps A^{(1)}_{21} & \eps A^{(1)}_{22} & \eps A^{(1)}_{23}  \\ 
 \eps A^{(1)}_{31} & \eps A^{(1)}_{32} & \eps A^{(1)}_{33}  \\ 
 \end{array} \right),
 \nonumber \\
 B^{(0)}
 & = &
 \left( \begin{array}{ccc}
 0 & 0 & 0 \\ 
 0 & 0 & 0  \\ 
 A^{(0)}_{31} & 0 & 0   \\ 
 \end{array} \right),
 \nonumber \\
 B^{(-1)}
 & = &
 \left( \begin{array}{ccc}
 0 & 0 & 0 \\ 
 A^{(0)}_{21} & 0 & 0  \\ 
 \frac{1}{\eps} A^{(-1)}_{31} & A^{(0)}_{32} & 0   \\ 
 \end{array} \right),
 \nonumber \\
 B^{(-2)}
 & = &
 \left( \begin{array}{ccc}
 A^{(0)}_{11} & 0 & 0 \\ 
 \frac{1}{\eps} A^{(-1)}_{21} & A^{(0)}_{22} & 0  \\ 
 \frac{1}{\eps^2} A^{(-2)}_{31} & \frac{1}{\eps} A^{(-1)}_{32} & A^{(0)}_{33}   \\ 
 \end{array} \right).
\eq
The differential equation is in $\eps$-factorised form if all terms of $B$-order $k<1$ are absent.
Below we present an algorithm, which iteratively first removes terms of $B$-order $(-\NV)$, then terms
of $B$-order $(-\NV+1)$, etc., until we arrive at an $\eps$-factorised form.
We consider a transformation 
\bq
\label{trafo_J_to_K}
 K & = & R_2^{-1} J,
\eq
such that the new basis $K$ satisfies an $\eps$-factorised differential equation (on the maximal cut).
It is convenient to use the inverse matrix $R_2^{-1}$ in eq.~(\ref{trafo_J_to_K}).
Obviously, we have $J = R_2 K$. 
The matrix $R_2$ is given as
\bq
 R_2 & = & R_2^{(-\NV)} R_2^{(-\NV+1)} \dots R_2^{(-1)} R_2^{(0)},
\eq
where $R_2^{(k)}$ removes terms of $B$-order $k$.
All matrices $R_2^{(k)}$ are lower block-triangular.
The matrix $R_2^{(-\NV)}$ is of $B$-order $(-\NV)$,
the matrices $R_2^{(k)}$ with $-\NV < k \le 0$ are given by
\bq
 R_2^{(k)}
 & = &
 {\bf 1} + T_2^{(k)},
\eq
where $T_2^{(k)}$ is of $B$-order $k$, and
${\bf 1}$ denotes the $\NF \times \NF$ unit matrix.

\begin{lemma}[Compatibility of the $B$-order with matrix multiplication]
\label{lemma_B_order_multiplication}
Let $A$ be a matrix with terms of $B$-order $\{k,k+1,\dots,1\}$ (with $k \ge -\NV$) and $R_2^{(l)}$ with $l \in \{k,k+1,\dots,0\}$ as above.
Then the products $A R_2^{(l)}$ and $R_2^{(l)} A$ have again only terms of $B$-order $\{k,k+1,\dots,1\}$.
\end{lemma}
\begin{proof}
The multiplication with the unit matrix ${\bf 1}$ does not change the $B$-order of $A$, hence it is sufficient
to consider the multiplication of $A$ with $R_2^{(-\NV)}, T_2^{(-\NV+1)}, \dots, T_2^{(0)}$.
Thus, we need to show that the multiplication of a matrix $A$ with terms of $B$-order $\{k,k+1,\dots,1\}$
with a matrix $X$ of $B$-order $l$ (with $l \in \{k,k+1,\dots,0\}$) results again 
in a matrix with terms of $B$-order $\{k,k+1,\dots,1\}$.
This follows directly from matrix multiplication.
We index the blocks of the matrices $A$ and $X$ by $i$ and $j$.
The $\eps$-order of the blocks is
\bq
 A_{ij}
 & \in &
 \left\{ \begin{array}{ll}
 \left\{ {\mathcal O}\left(\eps^{\min(1,\NV+k+j-i)}\right), \dots, {\mathcal O}\left(\eps^{1}\right) \right\}, & j-i\le 1, \\
 0, & j-i>1.
 \end{array} \right.
 \nonumber \\
 X_{ij}
 & = &
 \left\{ \begin{array}{ll}
 {\mathcal O}\left(\eps^{\NV+l+j-i}\right), & \NV+l+j-i \le 0, \\
 0, & \NV+l+j-i >0.
 \end{array} \right.
 \nonumber \\
\eq
We have
\bq
 \sum\limits_{m=1}^{\NV} A_{im} X_{mj}
 & = &
 \sum\limits_{m=\max(1,\NV+l+j)}^{\min(\NV,i+1)} A_{im} X_{mj}
 \; \in \;
 \left\{ {\mathcal O}\left(\eps^{\min(1,\NV+l+\NV+k+j-i)}\right), \dots, {\mathcal O}\left(\eps^{1}\right) \right\}.
\eq
As $\NV+l \ge 0$ we have shown the claim for $A R_2^{(l)}$.
In a similar way, we have
\bq
 \sum\limits_{m=1}^{\NV} X_{im} A_{mj}
 & = &
 \sum\limits_{m=\max(1,j-1)}^{\min(\NV,i-\NV-l)} X_{im} A_{mj}
 \; \in \;
 \left\{ {\mathcal O}\left(\eps^{\min(1,\NV+l+\NV+k+j-i)}\right), \dots, {\mathcal O}\left(\eps^{1}\right) \right\}.
\eq
This shows the claim for $R_2^{(l)} A$.
\end{proof}

\begin{corollary}[Multiplication of matrices of $B$-order $k \le 0$]
Let $S$ be a matrix of $B$-order $k_1$ and let $T$ be a matrix of $B$-order $k_2$.
If $k_1, k_2 \in \{-\NV,\dots,0\}$ then the product $S \cdot T$ has $B$-order $(\NV+k_1+k_2)$. 
\end{corollary}
\begin{proof}
The proof follows along the lines of lemma~\ref{lemma_B_order_multiplication}.
\end{proof}

\begin{lemma}[Inverse matrix]
\label{lemma_inverse_matrix}
Let $R_2^{(k)}$ with $k\in\{-\NV,\dots,0\}$ as above. 
If $k=-\NV$, the inverse matrix $(R_2^{(-\NV)})^{-1}$ is again of $B$-order $(-\NV)$, otherwise for $k\in\{-\NV+1,\dots,0\}$ the inverse matrix can be written as
\bq
\label{R2_inverse_k}
 \left(R_2^{(k)}\right)^{-1}
 & = & {\bf 1} + \sum\limits_{j=k}^{0} S_2^{(j)},
\eq
with $S_2^{(j)}$ of $B$-order $j$.
\end{lemma}
\begin{proof}
We first consider the case $k=-\NV$.
We write $R_2^{(-\NV)}=D_2^{(-\NV)}+N_2^{(-\NV)}$, where $D_2^{(-\NV)}$ is block-diagonal and invertible.
The matrix $N_2^{(-\NV)}$ is nilpotent, with
$(N_2^{(-\NV)})^{\NV+1}={\bf 0}$.
We have
\bq
 \left( R_2^{(-\NV)}\right)^{-1} & = & \sum\limits_{j=0}^{\NV} \left(-\left(D_2^{(-\NV)}\right)^{-1} N_2^{(-\NV)} \right)^j \left(D_2^{(-\NV)}\right)^{-1}.
\eq
Multiplication with the block-diagonal matrix $(D_2^{(-\NV)})^{-1}$ does not change the $B$-order.
The expression $(-(D_2^{(-\NV)})^{-1} N_2^{(-\NV)} )^j$ is again of $B$-order $(-\NV)$.

We now consider the case $k\in\{-\NV+1,\dots,0\}$. We have $R_2^{(k)} = {\bf 1} + T_2^{(k)}$, where $T_2^{(k)}$ is nilpotent. 
The inverse is given by
\bq
 \left( R_2^{(k)}\right)^{-1}
 & = &
 {\bf 1}
 +
 \sum\limits_{j=1}^{\left\lfloor\frac{\NV}{\NV+k}\right\rfloor} \left(- T_2^{(k)} \right)^j.
\eq
The $B$-order of $(- T_2^{(k)} )^j$ is $j(\NV+k)-\NV$.
\end{proof}

\begin{corollary}[Transformation by $R_2^{(k)}$]
\label{corollary_B_order_multiplication}
Let $A$ be a matrix with terms of $B$-order $\{k,k+1,\dots,1\}$ and $R_2^{(k)}$ as above.
Then
\bq
 \left(R_2^{(k)}\right)^{-1} A R_2^{(k)} - \left(R_2^{(k)}\right)^{-1} d_B R_2^{(k)}
\eq
has only terms of $B$-order $\{k,k+1,\dots,1\}$.
\end{corollary}
\begin{proof}
This is a direct consequence of lemmata~\ref{lemma_B_order_multiplication} and \ref{lemma_inverse_matrix}.
\end{proof}

\begin{theorem}[$\eps$-factorisation]
\label{factorisation_theorem}
Let $J=(J_1,\dots,J_{\NF})^T$ be a basis for a given sector, which satisfies (on the maximal cut) 
a differential equation compatible with a filtration $\Fgen^\bullet$.
Then there exists a transformation $K = R_2^{-1} J$ with $R_2=R_2^{(-\NV)} R_2^{(-\NV+1)} \dots R_2^{(-1)} R_2^{(0)}$, where the 
matrices $R_2^{(k)}$ are in the form as above, such that the differential equation is in $\eps$-factorised form.
\end{theorem}

\begin{proof}
We construct the matrices $R_2^{(k)}$ iteratively, starting from $k=-\NV$ and ending with $k=0$.
We set $\tilde{A}^{(-\NV)}=A$ and
\bq
 \tilde{A}^{(k+1)}
 =
 \left(R_2^{(k)}\right)^{-1} \tilde{A}^{(k)} R_2^{(k)} - \left(R_2^{(k)}\right)^{-1} d_B R_2^{(k)}.
 \;
\eq
The matrix $R_2^{(k)}$ is determined by
\bq
\label{eq_U_k}
 \left. \left[ \left(R_2^{(k)}\right)^{-1} \tilde{A}^{(k)} R_2^{(k)} - \left(R_2^{(k)}\right)^{-1} d_B R_2^{(k)} \right] \right|_{k}
 = 0,
 \;
\eq
where $|_k$ indicates that only terms of $B$-order $k$ are taken.
Eq.~(\ref{eq_U_k}) defines an $\eps$-independent system of first-order differential equations for the unknown functions in the ansatz for $R_2^{(k)}$.
There are as many equations as there are unknown functions.
Eq.~(\ref{eq_U_k}) also ensures that $\tilde{A}^{(k+1)}$ only has terms of $B$-order $\{k+1,\dots,1\}$.
Thus, in every iteration step, we improve the $B$-order.
After transformation with $R_2^{(0)}$, the matrix $\tilde{A}^{(1)}$ only has terms of $B$-order $1$.
\end{proof}
A few comments are in order: At first sight it may seem that eq.~(\ref{eq_U_k}) yields a non-linear system of differential
equations. This is not the case.
In fact, eq.~(\ref{eq_U_k}) decomposes into a block-triangular system of first-order linear inhomogeneous differential equations.
In practice, we solve this system as follows:
We first determine the functions appearing in $R_2^{(-\NV)}$, followed by the functions appearing in $R_2^{(-\NV+1)}$.
We continue in this way until we reach $R_2^{(0)}$.
The differential equations for the functions appearing in $R_2^{(k)}$ will not depend on the functions appearing in $R_2^{(j)}$
with $j>k$.
Furthermore, we may assume that all functions appearing in $R_2^{(j)}$ with $j<k$ are already known at this stage.
In addition, the differential equations for the functions appearing in $R_2^{(k)}$ will form a block-triangular system
according to the block-columns.
The functions appearing in block-column $l$ will not be influenced by functions from the block-columns $l'>l$ and
functions from block-columns $l'<l$ can be considered to be known at this stage.
The net result is that the system decomposes into smaller systems of first-order linear inhomogeneous differential equations.

Note that we cannot expect that the system uncouples into single first-order linear inhomogeneous differential equations.
For example, in the case of an elliptic Feynman integral, we will obtain the differential equations describing the variation
of the periods of the elliptic curve with the kinematic variables.
This will always be a coupled system of two first-order differential equations.
Of course, the variation along a given path can be converted (if desired) into a single second-order differential equation.

We may simplify eq.~(\ref{eq_U_k}), depending on whether we consider the case $k=-\NV$ or the case $k\in\{-\NV+1,\dots,0\}$.
We have the following two lemmata:
\begin{lemma}[Differential equation for the rotation matrix $R_2^{(-\NV)}$]
\label{lemma_diff_eq_R_2_n}
The rotation matrix $R_2^{(-\NV)}$ is determined from the differential equation
\bq
 d_B R_2^{(-\NV)} 
 & = & 
 B^{(-\NV)} R_2^{(-\NV)} 
 +
 R_2^{(-\NV)} \left. \left[ \left(R_2^{(-\NV)}\right)^{-1} B^{(1)} R_2^{(-\NV)} \right] \right|_{(-\NV)}.
\eq
\end{lemma}
\begin{proof}
The matrices $(R_2^{(-\NV)})^{-1}$ and $d R_2^{(-\NV)}$ are both of $B$-order $(-\NV)$, hence we may drop the $\left.\right|_{-\NV}$-prescription
in the second term of eq.~(\ref{eq_U_k}).
We have 
\bq
 \tilde{A}^{(-\NV)} \; = \; A \; = \; \sum_{k=-\NV}^1 B^{(k)}.
\eq
Only the terms with $B^{(-\NV)}$ and $B^{(1)}$ will contribute. In the term with $B^{(-\NV)}$ we may again drop the $\left.\right|_{-\NV}$-prescription.
\end{proof}

\begin{lemma}[Differential equation for the rotation matrix $R_2^{(k)}$ for $k>-\NV$]
\label{lemma_diff_eq_R_2_k}
The rotation matrix $R_2^{(k)}$ for $k \in \{-\NV+1,\dots,0\}$ is determined from the differential equation
\bq
 d_B R_2^{(k)}
 & = &
 \left. \left[ \left(R_2^{(k)}\right)^{-1} \tilde{A}^{(k)} R_2^{(k)} \right] \right|_{k}.
\eq
\end{lemma}
\begin{proof}
We consider the term $\left. (R_2^{(k)})^{-1} d_B R_2^{(k)} \right|_{k}$. The matrix $d_B R_2^{(k)}$ is of $B$-order $k$. 
The inverse matrix $(R_2^{(k)})^{-1}$ has the structure as given in eq.~(\ref{R2_inverse_k}). With corollary~\ref{corollary_B_order_multiplication} it follows that only the
identity matrix in eq.~(\ref{R2_inverse_k})
contributes to the final $B$-order $k$.
\end{proof}

\begin{myexample}
\label{example_two_blocks}
We illustrate the block-triangular structure of the differential equations for an example.
We consider the case 
where we have two non-trivial parts in the filtration
\bq
 \emptyset = F^2 V \subseteq F^1 V \subseteq F^0 V = V.
\eq
We set
\bq
 d_1 \; = \; \dim \mathrm{Gr}_F^1 V,
 \;\;\;\;\;\;
 d_2 \; = \; \dim \mathrm{Gr}_F^0 V.
\eq
The matrix $A$ appearing in the differential equation $d_BJ = A J$ is then of the form
\bq
 A
 & = &
 \left( \begin{array}{rr}
 A^{(0)}_{11} + \eps A^{(1)}_{11} & \eps A^{(1)}_{12} \\
 \frac{1}{\eps} A^{(-1)}_{21} + A^{(0)}_{21} + \eps A^{(1)}_{21} & A^{(0)}_{22} + \eps A^{(1)}_{22} \\
 \end{array} \right),
\eq
where the entries are in general matrices.
We have $\dim A^{(j)}_{11} = d_1$ and $\dim A^{(j)}_{22} = d_2$.
The off-diagonal entries have the dimensions dictated by the diagonal blocks.
We construct a transformation
\bq
 R_2 & = & R_2^{(-1)} R_2^{(0)},
\eq
leading to the $\eps$-factorised basis $K=R_2^{-1} J$.
The ansatz for $R_2^{(-1)}$ reads
\bq
 R_2^{(-1)}
 & = &
 \left( \begin{array}{rr}
 R^{(0)}_{11} & 0 \\
 \frac{1}{\eps} R^{(-1)}_{21} & R^{(0)}_{22} \\
 \end{array} \right).
\eq
For a basis transformation the determinant of $R_2^{(-1)}$ has to be non-zero, this implies that 
$R^{(0)}_{11}$ and $R^{(0)}_{22}$ are invertible matrices.

Requiring that terms of $B$-order $(-1)$ vanish gives us three equations.
These equations group into $3=2+1$ as follows:
The first group of two equations involves only $R^{(0)}_{11}$ and $R^{(0)}_{21}$:
\bq
 d_B R^{(0)}_{11} & = & A^{(0)}_{11} R^{(0)}_{11} + A^{(1)}_{12} R^{(-1)}_{21},
 \nonumber \\
 d_B R^{(-1)}_{21} & = & A^{(-1)}_{21} R^{(0)}_{11} + A^{(0)}_{22} R^{(-1)}_{21}.
\eq
The second set involves in addition $R^{(0)}_{22}$:
\bq
 d_B R^{(0)}_{22} & = & A^{(0)}_{22} R^{(0)}_{22} - R^{(-1)}_{21} \left( R^{(0)}_{11} \right)^{-1} A^{(1)}_{12} R^{(0)}_{22}.
\eq
We then turn to $R_2^{(0)}$. The ansatz for $R_2^{(0)}$ reads
\bq
 R_2^{(0)}
 & = &
 \left( \begin{array}{rr}
 1 & 0 \\
 R^{(0)}_{21} & 1 \\
 \end{array} \right).
\eq
Requiring that terms of $B$-order $0$ vanish, gives us one equation:
\bq
 d_B R^{(0)}_{21} & = & \left( R^{(0)}_{22} \right)^{-1} \left( A^{(0)}_{21} R^{(0)}_{11} + A^{(1)}_{22} R^{(-1)}_{21}   
                   - R^{(-1)}_{21} \left( R^{(0)}_{11} \right)^{-1} A^{(1)}_{11} R^{(0)}_{11} \right).
\eq
A similar, but slightly more complicated example with
three non-trivial parts in the filtration can be found in appendix B of ref.~\cite{Pogel:2025bca}.
\end{myexample}
 
This completes the construction of the $\eps$-factorised differential equation on the maximal cut.
The extension beyond the maximal cut is straightforward: Any offending term is strictly lower block triangular
and can be removed with an ansatz similar to $R_2$.
We do not need to assume that the $\eps$-dependence of the terms in the non-diagonal blocks is 
given by a Laurent polynomial in $\eps$.
The algorithm can handle a rational dependence in $\eps$: 
One first performs a partial fraction decomposition in $\eps$ and then
removes any term which is not proportional to $\eps^1$.
\begin{myexample}
We illustrate this with a simple example consisting of two sectors.
Assume that the differential equation $d_BJ = A J$ is of the form
\bq
 A & = &
 \left( \begin{array}{rr}
 \eps A^{(1)}_{11} & 0 \\
 \frac{1}{1+\eps} X_{21} + A^{(0)}_{21} + \eps A^{(1)}_{21} & \eps A^{(1)}_{22} \\
 \end{array} \right).
\eq
We first remove $X_{21}$ with the transformation 
\bq
 R & = & 
 \left( \begin{array}{rr}
 1 & 0 \\
 \frac{1}{1+\eps} R_{21} & 1 \\
 \end{array} \right).
\eq
$R_{21}$ is determined by the $\eps$-independent differential equation
\bq
 d_B R_{21} & = & X_{21} - A^{(1)}_{22} R_{21} + R_{21} A^{(1)}_{11}.
\eq
After this transformation the lower-left entry is linear in $\eps$. The $\eps^0$-term is then removed
as in example~\ref{example_two_blocks}.
\end{myexample}


\subsection{The top weight master integrals}
\label{sect:top_weight}

The master integrals in $V^{(\NV,\NV)}$ are particularly simple.
They have Hodge weight $w=2\NV$ and correspond to zero-dimensional geometries (i.e. a set of points).
In particular, all Feynman integrals, which evaluate to multiple polylogarithms are of this type.
The leading singularities \cite{Cachazo:2008vp,Arkani-Hamed:2010pyv} of the top weight master integrals are algebraic.
Normalising these master integrals by their leading singularities reduces the number of auxiliary functions
in step $2$ of our algorithm.

\begin{myexample}
We consider sector $57$ in the family of Feynman integrals defined by the inverse propagators
\begin{align}
 \sigma_1 & = -\left(k_1-p_1\right)^2 +m^2,
 &
 \sigma_2 & = -\left(k_1-p_{12}\right)^2,
 &
 \sigma_3 & = -k_1^2,
 \nonumber \\
 \sigma_4 & = -\left(k_1+k_2\right)^2 + m^2,
 &
 \sigma_5 & = -\left(k_2+p_{12}\right)^2,
 &
 \sigma_6 & = -k_2^2,
 \nonumber \\
 \sigma_7 & = -\left(k_2+p_{123}\right)^2 + m^2,
 & 
 \sigma_8 & = -\left(k_1-p_{13}\right)^2,
 &
 \sigma_9 & = -\left(k_2+p_{13}\right)^2.
\end{align}
The Feynman graph for sector $57$ is shown in fig.~\ref{fig:sector57}.
\begin{figure}
\begin{center}
\includegraphics[scale=1.0]{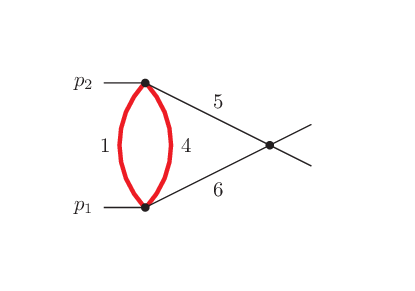}
\end{center}
\caption{
The Feynman graph for sector $57$.
Massive propagators are indicated by red lines.
}
\label{fig:sector57}
\end{figure}
With $x_1=m^2/s$ and $z_1=(\sigma_9-\sigma_7-t+m^2)/s$ we obtain the twist function
\bq
 U & = & P_0^{2\eps} P_1^{-\eps} P_2^{-\frac{1}{2}-\eps} P_3^{-\frac{1}{2}-\eps}
\eq
where the Baikov polynomials are given by
\bq
 P_0 \; =\; z_0,
 \;\;\;
 P_1 \; =\; z_1-z_0,
 \;\;\;
 P_2 \; =\; z_1,
 \;\;\;
 P_3 \; =\; z_1 + 4 x_1 z_0.
\eq
We further have
\bq
 \prebaikov \; = \; -\frac{8i e^{2 \eps \Eulerconstant} \pi^3}{\Gamma\left(1-2\eps\right)},
 & & 
 \preabs \; = \; \eps^3.
\eq
This sector has two master integrals, which reside in $V^{(1,1)}$.
We set 
\bq
 J_1 \; = \; 2 \eps^3 I_{200111000},
 & &
 J_2 \; = \; -\eps^2 I_{200112000}.
\eq
We have
\bq
 \iota\left(J_1\right) \; = \; \differentialform_{1000}\left[1\right],
 & & 
 \iota\left(J_2\right) \; = \; \differentialform_{0100}\left[1\right].
\eq
The differential equation in the basis $J=(J_1,J_2)$ is of the form
\bq
 d_B J & = & \left( A^{(0)} + \eps A^{(1)} \right) J,
 \;\;\;\;\;\;
 A^{(0)}
 \; = \;
 \left( \begin{array}{cc}
 0 & 0 \\
 0 & A^{(0)}_{22} \\
 \end{array} \right).
\eq
The differential form $\differentialform^0_{1000}\left[1\right]$ has constant leading singularities
\bq
\mathrm{Res}_{P_0} \differentialform^0_{1000}\left[1\right] \; = \; C,
 \;\;\;\;\;\;
 C \; = \; -16 i \frac{e^{2 \eps \Eulerconstant} \pi^3 \eps^3}{\Gamma\left(1-2\eps\right)},
\eq
where $C$ is a constant of uniform transcendental weight zero.
However, the differential form $\differentialform^0_{0100}\left[1\right]$ does not have constant leading singularities. We find
\bq
\mathrm{Res}_{P_1} \differentialform^0_{0100}\left[1\right] \; = \; -\frac{1}{2\sqrt{1+4x_1}} C.
\eq
It is well-known that an $\eps$-factorised form can be achieved by normalising $J_2$ with $\sqrt{1+4x_1}$, e.g. setting
\bq
\label{sector_57_J_to_K}
 J_1 \; = \; K_1, & & J_2 \; = \; \frac{1}{\sqrt{1+4x_1}} K_2
\eq
will lead on the maximal cut to an $\eps$-factorised differential equation in the basis $K$.
The same is achieved by our algorithm.
We start from the ansatz
\bq
 J & = & R^{(0)}_2 K,
 \;\;\;\;\;\;
 R^{(0)}_2
 \; = \; 
 \left( \begin{array}{cc}
 R^{(0)}_{11} & R^{(0)}_{12} \\
 R^{(0)}_{21} & R^{(0)}_{22} \\
 \end{array} \right).
\eq
The functions in $R^{(0)}_2$ are determined by the differential equation
\bq
 d_B R^{(0)}_2 & = & A^{(0)} R^{(0)}_2.
\eq
We immediately see that $R^{(0)}_{11}$ and $R^{(0)}_{12}$ are constant and we may choose $R^{(0)}_{11}=1$ and $R^{(0)}_{12}=0$.
$R^{(0)}_{21}$ and $R^{(0)}_{22}$ satisfy the differential equation
\bq
 d_B R^{(0)}_{2i} & = & A^{(0)}_{22} R^{(0)}_{2i},
 \;\;\;\;\;\; i \in \{1,2\}.
\eq
For $R^{(0)}_{21}$ we may pick the trivial solution $R^{(0)}_{21}=0$, however we cannot do this for $R^{(0)}_{22}$ as we must have $\det R^{(0)}_2 \neq 0$.
Thus we seek a non-trivial solution to the differential equation
\bq
 \frac{d}{dx_1} R^{(0)}_{22} & = & - \frac{2}{1+4x_1} R^{(0)}_{22}.
\eq
This differential equation has the solution
\bq
 R^{(0)}_{22} & = & \frac{C'}{\sqrt{1+4x_1}},
\eq
where $C'$ is an integration constant. 
We see that we recover exactly the basis change of eq.~(\ref{sector_57_J_to_K}).

\end{myexample}


\section{Examples}
\label{sect:examples}

In this section, we illustrate the algorithm with several examples.
We start with a trivial example, the one-loop massless box integral in section~\ref{sect:one_loop_box}.
All one-loop integrals correspond on the maximal cut
to differential $0$-forms in ${\mathbb C}{\mathbb P}^0$ and are therefore trivial.
Starting from two-loops, we get on the maximal cut Baikov representations of dimension larger than zero.
We first discuss in section~\ref{sect:double_box} the massless planar double box integral, followed by the massless planar pentabox integral
in section~\ref{sect:pentabox}.
These are still rather simple integrals, as they both evaluate to multiple polylogarithms.
In the language of our algorithm, this means that all integrands can be localised on points.
In order to appreciate the capabilities of our algorithm, we have to go beyond multiple polylogarithms.
In section~\ref{sect:sector_93_moeller} and section~\ref{sect:electron_self_energy}, 
we discuss two examples of Feynman integrals associated with elliptic curves.
In both examples, we start with differential forms on ${\mathbb C}{\mathbb P}^2$, and in both cases, the differential one-forms associated
with the elliptic curves live on a localisation.
In section~\ref{sect:four_loop_banana}, we discuss an example associated to a Calabi--Yau threefold.

In the examples we mainly focus on step $1$ of the algorithm, as this part contains most of the new features.
Part $2$ of the algorithm is mechanical, we provide for this part one example in section~\ref{sect:electron_self_energy},
other examples can be found in refs.~\cite{e-collaboration:2025frv,Pogel:2025bca}.


\subsection{The one-loop box}
\label{sect:one_loop_box}

We start with a trivial example, the massless one-loop box integral.
There is one master integral in the top sector.
The inverse propagators are
\begin{align}
 \sigma_1 & = -k_1^2,
 &
 \sigma_2 & = -\left(k_1-p_1\right)^2,
 &
 \sigma_3 & = -\left(k_1-p_{12}\right)^2,
 &
 \sigma_4 & = -\left(k_1-p_{123}\right)^2,
\end{align}
with $p_{ij}=p_i+p_j$, $p_{ijk}=p_i+p_j+p_k$ and likewise for $p_{ijkl}$ if any, which will be used for later examples as well. 
We set $s=(p_1+p_2)^2$, $t=(p_2+p_3)^2$, $x=s/t$ and the arbitrary scale $\arbitraryscale^2=t$.
\begin{figure}
\begin{center}
\includegraphics[scale=1.0]{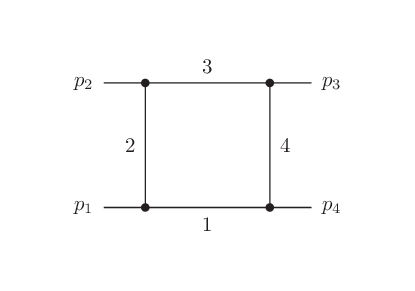}
\end{center}
\caption{
The one-loop box graph.
}
\label{fig:one_loop_box}
\end{figure}
The graph is shown in fig.~\ref{fig:one_loop_box}.
The Baikov representation on the maximal cut in $D=4-2\eps$ space-time dimensions reads
\bq
 e^{\eps \Eulerconstant} 
 t^{2+\eps}
 \int\limits_{{\mathcal C}_{\mathrm{maxcut}}}
 \frac{d^Dk_1}{i \pi^{2-\eps}}
 \frac{1}{\prod\limits_{j=1}^{4} \sigma_j}
 & = &
 \frac{2^{3+2\eps} \pi^{\frac{5}{2}} e^{\eps \Eulerconstant}}{\Gamma\left(\frac{1}{2}-\eps\right)} 
 x^{-1-\eps}
 \left(1+x\right)^\eps.
\eq
This is a zero-dimensional Baikov representation.
We have
\bq
 \prebaikov & = &
 \frac{2^{3+2\eps} \pi^{\frac{5}{2}} e^{\eps \Eulerconstant}}{\Gamma\left(\frac{1}{2}-\eps\right)} 
 x^{-1-\eps}
 \left(1+x\right)^\eps.
\eq
Setting 
\bq
 \preabs & = & \eps^2 x
\eq
one easily verifies that $\prebaikov \cdot \preabs$ is pure of transcendental weight zero.

In twisted cohomology, we consider ${\mathbb C}{\mathbb P}^0$. This is a point.
We have only one polynomial $P_0=z_0$ and the twist is simply $U=1$.
Following section~\ref{sect:localisation_on_a_point} there is one master integrand
\bq
 \differentialform_{1}\left[1\right]
 & = &
 -
 \frac{2^{3+2\eps} \pi^{\frac{5}{2}} e^{\eps \Eulerconstant} \eps^2}{\Gamma\left(\frac{1}{2}-\eps\right)} 
 x^{-\eps}
 \left(1+x\right)^\eps.
\eq
In this example, the integrand is a $0$-form.
The differential equation for $\differentialform_{1}[1]$ is already in $\eps$-factorised form:
\bq
 \frac{d}{dx} \differentialform_{1}\left[1\right] & = & \eps \left( \frac{1}{1+x}-\frac{1}{x}\right) \differentialform_{1}\left[1\right].
\eq
For this example we have an isomorphism $\iota : V^0 \rightarrow \Hgen^0_\omega$, given by
\bq
 \iota\left( - \eps^2 x \; I_{1111} \right) & = & \differentialform_{1}\left[1\right].
\eq
It is easily verified that $(- \eps^2 x \; I_{1111})$ is a pure master integral of transcendental weight zero
and satisfies an $\eps$-factorised differential equation, including sub-sectors.


\subsection{The two-loop double box}
\label{sect:double_box}

We continue with a simple example, the two-loop planar massless double box integral.
This is now an example with two master integrals in the top sector.
The inverse propagators are
\begin{align}
 \sigma_1 & = -\left(k_1-p_1\right)^2,
 &
 \sigma_2 & = -\left(k_1-p_{12}\right)^2,
 &
 \sigma_3 & = -k_1^2,
 \nonumber \\
 \sigma_4 & = -\left(k_1+k_2\right)^2,
 &
 \sigma_5 & = -\left(k_2+p_{12}\right)^2,
 &
 \sigma_6 & = -k_2^2,
 \nonumber \\
 \sigma_7 & = -\left(k_2+p_{123}\right)^2,
 & 
 \sigma_8 & = -\left(k_1-p_{13}\right)^2,
 &
 \sigma_9 & = -\left(k_2+p_{13}\right)^2.
\end{align}
The kinematics is as in the previous example.
\begin{figure}
\begin{center}
\includegraphics[scale=1.0]{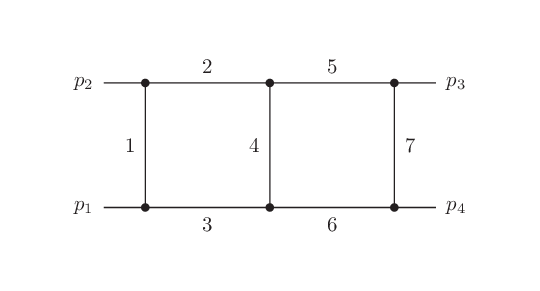}
\end{center}
\caption{
The two-loop double box graph.
}
\label{fig:double_box}
\end{figure}
The graph is shown in fig.~\ref{fig:double_box}.
We consider the loop-by-loop Baikov representation, where the inner loop is formed by the 
inverse propagators $\sigma_4, \sigma_5, \sigma_6, \sigma_7$.
With
\bq
 r
 & = & 
 \frac{t^5 \sigma_8 \left(\sigma_8-t\right)}{\sigma_1 \sigma_2 \sigma_3 \sigma_4 \sigma_5 \sigma_6 \sigma_7}
\eq
and $z_1=\sigma_8/t$
we obtain a Baikov representation with the minimal twist function:
\bq
\lefteqn{
 e^{2 \eps \Eulerconstant} 
 t^{-4+2\eps}
 \int\limits_{{\mathcal C}_{\mathrm{maxcut}}}
 \frac{d^Dk_1}{i \pi^{2-\eps}}
 \frac{d^Dk_2}{i \pi^{2-\eps}}
 r\left(\sigma\right)
 = } & &
 \nonumber \\
 & &
 \frac{2^{6+4\eps} \pi^{5} e^{2 \eps \Eulerconstant}}{\Gamma\left(\frac{1}{2}-\eps\right)^2} 
 x^{-2-2\eps}
 \left(1+x\right)^\eps
 \int \frac{dz_1}{\left(2\pi i\right)}
 z_1^{-2\eps} \left(z_1-1\right)^{-\eps} \left(z_1-x-1\right)^\eps
\eq
We read off
\bq
 \prebaikov
 &= &
 \frac{2^{6+4\eps} \pi^{5} e^{2 \eps \Eulerconstant}}{\Gamma\left(\frac{1}{2}-\eps\right)^2} 
 x^{-2-2\eps}
 \left(1+x\right)^\eps.
\eq
Setting
\bq
 \preabs & = & \eps^4 x^2
\eq
ensures that $\prebaikov \cdot \preabs$ is pure of transcendental weight zero.
In projective space, the minimal twist function reads 
\bq
 U
 & = &
 z_0^{2\eps}
 z_1^{-2\eps} \left(z_1-z_0\right)^{-\eps} \left[z_1-\left(x+1\right)z_0\right]^\eps.
\eq
We set
\begin{align}
 P_0 & = z_0,
 &
 P_1 & = z_1,
 &
 P_2 & = z_1-z_0,
 &
 P_3 & = z_1-\left(x+1\right)z_0.
\end{align}
The sets $I_{\mathrm{even}}^0$ and $I_{\mathrm{odd}}^0$ are given by
\bq
 I_{\mathrm{even}}^0 \; = \; \left\{0,1,2,3\right\},
 & &
 I_{\mathrm{odd}}^0 \; = \; \emptyset.
\eq
We have
\bq
 \dim V^1 \; = \; \dim \Hgen^1_\omega \; = \; 2,
\eq
hence there is an isomorphism between $V^1$ and $\Hgen^1_\omega$, and we do not need to worry about symmetries nor super-sectors.
The rational function $\hat{\Phi}_{\mu_0 \mu_1 \mu_2 \mu_3}[Q]$ has to be homogeneous of degree $(-2)$,
as $\eta=z_0 dz_1 - z_1 dz_0$ is homogeneous of degree $2$.
As we have a one-dimensional Baikov representation, we work in projective space ${\mathbb C}{\mathbb P}^1$.
The algorithm will first consider all possible localisations. As we have four even polynomials, we may localise
on the points
\bq
 \left[z_0:z_1\right]
 & \in & 
 \left\{
 \;\;
 \left[0:1\right],
 \;\;
 \left[1:0\right],
 \;\;
 \left[1:1\right],
 \;\;
 \left[1:1+x\right]
 \;\;
 \right\}.
\eq
From each localisation, we obtain one master integrand. We may obtain the same master integrand from different localisations.
An example is given by
\bq
 \differentialform_{1100}\left[1\right]
 & = & 
 -4 \eps^4 x^2 \prebaikov U \frac{\eta}{z_0z_1},
\eq
which has residues at $P_0=0$ and $P_1=0$.
Note that the number of consecutive non-zero residues equals one for this differential form.
After we have taken a residue at either $P_0=0$ or $P_1=0$, we localise to a point, and we may no longer take another
residue at another point.
Note further that the differential forms 
\bq
 \differentialform_{1100}\left[1\right],
 \;\;\;
 \differentialform_{1010}\left[1\right],
 \;\;\;
 \differentialform_{1001}\left[1\right]
\eq
are up to trivial prefactors equal on the localisation $P_0=0$.
It will depend on the unspecified dots in the order relation $\laportaorder$, which differential form is picked.
After considering all possible localisations, we will have up to four differential forms with $a=-2$.
Running the Laporta algorithm on the full sector will reduce them to two master integrands.
A possible choice for the master integrands is
\bq
\label{master_integrands_doublebox}
 \differentialform_{0110}\left[1\right]
 & = &
 2 \eps^4 x^2 \prebaikov U \frac{\eta}{z_1\left(z_1-z_0\right)},
 \nonumber \\
 \differentialform_{1010}\left[1\right]
 & = & 
 -2 \eps^4 x^2 \prebaikov U \frac{\eta}{z_0\left(z_1-z_0\right)}.
\eq
The exact choice will depend on the unspecified dots in the order relation $\laportaorder$. 
This choice is of no relevance here.
The isomorphism $\iota : V^1 \rightarrow \Hgen^1_\omega$ is given by
\bq
 \iota\left( 2 \eps^4 x^2 I_{111111100} \right) & = & \differentialform_{0110}\left[1\right],
 \nonumber \\
 \iota\left( -2 \eps^4 x^2 I_{1111111\left(-1\right)0} \right) & = & \differentialform_{1010}\left[1\right].
\eq
Up to trivial prefactors, this choice corresponds to the example discussed in ref.~\cite{Weinzierl:2022eaz}.
The differential equation for the master integrands defined in eq.~(\ref{master_integrands_doublebox}) is already in
$\eps$-factorised form, therefore there is nothing to be done in step $2$ for the top sector.
In the final step, we consider the full system, including all subsectors.
For the choice of the master integrands as in eq.~(\ref{master_integrands_doublebox})
one finds that $(-2 \eps^4 x^2 I_{1111111\left(-1\right)0})$ will receive corrections from subsectors.
These are obtained in a straightforward way with an ansatz as in step 2, and
the final result corresponds to the one given in ref.~\cite{Weinzierl:2022eaz}.

In summary, the Hodge-like diagram for the decomposition of $V^1$ 
with respect to the $W_\bullet$-filtration and the $\Fgeom^\bullet$-filtration is
\begin{center}
\begin{axopicture}(280,140)(0,0)
\Text(110,60)[c]{$0$}
\Text(150,60)[c]{$0$}
\Text(130,80)[c]{$2$}
\DashLine(30,70)(260,70){6}
\DashLine(30,90)(260,90){6}
\Text(253,60)[c]{$W_1$}
\Text(253,80)[c]{$W_2$}
\Line(260,70)(260,64)
\Line(260,90)(260,84)
\DashLine(120,10)(220,110){3}
\DashLine(80,10)(180,110){3}
\Text(110,20)[c]{$\Fgeom^0$}
\Text(70,20)[c]{$\Fgeom^1$}
\Line(120,10)(117,10)
\Line(80,10)(77,10)
\end{axopicture}
\end{center}

\subsection{The two-loop pentabox integral}
\label{sect:pentabox}

The next example we are considering is the two-loop five-point pentagon box integral, see fig.~\ref{fig:penta_box}. 
An $\eps$-factorised form for this family of Feynman integrals has been given in ref.~\cite{Gehrmann:2018yef}.
This is an example with three master integrals in the top sector.
The inverse propagators are defined as:
\begin{align}
 \sigma_1 & = -k_1^2 ,
 &
 \sigma_2 & = -\left(k_{1}+p_1\right)^2 ,
 &
 \sigma_3 & = -\left(k_{1}+p_{12}\right)^2,
 \nonumber \\
 \sigma_4 & = -\left(k_{1}+p_{123}\right)^2,
 &
 \sigma_5 & = - k_2^2,
 &
 \sigma_6 & = -\left(k_{2}+p_{123}\right)^2,
 \nonumber \\
 \sigma_7 & = -\left(k_{2}+p_{1234}\right)^2,
 &
 \sigma_8 & = -\left(k_1-k_2\right)^2,
 &
 \sigma_9 & = -\left(k_{1}+p_{1234}\right)^2,
 \nonumber \\
 \sigma_{10} & = -\left(k_{2}+p_{1}\right)^2,
 &
 \sigma_{11} & = -\left(k_{2}+p_{12}\right)^2.
\end{align}
\begin{figure}
\begin{center}
\includegraphics[scale=1.0]{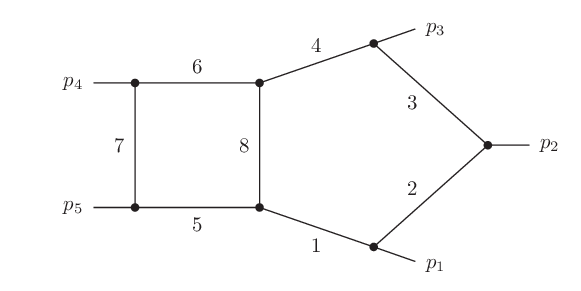}
\end{center}
\caption{
The two-loop pentagon box graph.
}
\label{fig:penta_box}
\end{figure}
The integral depends on four independent (dimensionless) kinematic variables
\bq
	x_1 = \dfrac{s_{12}}{s_{45}}, \quad x_2 = \dfrac{s_{23}}{s_{45}}, \quad x_3 = \dfrac{s_{34}}{s_{45}}, \quad x_4 = \dfrac{s_{15}}{s_{45}}, 
\eq
where $s_{ij}=2 p_i \cdot p_j$. We set $\arbitraryscale^2 = s_{45}$. A loop-by-loop Baikov representation on the maximal cut is
\begin{align}
& e^{2 \eps \Eulerconstant}   \left(s_{45}\right)^{4+2\eps} \int\limits_{{\mathcal C}_{\mathrm{maxcut}}} \frac{d^Dk_1}{i \pi^{2-\eps}} \frac{d^Dk_2}{i \pi^{2-\eps}} \prod_{i=1}^{8}\dfrac{1} {\sigma_i}=  \\
\prebaikov
& \int \frac{dz_1}{\left(2\pi i\right)}
 z_1^{-1-\eps}\left(1-z_1\right)^{\eps} \left(\left(1-x_1-x_2\right) z_1^2+\left(x_3 - x_2 x_3+x_1\left(x_2 -x_4 \right) + x_4\right) z_1+x_3 x_4 \right)^{-1-\eps},\nonumber
\end{align}
where $z_1=\sigma_9/s_{45}$ and 
\begin{align}
\prebaikov &= \dfrac{32 i\,e^{2 \eps \Eulerconstant}  \pi^{5}}{\Gamma\left( -2 \eps\right)} x_1^{-1-\eps} x_2^{-1-\eps} \Delta^{\frac{1}{2}+\eps}, \nonumber \\
\Delta &= \left(x_1 x_2+x_2 x_3-x_3 +x_4-x_1 x_4\right){}^2-4 x_1 x_2 x_3 \left(x_2-x_4-1\right).
\end{align}
By setting 
\bq
	\preabs= \dfrac{\eps^4 x_1 x_2}{ \sqrt{\Delta}}
\eq
the product $\prebaikov \cdot \preabs$ is pure of transcendental weight zero. The minimal twist function is 
\bq
	U =  z_0^{2 \eps }  z_1^{-\eps } \left(z_1-z_0\right)^{\eps } \left(\left(1-x_1-x_2\right) z_1^2+\left(x_3 - x_2 x_3+x_1\left(x_2 -x_4 \right) + x_4\right) z_0 z_1+x_3 x_4 z_0^2\right)^{-\eps}
\eq
in projective space.
We set 
\begin{align}
	P_0&=z_0, \qquad P_1=z_1, \qquad P_2=z_1 - z_0, \nonumber \\
	P_3&=\left(1-x_1-x_2\right) z_1^2+\left(x_3 - x_2 x_3+x_1\left(x_2 -x_4 \right) + x_4\right) z_0 z_1+x_3 x_4 z_0^2.
\end{align}
Therefore, the sets $I_{\mathrm{even}}^0$ and $I_{\mathrm{odd}}^0$ are 
\bq
 I_{\mathrm{even}}^0 \; = \; \left\{0,1,2,3\right\},
 & &
 I_{\mathrm{odd}}^0 \; = \; \emptyset.
\eq
We find that 
\bq
 \dim V^1 \; = \; \dim \Hgen^1_\omega \; = \; 3.
\eq
Consequently, there is an isomorphism between $V^1$ and $\Hgen^1_\omega$, and we do not have to worry about symmetries and super-sectors, as in the previous example.

Since $\eta=z_0 dz_1 - z_1 dz_0$ is homogeneous of degree $2$ and $d_U=0$, the rational functions $\hat{\Phi}_{\mu_0 \mu_1 \mu_2 \mu_3}[Q]$ have to be homogeneous of degree $(-2)$. We start by looking at all possible localisations. 
There are five possibilities to localise on a point, and we will get up to five differential forms with $a=-2$.
Running the Laporta algorithm on the full sector will reduce them to three master integrands.
A possible choice for the master integrands is
\bq
\label{master_integrands_pentabox}
 \differentialform_{0001}\left[1 \right]
  & = &  -\eps^4\, \frac{x_1x_2}{\sqrt{\Delta}} C_{\rm Baikov}\, U \frac{\eta}{P_3},
 \nonumber \\
 \differentialform_{0101}\left[z_0 \right]
 & = &
 \eps^4\, \frac{x_1x_2}{\sqrt{\Delta}} C_{\rm Baikov}\, U \frac{z_0\,\eta}{P_1P_3},
 \nonumber \\
 \differentialform_{1001}\left[z_1 \right]
 & = & -2\eps^4\,\frac{x_1x_2}{\sqrt{\Delta}} C_{\rm Baikov}\, U \frac{z_1\,\eta}{P_0P_3}.
\eq
The exact choice will depend on the unspecified dots in the order relation $\laportaorder$. 
This choice is of no relevance here.
The isomorphism $\iota : V^1 \rightarrow \Hgen^1_\omega$ is given by
\bq
 \iota\left( - \eps^4 \frac{x_1x_2}{\sqrt{\Delta}} I_{11111111(-1)00} \right) & = & \differentialform_{0001}\left[1\right],
 \nonumber \\
 \iota\left( \eps^4 \frac{x_1x_2}{\sqrt{\Delta}} I_{11111111000} \right) & = & \differentialform_{0101}\left[z_0\right],\\
 \iota\left( -2\eps^4 \frac{x_1x_2}{\sqrt{\Delta}} I_{11111111(-2)00} \right) & = & \differentialform_{1001}\left[z_1\right].\nonumber
\eq
The decomposition of $V^1$ with respect to $(W_\bullet,\Fgeom^\bullet)$ is
\begin{center}
\begin{axopicture}(280,140)(0,0)
\Text(110,60)[c]{$0$}
\Text(150,60)[c]{$0$}
\Text(130,80)[c]{$3$}
\DashLine(30,70)(260,70){6}
\DashLine(30,90)(260,90){6}
\Text(253,60)[c]{$W_1$}
\Text(253,80)[c]{$W_2$}
\Line(260,70)(260,64)
\Line(260,90)(260,84)
\DashLine(120,10)(220,110){3}
\DashLine(80,10)(180,110){3}
\Text(110,20)[c]{$\Fgeom^0$}
\Text(70,20)[c]{$\Fgeom^1$}
\Line(120,10)(117,10)
\Line(80,10)(77,10)
\end{axopicture}
\end{center}
The differential equation for $(\differentialform_{0001}\left[1\right], \differentialform_{0101}\left[z_0\right], \differentialform_{1001}\left[z_1\right])$
is in Laurent-polynomial form and compatible with the $\Fcomb^\bullet$-filtration. 
As there is only one non-trivial part in the $\Fcomb^\bullet$-filtration, we have that the matrix $A$ is linear in $\eps$.
We may then proceed to step $2$ of the algorithm and arrive at an $\eps$-factorised basis equivalent to the one given in ref.~\cite{Gehrmann:2018yef}.

\subsection{A two-loop contribution to M{\o}ller scattering}
\label{sect:sector_93_moeller}

As the next example, we consider a two-loop integral contributing to M{\o}ller scattering.
This is now a non-trivial example with five master integrals in the top sector and an elliptic curve associated to
the top sector.
We use this example to illustrate details of the construction of the intermediate basis $J$.
The inverse propagators are defined by
\begin{align}
 \sigma_1 & = -\left(k_1-p_1\right)^2 +m^2,
 &
 \sigma_2 & = -\left(k_1-p_{12}\right)^2,
 &
 \sigma_3 & = -k_1^2,
 \nonumber \\
 \sigma_4 & = -\left(k_1+k_2\right)^2 + m^2,
 &
 \sigma_5 & = -\left(k_2+p_{12}\right)^2,
 &
 \sigma_6 & = -k_2^2,
 \nonumber \\
 \sigma_7 & = -\left(k_2+p_{123}\right)^2 + m^2,
 & 
 \sigma_8 & = -\left(k_1-p_{13}\right)^2,
 &
 \sigma_9 & = -\left(k_2+p_{13}\right)^2.
\end{align}
\begin{figure}
\begin{center}
\includegraphics[scale=1.0]{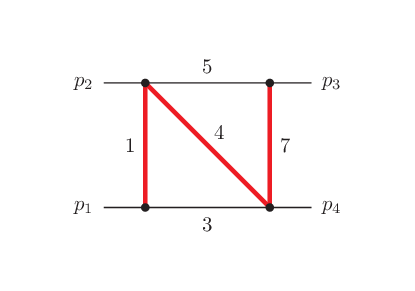}
\end{center}
\caption{
A two-loop contribution to M{\o}ller scattering.
Massive propagators are indicated by red lines.
}
\label{fig:moeller_sector_93}
\end{figure}
We are interested in sector $93$ with inverse propagators $\{\sigma_1,\sigma_3,\sigma_4,\sigma_5,\sigma_7\}$.
The graph is shown in fig.~\ref{fig:moeller_sector_93}.
We consider the loop-by-loop Baikov representation, where the inner loop is formed by the inverse propagators $\{\sigma_4,\sigma_5,\sigma_7\}$.
We set $x_1=m^2/s$, $x_2=t/s$.
With
\bq
 r
 & = & 
 \frac{s^4 \left(\sigma_8+m^2-t\right)}{\sigma_1 \sigma_3 \sigma_4 \sigma_5 \sigma_7}
\eq
and $z_1=(\sigma_8+m^2-t)/s$, $z_2=(\sigma_2+\sigma_8+m^2-t)/s$,
we obtain a Baikov representation with a minimal twist function:
\bq
 e^{2 \eps \Eulerconstant} 
 s^{-4+2\eps}
 \int\limits_{{\mathcal C}_{\mathrm{maxcut}}}
 \frac{d^Dk_1}{i \pi^{2-\eps}}
 \frac{d^Dk_2}{i \pi^{2-\eps}}
 r\left(\sigma\right)
 & = &
 \prebaikov
 \int \frac{dz_1}{\left(2\pi i\right)} \frac{dz_2}{\left(2\pi i\right)} \left. U\left(z\right) \right|_{z_0=1},
\eq
where
\bq
 \prebaikov
 & = & 
 - \frac{32 \pi^4 e^{2 \eps \Eulerconstant}}{\Gamma\left(1-2\eps\right)} x_1^{-\eps} x_2^\eps \left(1+x_2\right)^\eps.
\eq
Setting
\bq
 \preabs & = & \eps^4 
\eq
ensures that $\prebaikov \cdot \preabs$ is pure of transcendental weight zero.
In projective space, the minimal twist function reads 
\bq
 U & = &
 P_0^{2\eps} P_1^{2\eps} P_2^{-\eps} P_3^{-\frac{1}{2}-\eps}.
\eq
The polynomials are given by
\bq
 P_0
 & = &
 z_0,
 \nonumber \\
 P_1
 & = &
 z_1,
 \nonumber \\
 P_2
 & = &
 z_1 z_2 + x_1 z_0 z_2 - z_1^2,
 \nonumber \\
 P_3
 & = &
 \left[ x_2 z_1 - \left(1+x_2\right) z_2 + \left(x_1-x_2\right) z_0\right]^2 
 + 4 x_1 \left(1+x_2\right) z_0 z_2.
\eq
We have three even polynomials $I_{\mathrm{even}}^0=\{0,1,2\}$ and one odd polynomial $I_{\mathrm{odd}}^0=\{3\}$.
In this example, the dimensions of $V^2$ and $\Hgen^2_\omega$ do not match.
We find
\bq
 \dim V^2 \; = \; 5,
 & &
 \dim \Hgen^2_\omega \; = \; 6.
\eq
The mismatch in dimensions is related to symmetries. An example of a symmetry relation is
\bq
 I_{\nu_1 \nu_2 \nu_3 \nu_4 \nu_5 \nu_6 \nu_7 0 0}
 & = &
 I_{\nu_7 \nu_6 \nu_5 \nu_4 \nu_3 \nu_2 \nu_1 0 0}.
\eq
We have $d_U=-1$, hence $\hat{\Phi}$ must be homogeneous of degree $-2$.
In order to keep the notation compact, we will use 
whenever there is no ambiguity
the pre-images $\differentialform_{\mu_0 \mu_1 \mu_2 \mu_3}[Q]$
instead of the more lengthy notation $\mathrm{Res}_{{\mathcal P}_1,\dots,{\mathcal P}_r} \differentialform_{\mu_0 \mu_1 \mu_2 \mu_3}[Q,I,\Divisor_s]$.
 
In the construction of the intermediate basis $J$, we first consider recursively all possible localisations.
We have three even polynomials, and we consider in turn the localisation on each of them.

We start with the localisation at $P_0=0$. In this limit, 
$P_2$ factorises
\bq
 \left. P_2 \right|_{z_0=0} & = & 
 z_1 \left( z_2 - z_1\right),
\eq
and $P_3$ becomes a perfect square
\bq
 \left. P_3 \right|_{z_0=0} & = & 
 \left[ x_2 z_1 - \left(1+x_2\right) z_2 \right]^2.
\eq
We may further localise on $P_1$, $P_2$ or $P_3$, resulting in the localisations on the points
\bq
 \left[z_0:z_1:z_2\right]
 & \in & 
 \left\{
 \;
  \left[0:0:1\right],
 \;
  \left[0:1:1\right],
 \;
  \left[0:1+x_2:x_2\right]
 \;
 \right\}.
\eq
The differential form $\differentialform_{1100}[1]$ has non-zero residues at the points $[0:0:1]$ and $[0:1+x_2:x_2]$,
the differential form $\differentialform_{1010}\left[z_1\right]$ has non-zero residues at the points $[0:1:1]$ and $[0:1+x_2:x_2]$.
Hence, we obtain at weight $w=4$ the master integrands
\bq
\label{example_moeller_93_localisation_P0}
 \differentialform_{\mu_0 \mu_1 \mu_2 \mu_3}
 & \in &
 \left\{
 \;
  \differentialform_{1100}\left[1\right],
 \;
  \differentialform_{1010}\left[z_1\right]
 \;
 \right\}.
\eq
We assign $a=-4$, $r=2$ and $o=2$ to these three differential forms.

There is no merge procedure in going from zero-dimensional varieties to one-dimensional varieties, as there are no intersections between distinct points.

We then run the Laporta algorithm for the localisation at $P_0=0$. 
We choose $P_1$ as a scale polynomial.
This gives us
\bq
 U^{\mathrm{loc}}_{\langle P_0 \rangle} & = & P_1^{4\eps} P_2^{-\eps} P_3^{-\frac{1}{2}-\eps}
\eq
as the modified twist function.
From an analysis of the number of critical points on the localisation $P_0=0$, we expect to find one master integrand,
and indeed there is one linear relations among the two candidates in eq.~(\ref{example_moeller_93_localisation_P0}).
The unspecified dots in the order relation will determine the candidate to be eliminated.
We assume that $\differentialform_{1010}\left[z_1\right]$ is eliminated.
Thus we obtain 
\bq
 \differentialform_{1100}\left[1\right]
\eq
as master integrand from the localisation at $P_0=0$.

Next, we consider the localisation at $P_1=0$.
In this limit, we have
\bq
 \left. P_2 \right|_{z_1=0} & = & 
 x_1 z_0 z_2,
 \nonumber \\
 \left. P_3 \right|_{z_1=0} & = & 
 \left[ - \left(1+x_2\right) z_2 + \left(x_1-x_2\right) z_0\right]^2 
 + 4 x_1 \left(1+x_2\right) z_0 z_2.
\eq
In principle, the square root associated with $\left. P_3 \right|_{z_1=0}$ can be rationalised, but this is not needed.
After localisation on $z_1=0$, we may further localise on $z_0=0$ or on $z_2=0$,
i.e. on the points
\bq
 \left[z_0:z_1:z_2\right]
 & \in & 
 \left\{
 \;
  \left[0:0:1\right],
 \;
  \left[1:0:0\right]
 \;
 \right\}.
\eq
The first possibility corresponds to $\differentialform_{1100}[1]$, which we already obtained
from the localisation on $P_0=0$.
The second possibility gives $\differentialform_{0110}[z_0]$.
We assign $a=-4$ to these two candidates and run 
the Laporta algorithm on the localisation $P_1=0$.
As a scale polynomial, we choose $P_0$,
giving us the modified twist function
\bq
 U^{\mathrm{loc}}_{\langle P_1 \rangle}
 & = &
 P_0^{4\eps} P_2^{-\eps} P_3^{-\frac{1}{2}-\eps}.
\eq
There are no further linear relations among these two candidates and 
$\{\differentialform_{1100}[1],\differentialform_{0110}[z_0]\}$ forms a basis for the localisation on $P_1=0$.

We then consider the localisation at $P_2=0$.
The polynomial $P_2$ is of degree $2$, however, it is linear in $z_0$ and $z_2$.
From an analysis of the number of critical points on the localisation $P_2=0$, we expect to find five master integrands.
We first consider the possible localisations on points.
We find three points
\bq
 \left\{
 \;
  \left[0:0:1\right],
 \;
  \left[0:1:1\right],
 \;
  \left[1:0:0\right]
 \;
 \right\},
\eq
giving us the master integrands
\bq
 \left\{
 \;
 \differentialform_{1010}\left[z_2\right],
 \;
 \differentialform_{1010}\left[z_1\right],
 \;
 \differentialform_{0110}\left[z_0\right]
 \;
 \right\},
\eq
to which we assign $a=-4$.
As a scale polynomial for the localisation $\langle \Divisor_2 \rangle$, we choose $\Divisor_3$,
giving us the modified twist function
\bq
 U^{\mathrm{loc}}_{\langle P_2 \rangle}
 & = &
 P_0^{2\eps} P_1^{2\eps} P_3^{-\frac{1}{2}-2\eps}.
\eq
Note that the choice of $\Divisor_0$ or $\Divisor_1$ as scale polynomial would give us 
a non-generic modified twist function, where either the divisor $\Divisor_0$ or $\Divisor_1$
is absent.
Running the Laporta algorithm on the localisation $\langle P_2 \rangle$, one obtains two additional 
master integrands
\bq
 \left\{
 \;
  \differentialform_{0010}\left[1\right],
 \;
  \differentialform_{0011}\left[z_0^2\right]
 \;
 \right\},
\eq
to which we assign $a=-3$.

We have now done all localisations. It is instructive to discuss the merge step on the $1$-skeleton in detail.
In order to do so, we switch to the longer notation $\mathrm{Res}_{{\mathcal P}_1,\dots,{\mathcal P}_r} \differentialform_{\mu_0 \mu_1 \mu_2 \mu_3}[Q,I,\Divisor_s]$.
The $1$-skeleton is given by
\bq
 V\left(\langle \Divisor_0 \rangle\right)
 \cup
 V\left(\langle \Divisor_1 \rangle\right)
 \cup
 V\left(\langle \Divisor_2 \rangle\right).
\eq
The $0$-skeleton is given by the four points
\bq
 \left[0:0:1\right]
 \cup
 \left[0:1:1\right]
 \cup
 \left[1:0:0\right]
 \cup
 \left[0:1+x_2:x_2\right].
\eq
The intersections of the components of the $1$-skeleton are given by the three points
\bq
 \bigcup\limits_{(ij)} V({\mathcal I_{ij}})
 & = &
 \left[0:0:1\right]
 \cup
 \left[0:1:1\right]
 \cup
 \left[1:0:0\right].
\eq
In the notation of section~\ref{sect:order_relation}, we write
\bq
 \bigcup\limits_{(ij)} V({\mathcal I_{ij}})
 & = &
 C_1 \cup C_2 \cup C_3
\eq
with
\bq
 C_1 \; = \; \left\{ \left[0:0:1\right] \right\},
 \;\;\;
 C_2 \; = \; \left\{ \left[0:1:1\right] \right\},
 \;\;\;
 C_3 \; = \; \left\{ \left[1:0:0\right] \right\}.
\eq
The skeleton is sketched in fig.~\ref{fig:moeller_sector_93_skeleton}.
\begin{figure}
\begin{center}
\includegraphics[scale=1.0]{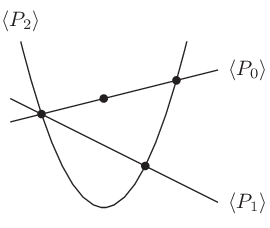}
\end{center}
\caption{
A sketch of the $1$-skeleton and the $0$-skeleton.
}
\label{fig:moeller_sector_93_skeleton}
\end{figure}
Note that
\bq
 C_1
 \; = \;
 V\left( \left\langle P_0 \right\rangle \right) \cap V\left( \left\langle P_1 \right\rangle \right),
 \;\;\;
 C_1 
 \; \subset \;  
 V\left( \left\langle P_0 \right\rangle \right) \cap V\left( \left\langle P_2 \right\rangle \right),
 \;\;\;
 C_1
 \; \subset \; 
 V\left( \left\langle P_1 \right\rangle \right) \cap V\left( \left\langle P_2 \right\rangle \right).
\eq
In the notation of section~\ref{sect:order_relation}, we have
\bq
 S_1 & = & \left\{ 0,1,2 \right\}.
\eq
The bases on the localisations of dimension $1$ are
\bq
\label{objects_from_localisations}
 V\left(\langle \Divisor_0 \rangle\right):
 & &
 \mathrm{Res}_{\Divisor_0} \differentialform_{1100}\left[1,\langle \Divisor_0 \rangle,\Divisor_1\right],
 \\
 V\left(\langle \Divisor_1 \rangle\right):
 & &
 \mathrm{Res}_{\Divisor_1} \differentialform_{1100}\left[1,\langle \Divisor_1 \rangle,\Divisor_0\right],
 \;
 \mathrm{Res}_{\Divisor_1} \differentialform_{0110}\left[z_0,\langle \Divisor_1 \rangle,\Divisor_0\right],
 \nonumber \\
 V\left(\langle \Divisor_2 \rangle\right):
 & &
 \mathrm{Res}_{\Divisor_2} \differentialform_{1010}\left[z_2,\langle \Divisor_2 \rangle,\Divisor_3\right],
 \;
 \mathrm{Res}_{\Divisor_2} \differentialform_{1010}\left[z_1,\langle \Divisor_2 \rangle,\Divisor_3\right],
 \;
 \mathrm{Res}_{\Divisor_2} \differentialform_{0110}\left[z_0,\langle \Divisor_2 \rangle,\Divisor_3\right],
 \nonumber \\
 & &
 \mathrm{Res}_{\Divisor_2} \differentialform_{0010}\left[1,\langle \Divisor_2 \rangle,\Divisor_3\right],
 \;
 \mathrm{Res}_{\Divisor_2} \differentialform_{0011}\left[z_0^2,\langle \Divisor_2 \rangle,\Divisor_3\right].
 \nonumber
\eq
Let us assume that we have chosen at the point $\left[0:0:1\right]$ the basis
\bq
 \left[0:0:1\right]: & &  \mathrm{Res}_{\Divisor_0,\Divisor_1} \differentialform_{1100}\left[1,\langle \Divisor_0,\Divisor_1 \rangle,\Divisor_3\right].
\eq
There are eight different objects appearing in eq.~(\ref{objects_from_localisations}).
From the intersections of the components of the $1$-skeleton, we obtain three equations:
\bq
 \left[0:0:1\right]:
 & &
 \mathrm{Res}_{\Divisor_0} \differentialform_{1100}\left[1,\langle \Divisor_0 \rangle,\Divisor_1\right]
 \; = \;
 \mathrm{Res}_{\Divisor_1} \differentialform_{1100}\left[1,\langle \Divisor_1 \rangle,\Divisor_0\right],
 \nonumber \\
 \left[0:1:1\right]:
 & &
 \mathrm{Res}_{\Divisor_0} \differentialform_{1010}\left[z_1,\langle \Divisor_0 \rangle,\Divisor_1\right]
 \; = \; 
 \mathrm{Res}_{\Divisor_2} \differentialform_{1010}\left[z_1,\langle \Divisor_2 \rangle,\Divisor_3\right],
 \nonumber \\
 \left[1:0:0\right]:
 & &
 \mathrm{Res}_{\Divisor_1} \differentialform_{0110}\left[z_0,\langle \Divisor_1 \rangle,\Divisor_0\right]
 \; = \;
 \mathrm{Res}_{\Divisor_2} \differentialform_{0110}\left[z_0,\langle \Divisor_2 \rangle,\Divisor_3\right].
\eq
Note that there is only one equation for the point $\left[0:0:1\right]$, as $\differentialform_{1100}\left[1\right]$ has a residue
on $\langle \Divisor_0 \rangle$ and $\langle \Divisor_1 \rangle$, but not on $\langle \Divisor_2 \rangle$.
Note further that $\mathrm{Res}_{\Divisor_0} \differentialform_{1010}\left[z_1,\langle \Divisor_0 \rangle,\Divisor_1\right]$
is not a master integrand on $\langle \Divisor_0 \rangle$.
We therefore add one equation, which relates this integrand to the master integrand on $\langle \Divisor_0 \rangle$:
\bq
 c_1 \mathrm{Res}_{\Divisor_0} \differentialform_{1100}\left[1,\langle \Divisor_0 \rangle,\Divisor_1\right]
 +
 c_2 \mathrm{Res}_{\Divisor_0} \differentialform_{1010}\left[z_1,\langle \Divisor_0 \rangle,\Divisor_1\right]
 & = & 0.
\eq
The coefficients are obtained from the reduction on $\langle \Divisor_0 \rangle$.
The actual values do not matter, it is only relevant for us that $c_1 \neq 0$ and $c_2 \neq 0$.
We now have nine objects and four equations, leaving us with five independent integrands.
After renaming, we have
\bq
 \differentialform_{1100}\left[1\right],
 \;
 \differentialform_{1010}\left[z_2\right],
 \;
 \differentialform_{0110}\left[z_0\right],
 \;
 \differentialform_{0010}\left[1\right],
 \;
 \differentialform_{0011}\left[z_0^2\right].
\eq
Finally, we run the reduction on the full system without any localisation.
This will give one additional master integrand
\bq
  \differentialform_{0001}\left[1\right].
\eq
The decomposition of $\Hgen^2_\omega$ with respect to $(W_\bullet,\Fgeom^\bullet)$ is
\begin{center}
\begin{axopicture}(280,140)(0,0)
\Text(70,60)[c]{$0$}
\Text(110,60)[c]{$1$}
\Text(150,60)[c]{$0$}
\Text(90,80)[c]{$1$}
\Text(130,80)[c]{$1$}
\Text(110,100)[c]{$3$}
\DashLine(30,70)(260,70){6}
\DashLine(30,90)(260,90){6}
\DashLine(30,110)(260,110){6}
\Text(253,60)[c]{$W_2$}
\Text(253,80)[c]{$W_3$}
\Text(253,100)[c]{$W_4$}
\Line(260,70)(260,64)
\Line(260,90)(260,84)
\Line(260,110)(260,104)
\DashLine(120,10)(240,130){3}
\DashLine(80,10)(200,130){3}
\DashLine(40,10)(160,130){3}
\Text(110,20)[c]{$\Fgeom^0$}
\Text(70,20)[c]{$\Fgeom^1$}
\Text(30,20)[c]{$\Fgeom^2$}
\Line(120,10)(117,10)
\Line(80,10)(77,10)
\Line(40,10)(37,10)
\end{axopicture}
\end{center}
Ordered by $|\mu|$ we have the following basis for $\Hgen^2_\omega$:
\bq
 \differentialform
 & \in &
 \left\{
  \differentialform_{0001}\left[1\right],
 \;
  \differentialform_{0010}\left[1\right],
 \;
  \differentialform_{0011}\left[z_0^2\right],
 \;
  \differentialform_{1100}\left[1\right],
 \;
  \differentialform_{1010}\left[z_2\right],
 \;
 \differentialform_{0110}[z_0]
 \;
  \right\}.
\eq
For this basis, one finds a differential equation of the form as in eq.~(\ref{refined_statement}).
For the map $\iota : V^2 \rightarrow \Hgen^2_\omega$, we have 
at weight four
\bq
 \iota\left( 4 \eps^4 I_{101110100} \right) & = & \differentialform_{1100}\left[1\right],
 \nonumber \\
 \iota\left( -2 \eps^3 I_{101210100}\left[z_2\right] \right) & = & \differentialform_{1010}\left[z_2\right],
 \nonumber \\
 \iota\left( \eps^3 \left( I_{101120100} + \frac{1}{x_1} I_{101210100}\left[z_2\right]  \right) \right) & = & \differentialform_{0110}\left[z_2\right],
\eq
where we used the notation
\bq
 I_{101210100}\left[z_2\right]
 & = & 
 I_{1\left(-1\right)1210100} + I_{1012101\left(-1\right)0} + \left(x_1-x_2\right) I_{101210100}.
\eq
The pre-image of $\differentialform_{0010}\left[1\right]$ is given by
\bq
 \iota\left( -\eps^3 I_{101210100} \right) & = & \differentialform_{0010}\left[1\right].
\eq
Note that instead of $\differentialform_{0011}\left[z_0^2\right]$
we could have chosen $\differentialform_{0011}\left[\partial_{x_2} P_3 \right]$, 
in the order relation they only differ in the unspecified dots.
The latter integrand has the pre-image
\bq
 \iota\left( -\eps^2 \partial_{x_2} I_{101210100} +\eps^3 \left(\frac{1}{x_2}+\frac{1}{1+x_2}\right)I_{101210100} \right) & = & \differentialform_{0011}\left[\partial_{x_2} P_3 \right].
\eq
We may either replace $\differentialform_{0011}\left[z_0^2\right]$ by $\differentialform_{0011}\left[\partial_{x_2} P_3 \right]$ or express $\differentialform_{0011}\left[\partial_{x_2} P_3 \right]$ in terms of $\differentialform_{0011}\left[z_0^2\right]$ and the other basis elements of
$\Hgen^2_\omega$.

The sector discussed in this example is a sub-sector of the planar double box integral
with three $Z$-boson exchanges contributing to M{\o}ller scattering.
The planar and the non-planar double box integrals with all sub-sectors included will be discussed in detail
in a forthcoming publication.


\subsection{A three-loop contribution to the electron self-energy}
\label{sect:electron_self_energy}

As the next example, we consider a specific three-loop contribution to the electron self-energy~\cite{Duhr:2024bzt}.
We discuss the three-loop banana integral with one massless propagator.
This is an example with non-normal crossing singularities.
The example has three master integrals in the top sector and an elliptic curve associated with
the top sector.
The inverse propagators are defined by
\begin{align}
 \sigma_1 & = -k_1^2,
 &
 \sigma_2 & = -k_2^2 + m^2,
 &
 \sigma_3 & = -k_3^2 + m^2,
 \nonumber \\
 \sigma_4 & = -\left(k_1+k_2+k_3-p\right)^2 + m^2,
 &
 \sigma_5 & = -\left(k_1+k_2-p\right)^2,
 &
 \sigma_6 & = -\left(k_1-p\right)^2,
 \nonumber \\
 \sigma_7 & = -\left(k_1+k_2\right)^2,
 & 
 \sigma_8 & = -\left(k_1+k_3\right)^2,
 &
 \sigma_9 & = -\left(k_2+k_3\right)^2.
\end{align}
\begin{figure}
\begin{center}
\includegraphics[scale=1.0]{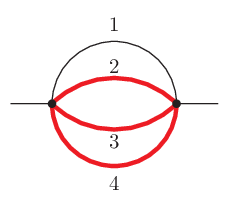}
\end{center}
\caption{
A three-loop contribution to the electron self-energy.
Massive propagators are indicated by red lines.
}
\label{fig:SE_electron}
\end{figure}
We are interested in sector $15$ with inverse propagators $\{\sigma_1,\sigma_2,\sigma_3,\sigma_4\}$.
The graph is shown in fig.~\ref{fig:SE_electron}.
It will be convenient to discuss this example in $D=2-2\eps$ space-time dimensions.
We consider the loop-by-loop Baikov representation, where the innermost loop is formed by the inverse propagators
$\{\sigma_3,\sigma_4\}$,
followed by the loop where we add the inverse propagator $\sigma_2$. For the third and final loop, we add 
the inverse propagator $\sigma_1$.
This yields a two-dimensional Baikov representation with Baikov variables $\sigma_5$ and $\sigma_6$.
We set $x=m^2/p^2$, $z_1=\sigma_5/p^2$ and $z_2=\sigma_6/p^2$.
We have three master integrals in sector $15$ (this is the sector of interest) and one master integral in sector $47$
(this is a super-sector).
In addition, there are two independent symmetry relations in $V^2$, which can be taken as
\bq
\label{example_symmetry_relation}
 2 I_{1 2 1 1 0 0 0 0 0} - I_{1 1 2 1 0 0 0 0 0} - I_{1 1 1 2 0 0 0 0 0}
 & = & 0,
 \nonumber \\
 \frac{d}{dx} \left( 2 I_{1 2 1 1 0 0 0 0 0} - I_{1 1 2 1 0 0 0 0 0} - I_{1 1 1 2 0 0 0 0 0} \right)
 & = & 0.
\eq
Thus we have
\bq
 \dim V^2 \; = \; 3,
 & &
 \dim \Hgen^2_\omega \; = \; 6.
\eq
With
\bq
 r
 & = & 
 \frac{\left(p^2\right)^3 \left(\sigma_6+p^2\right)}{\sigma_1 \sigma_2 \sigma_3 \sigma_4},
\eq
we obtain a Baikov representation with a minimal twist function:
\bq
 e^{3 \eps \Eulerconstant} 
 \left(p^2\right)^{-3+3\eps}
 \int\limits_{{\mathcal C}_{\mathrm{maxcut}}}
 \frac{d^Dk_1}{i \pi^{1-\eps}}
 \frac{d^Dk_2}{i \pi^{1-\eps}}
 \frac{d^Dk_3}{i \pi^{1-\eps}}
 r\left(\sigma\right)
 & = &
 \prebaikov
 \int \frac{dz_1}{\left(2\pi i\right)} \frac{dz_2}{\left(2\pi i\right)} \left. U\left(z\right) \right|_{z_0=1},
\eq
where
\bq
 \prebaikov
 & = & 
 - \frac{2^{6+6\eps} i \pi^{\frac{9}{2}} e^{3 \eps \Eulerconstant}}{\Gamma\left(\frac{1}{2}-\eps\right)^3}.
\eq
Setting
\bq
 \preabs & = & \eps^3
\eq
ensures that $\prebaikov \cdot \preabs$ is pure of transcendental weight zero.
The minimal twist function reads
\bq 
 U\left(z_0,z_1,z_2\right)
 & = &
 P_0^{4\eps}
 P_1^\eps 
 P_2^{-2\eps}
 P_3^{-\frac{1}{2}}
 P_4^{-\frac{1}{2}-\eps}
 P_5^{-\frac{1}{2}-\eps},
\eq
with
\begin{align}
 P_0 & = z_0, & P_3 & = z_1, \nonumber \\
 P_1 & = z_2, & P_4 & = z_1 + 4 x z_0, \nonumber \\
 P_2 &= z_2+z_0, & P_5 & = \left(z_2-z_1\right)^2 +2 x z_0 \left( z_1+z_2\right) + x^2 z_0^2.
\end{align}
We have $d_U=-2$, hence $\hat{\Phi}$ must be homogeneous of degree $-1$.
We have three even polynomials ($I_{\mathrm{even}}^0=\{0,1,2\}$) and we consider in turn 
the localisations on each of the even polynomials.
As we have two Baikov variables, we have to consider for each of those possibilities the localisation on a further polynomial.
As in the previous example, we keep the notation compact
and use whenever there is no ambiguity
the pre-images $\differentialform_{\mu_0 \mu_1 \mu_2 \mu_3}[Q]$
instead of the more lengthy notation $\mathrm{Res}_{{\mathcal P}_1,\dots,{\mathcal P}_r} \differentialform_{\mu_0 \mu_1 \mu_2 \mu_3}[Q,I,\Divisor_s]$.

We start with the localisation at $P_0=0$.
From the analysis of the critical points, we expect one master integrand from the localisation on $P_0=0$.
With the choice of $\Divisor_1$ as a scale polynomial, we obtain the (non-minimal) modified twist function
\bq
 U^{\mathrm{loc}}_{\langle \Divisor_0 \rangle}
 & = &
 P_1^{5\eps} 
 P_2^{-2\eps}
 P_3^{-\frac{1}{2}}
 P_4^{-\frac{1}{2}-\eps}
 P_5^{-\frac{1}{2}-\eps}.
\eq
We have
\bq
 \left. U^{\mathrm{loc}}_{\langle \Divisor_0 \rangle} \right|_{z_0=0}
 & = &
 z_2^{3\eps} 
 z_1^{-1-\eps}
 \left(z_2-z_1\right)^{-1-2\eps}.
\eq
We may take a second residue at the points
\begin{align}
 \left[0:0:1\right],
 &&
 \left[0:1:1\right],
 &&
 \left[0:1:0\right].
\end{align}
The differential form $\differentialform_{100000}\left[1\right]$ has a non-zero residue
on the first two points, 
the differential form $\differentialform_{110000}\left[z_1\right]$
has a non-zero residue at the third point (and at the second point).
Therefore we obtain from the three points two candidates 
\bq
\label{masters_electron_SE_localisation_P0}
 \differentialform_{100000}\left[1\right],
 \;
 \differentialform_{110000}\left[z_1\right],
\eq
to which we assign $a=-4$.

Running then in dimension one (i.e. at weight $3$) will find a relation among those and eliminate
one candidate. 
In summary, we get one master from this localisation, which we may take as
\bq
 \differentialform_{100000}\left[1\right]
\eq
with $|\mu|=1$, $r=2$ and $o=2$.

Next, we consider the localisation at $\langle P_1 \rangle$.
The analysis of the critical points indicates two master integrands.
We choose $\Divisor_0$ as a scale polynomial and obtain the (non-minimal) modified twist function
\bq
 U^{\mathrm{loc}}_{\langle \Divisor_1 \rangle}
 & = &
 P_0^{5\eps}
 P_2^{-2\eps}
 P_3^{-\frac{1}{2}}
 P_4^{-\frac{1}{2}-\eps}
 P_5^{-\frac{1}{2}-\eps}.
\eq
We have
\bq
 \left. U^{\mathrm{loc}}_{\langle \Divisor_1 \rangle} \right|_{z_2=0}
 & = &
 z_0^{3\eps}
 z_1^{-\frac{1}{2}}
 \left(z_1+4xz_0\right)^{-\frac{1}{2}-\eps}
 \left(z_1+xz_0\right)^{-1-2\eps}.
\eq
We may take a second residue at the points
\begin{align}
 \left[0:1:0\right],
 &&
 \left[1:-x:0\right].
\end{align}
The differential form $\differentialform_{110000}[z_1]$ has a two-fold non-zero residue at $[ 0:1:0 ]$.
We choose this differential form as a master integrand on the point $[ 0:1:0 ]$.
This differential form also has a two-fold non-zero residue at the other point
$[ 1:-x:0 ]$, however, at this point, there is a simpler choice, given by the differential form
$\differentialform_{010000}[1]$.
The algorithm will pick the latter.
From dimension zero, we obtain these two candidates.
Running then in dimension one, we find that there are no further linear relations and
these two differential forms constitute a basis for the localisation on $\langle P_1 \rangle$.
In summary, we get two masters from the localisation $\Divisor_1=0$, given by
\bq
 \differentialform_{010000}\left[1\right],
 \;
 \differentialform_{110000}\left[z_1\right].
\eq
The former differential form has the values $(|\mu|,r,o)=(1,2,2)$, the latter
has the values $(|\mu|,r,o)=(2,2,2)$.

We then consider the localisation at $\langle P_2 \rangle$.
An analysis of the critical points indicates three master integrands at $\langle P_2 \rangle$.
In order to reduce polynomials modulo the ideal $\langle P_2 \rangle=\langle z_0+z_2 \rangle$,
we need to choose a monomial order for $(z_0,z_1,z_2)$.
Let's assume that we eliminate $z_2$ in favour of $z_0$, e.g. $z_2=-z_0$ on $\langle P_2 \rangle$.
We choose $\Divisor_0$ as a scale polynomial and obtain, up to a prefactor, the modified twist function
\bq
 U^{\mathrm{loc}}_{\langle \Divisor_2 \rangle}
 & = &
 P_0^{2\eps}
 P_1^\eps 
 P_3^{-\frac{1}{2}}
 P_4^{-\frac{1}{2}-\eps}
 P_5^{-\frac{1}{2}-\eps}.
\eq
We have
\bq
 \left. U^{\mathrm{loc}}_{\langle \Divisor_2 \rangle} \right|_{z_2=-z_0}
 & = & 
 z_0^{3\eps}
 z_1^{-\frac{1}{2}}
 \left(z_1+4xz_0\right)^{-\frac{1}{2}-\eps}
 \left[ \left(z_0+z_1\right)^2 +2 x z_0 \left( z_1-z_0\right) + x^2 z_0^2 \right]^{-\frac{1}{2}-\eps}.
\eq
At the point $[ 0:1:0 ]$, we may take a second residue, and we obtain for this point the master integrand
\bq
 \differentialform_{101000}\left[z_1\right],
\eq
with $|\mu|=2$, $r=2$ and $o=2$. We assign $a=-4$ to this integrand.
The remaining two master integrands on the localisation $\Divisor_2=0$ are related to an elliptic curve and 
are obtained by running the Laporta algorithm on the localisation $\Divisor_2=0$.
It is instructive to discuss this in detail:
The first ``elliptic'' integrand is straightforward:
\bq
 \differentialform_{001000}\left[1\right],
\eq
and has $|\mu|=1$, $r=1$ and $o=1$.
For the second one, there is a choice.
The candidates should have $|\mu|=2$, $r=1$ and $o=2$.
If we further require that the numerator polynomial has minimal degree, we need to consider
\bq
 \differentialform_{001100}\left[Q\right], \differentialform_{001010}\left[Q\right]
\eq
with $\deg Q=1$.
We require that the candidates have poles of order at most two on ${\mathbb C}{\mathbb P}^2$, not just on the localisation.
This excludes, for example
\bq
 \differentialform_{001100}\left[z_0\right],
\eq
which has a pole of order three (after blow-up) at $z_1=0, z_2=-x z_0$.
We take
\bq
 \differentialform_{001010}\left[z_0\right]
\eq
as the second master integrand related to the elliptic curve. This integrand has $|\mu|=2$, $r=1$ and $o=2$.

We have now done all localisation. 
In the next step, we merge the preferred candidates from the localisations of dimension $1$.
We switch back to the longer notation 
\bq
 \mathrm{Res}_{{\mathcal P}_1,\dots,{\mathcal P}_r} \differentialform_{\mu_0 \mu_1 \mu_2 \mu_3 \mu_4 \mu_5}[Q,I,\Divisor_s].
\eq
The $1$-skeleton is given by
\bq
 V\left(\langle \Divisor_0 \rangle\right)
 \cup
 V\left(\langle \Divisor_1 \rangle\right)
 \cup
 V\left(\langle \Divisor_2 \rangle\right).
\eq
The $0$-skeleton is given by the four points
\bq
 \left[0:0:1\right]
 \cup
 \left[0:1:1\right]
 \cup
 \left[0:1:0\right]
 \cup
 \left[1:-x:0\right].
\eq
The intersections of the components of the $1$-skeleton is given by a single point
\bq
 \bigcup\limits_{(ij)} V({\mathcal I_{ij}})
 & = &
 C_1
 \; = \; 
 \left[0:1:0\right].
\eq
We have $S_1=\{0,1,2\}$.
The skeleton is sketched in fig.~\ref{fig:SE_electron_skeleton}.
\begin{figure}
\begin{center}
\includegraphics[scale=1.0]{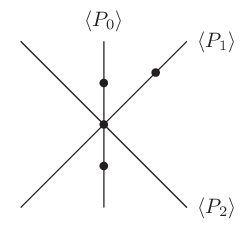}
\end{center}
\caption{
A sketch of the $1$-skeleton and the $0$-skeleton.
}
\label{fig:SE_electron_skeleton}
\end{figure}
The bases on the localisations of dimension $1$ are
\bq
\label{objects_from_localisations_electron}
 V\left(\langle \Divisor_0 \rangle\right):
 & &
 \mathrm{Res}_{\Divisor_0} \differentialform_{100000}\left[1,\langle \Divisor_0 \rangle,\Divisor_1\right]
 \\
 V\left(\langle \Divisor_1 \rangle\right):
 & &
 \mathrm{Res}_{\Divisor_1} \differentialform_{010000}\left[1,\langle \Divisor_1 \rangle,\Divisor_0\right],
 \;
 \mathrm{Res}_{\Divisor_1} \differentialform_{110000}\left[z_1,\langle \Divisor_1 \rangle,\Divisor_0\right],
 \nonumber \\
 V\left(\langle \Divisor_2 \rangle\right):
 & &
 \mathrm{Res}_{\Divisor_2} \differentialform_{101000}\left[z_1,\langle \Divisor_2 \rangle,\Divisor_0\right],
 \;
 \mathrm{Res}_{\Divisor_2} \differentialform_{001000}\left[1,\langle \Divisor_2 \rangle,\Divisor_0\right],
 \;
 \mathrm{Res}_{\Divisor_2} \differentialform_{001010}\left[z_0,\langle \Divisor_2 \rangle,\Divisor_0\right].
 \nonumber
\eq
These are six different objects.
Let's assume that we have chosen $\differentialform_{110000}[z_1]$ as a master integrand on the point $[0:1:0]$.
This differential form has a residue on $\langle P_0 \rangle$ and $\langle P_1 \rangle$, but not
on $\langle P_2 \rangle$. 
From the intersections of the components of the $1$-skeleton, we get, therefore, only one equation:
\bq
 \left[0:1:0\right]:
 & &
 \mathrm{Res}_{\Divisor_0} \differentialform_{110000}\left[z_1,\langle \Divisor_0 \rangle,\Divisor_1\right]
 \; = \;
 \mathrm{Res}_{\Divisor_1} \differentialform_{110000}\left[z_1,\langle \Divisor_1 \rangle,\Divisor_0\right].
\eq
$\mathrm{Res}_{\Divisor_0} \differentialform_{110000}\left[z_1,\langle \Divisor_0 \rangle,\Divisor_1\right]$
is not a master integrand on $\langle \Divisor_0 \rangle$.
We add one equation, which relates this integrand to the master integrand on $\langle \Divisor_0 \rangle$:
\bq
 c_1 \mathrm{Res}_{\Divisor_0} \differentialform_{110000}\left[z_1,\langle \Divisor_0 \rangle,\Divisor_1\right]
 +
 c_2 \mathrm{Res}_{\Divisor_0} \differentialform_{100000}\left[1,\langle \Divisor_0 \rangle,\Divisor_1\right]
 & = & 0.
\eq
The coefficients are obtained from the reduction on $\langle \Divisor_0 \rangle$.
The actual values do not matter, for us it is only relevant that $c_1 \neq 0$ and $c_2 \neq 0$.
We now have seven objects and two equations, leaving us with five independent integrands.
After renaming, we have
\bq
 \differentialform_{100000}\left[1\right],
 \;
 \differentialform_{010000}\left[1\right],
 \;
 \differentialform_{101000}\left[z_1\right],
 \;
 \differentialform_{001000}\left[1\right],
 \;
 \differentialform_{001010}\left[z_0\right].
\eq
We then run the full Laporta algorithm and find the last master integrand:
\bq
 \differentialform_{000010}\left[1\right].
\eq
Note that the algorithm excludes 
the candidate $\differentialform_{000100}\left[1\right]$, as this integrand has a pole of order three (after blow-up) at $z_1=0, z_2=-x z_0$.
In summary, we find the basis
\bq
 \mathrm{weight} \; 4:
 & &
  \differentialform_{100000}\left[1\right],
  \differentialform_{010000}\left[1\right],
  \differentialform_{101000}\left[z_1\right],
 \nonumber \\
 \mathrm{weight} \; 3:
 & &
 \differentialform_{001000}\left[1\right],
 \differentialform_{001010}\left[z_0\right],
 \nonumber \\
 \mathrm{weight} \; 2:
 & &
 \differentialform_{000010}\left[1\right].
\eq
The decomposition of $\Hgen^2_\omega$ with respect to $(W_\bullet,\Fgeom^\bullet)$ is
\begin{center}
\begin{axopicture}(280,140)(0,0)
\Text(70,60)[c]{$0$}
\Text(110,60)[c]{$1$}
\Text(150,60)[c]{$0$}
\Text(90,80)[c]{$1$}
\Text(130,80)[c]{$1$}
\Text(110,100)[c]{$3$}
\DashLine(30,70)(260,70){6}
\DashLine(30,90)(260,90){6}
\DashLine(30,110)(260,110){6}
\Text(253,60)[c]{$W_2$}
\Text(253,80)[c]{$W_3$}
\Text(253,100)[c]{$W_4$}
\Line(260,70)(260,64)
\Line(260,90)(260,84)
\Line(260,110)(260,104)
\DashLine(120,10)(240,130){3}
\DashLine(80,10)(200,130){3}
\DashLine(40,10)(160,130){3}
\Text(110,20)[c]{$\Fgeom^0$}
\Text(70,20)[c]{$\Fgeom^1$}
\Text(30,20)[c]{$\Fgeom^2$}
\Line(120,10)(117,10)
\Line(80,10)(77,10)
\Line(40,10)(37,10)
\end{axopicture}
\end{center}
We order the basis by 
$|\mu|$
\bq
 \left\{
  \differentialform_{000010}\left[1\right],
  \differentialform_{001000}\left[1\right],
  \differentialform_{100000}\left[1\right],
  \differentialform_{010000}\left[1\right],
  \differentialform_{001010}\left[z_0\right],
  \differentialform_{101000}\left[z_1\right]
 \right\}
\eq
and verify that it leads to a $\Fcomb^\bullet$-compatible differential equation.

Before proceeding to step $2$, it is advantageous to map back to the space $V^2$ of Feynman integrals.
The integrand on the left-hand side of the first symmetry relation in eq.~(\ref{example_symmetry_relation})
corresponds in $\Agen_\omega^2$ to
\bq
 \differentialform_{000001}\left[z_1+z_2+x z_0\right] - \differentialform_{000010}\left[1\right].
\eq
The symmetry relation tells us that this integrand integrates to zero.
In $\Hgen_\omega^2$ we may reduce the integrand by integration -by-parts identities and we obtain
the relation
\bq
\label{symmetry_1}
 6x \differentialform_{000010}\left[1\right]
 + \differentialform_{001000}\left[1\right]
 + \differentialform_{100000}\left[1\right]
 & = & 0.
\eq
We may proceed in a similar way with the second symmetry relation in eq.~(\ref{example_symmetry_relation}). We obtain a second symmetry relation in $\Hgen_\omega^2$, which we may write 
(by adding a suitable multiple of the first one) as
\bq
\label{symmetry_2}
 4\left(1+3x\right) \differentialform_{001000}\left[1\right]
 + 2 \differentialform_{100000}\left[1\right]
 + 3 \differentialform_{101000}\left[z_1\right]
 & = & 0.
\eq
Note that these two equations are independent of $\eps$. 
We may use these equations to eliminate in each equation the integrand with the highest $\absmu$-value,
this will preserve the compatibility with the $\Fcomb^\bullet$-filtration.
Thus we eliminate $\differentialform_{101000}\left[z_1\right]$ from the second equation.
For the first equation there is a choice, as all differential forms have $\absmu=1$.
Let us assume that we eliminate $\differentialform_{000010}\left[1\right]$.

We may take $\differentialform_{010000}\left[1\right]$ as the master integrand for the super sector $47$.
This leaves us with the integrands
\bq
 \left\{
  \differentialform_{001000}\left[1\right],
  \differentialform_{100000}\left[1\right],
  \differentialform_{001010}\left[z_0\right]
 \right\}
\eq
for sector $15$.
For the map $\iota : V^2 \rightarrow \Hgen^2_\omega$, we have 
\bq
 \iota\left( - 2 \eps^3 I_{111100000} \right) & = & \differentialform_{001000}\left[1\right],
 \nonumber \\
 \iota\left( 4 \eps^3 \left( I_{111100000} + I_{11110\left(-1\right)000} \right) \right) & = & \differentialform_{100000}\left[1\right],
 \nonumber \\
 \iota\left( \eps^2 I_{111200000} \right) & = &  \differentialform_{001010}\left[z_0\right].
\eq
We therefore set
\bq
 J_1
 & = &
 - 2 \eps^3 I_{111100000},
 \nonumber \\
 J_2
 & = & 
 4 \eps^3 \left( I_{111100000} + I_{11110\left(-1\right)000} \right),
 \nonumber \\
 J_3
 & = & 
 \eps^2 I_{111200000}.
\eq
For $J=(J_1,J_2,J_3)^T$ one obtains the differential equation in the form of eq.~(\ref{A_reorganised})
\bq
 d_B J \; = \; A J,
 & & 
 A \; = \; B^{(-1)} + B^{(0)} + B^{(1)}.
\eq
The matrices $B^{(-1)}$, $B^{(0)}$ and $B^{(1)}$ are explicitly given by
\bq
 B^{(-1)}
 & = &
 \left( \begin{array}{cc|c}
  0 & 0 & 0  \\
  0 & 0  & 0  \\
 \hline
 -\frac{3x-1}{2 \eps x\left(x-1\right)\left(9x-1\right)} & 0 & - \frac{27x^2-20x+1}{x\left(x-1\right)\left(9x-1\right)} \\
 \end{array} \right),
 \nonumber \\
 B^{(0)}
 & = & 
 \left( \begin{array}{cc|c}
  0 & 0 & 0 \\
  0 & 0 & 0 \\
 \hline
  -\frac{21x-5}{2 x\left(x-1\right)\left(9x-1\right)} & 0 & 0 \\
 \end{array} \right),
 \nonumber \\
 B^{(1)}
 & = &
 \left( \begin{array}{cc|c}
  0 & 0 & 6 \eps \\
  -\frac{4\eps}{x} & -\frac{4\eps}{x} & 0 \\
 \hline
  -\frac{\left(18x^2-3x+1\right)\eps}{x^2\left(x-1\right)\left(9x-1\right)} & -\frac{\eps}{x^2\left(x-1\right)\left(9x-1\right)} & -\frac{\left(63x^2-30x-1\right)\eps}{x\left(x-1\right)\left(9x-1\right)}\\
 \end{array} \right).
\eq
In these matrices, we indicated the block structure due to the $\Fcomb^\bullet$-filtration.
We then rotate the system to an $\eps$-form with the rotation matrix
\bq
 R_2 & = & R_2^{(-1)} R_2^{(0)}.
\eq
The ansatz for $R_2^{(-1)}$ and $R_2^{(0)}$ is
\begin{align}
 R_2^{(-1)}
 & = 
 \left( \begin{array}{cc|c}
  R^{(-1)}_{11} & R^{(-1)}_{12} & 0  \\
  R^{(-1)}_{21} & R^{(-1)}_{22} & 0 \\
 \hline
  \frac{1}{\eps} R^{(-1)}_{31} & \frac{1}{\eps} R^{(-1)}_{32} & R^{(-1)}_{33} \\
 \end{array} \right),
 &
 R_2^{(0)}
 & = 
 \left( \begin{array}{cc|c}
  1 & 0 & 0 \\
  0 & 1 & 0 \\
 \hline
  R^{(0)}_{31} & R^{(0)}_{32} & 1 \\
 \end{array} \right).
\end{align}
Note that the matrix $R_2^{(-1)}$ is required to be invertible. 
We first consider the differential equation for $R_2^{(-1)}$,
For this matrix we start with the first block-column, involving the six unknown functions
$R^{(-1)}_{11}$, $R^{(-1)}_{12}$, $R^{(-1)}_{21}$, $R^{(-1)}_{22}$, $R^{(-1)}_{31}$ and $R^{(-1)}_{32}$.
We immediately obtain
\bq
 \frac{d}{dx} R^{(-1)}_{21} \; = \; 0,
 & &
 \frac{d}{dx} R^{(-1)}_{22} \; = \; 0.
\eq
Thus these two functions are constant and we can choose $R^{(-1)}_{21}=0$ and $R^{(-1)}_{22}=1$.
Note that we cannot choose $R^{(-1)}_{21}=0$ and $R^{(-1)}_{22}=0$, as the matrix $R_2^{(-1)}$ would then not be invertible.
For the four remaining functions of the first-block column we obtain
\begin{align}
 \frac{d}{dx} R^{(-1)}_{1j} & = 6 R^{(-1)}_{3j},
 &
 \frac{d}{dx} R^{(-1)}_{3j}
 & =
 -\frac{3x-1}{2 x\left(x-1\right)\left(9x-1\right)} R^{(-1)}_{1j} - \frac{27x^2-20x+1}{x\left(x-1\right)\left(9x-1\right)} R^{(-1)}_{3j},
 &
 j \in \{1,2\}.
\end{align}
For $R^{(-1)}_{12}$ and $R^{(-1)}_{32}$ (or $R^{(-1)}_{11}$ and $R^{(-1)}_{31}$) we are free to pick the trivial solution $R^{(-1)}_{12}=0$ and $R^{(-1)}_{32}=0$,
but we cannot do this simultaneously for the second pair, otherwise the matrix $R_2^{(-1)}$ would not be invertible.
Eliminating $R^{(-1)}_{31}$ we obtain a second-order differential equation for $R^{(-1)}_{11}$:
\bq
\label{Picard_Fuchs}
 \left[
 \frac{d^2}{dx^2} 
 + \frac{27x^2-20x+1}{x\left(x-1\right)\left(9x-1\right)} \frac{d}{dx} 
 + \frac{3\left(3x-1\right)}{x\left(x-1\right)\left(9x-1\right)} 
 \right] R^{(-1)}_{11}
 & = &
 0.
\eq
This is the Picard-Fuchs equation of a family of elliptic curves, parameterised by $x$.
The function $R^{(-1)}_{11}$ is therefore a period of the elliptic curve.

We then turn to the second block-column. In this block-column, there is only one unknown function $R^{(-1)}_{33}$.
The differential equation for $R^{(-1)}_{33}$ reads
\bq
 \frac{d}{dx} R^{(-1)}_{33}
 & = &
 \left[ -\frac{27x^2-20x+1}{x\left(x-1\right)\left(9x-1\right)} - \frac{d}{dx} \ln R^{(-1)}_{11} \right] R^{(-1)}_{33}
\eq
A possible solution is
\bq
 R^{(-1)}_{33} & = & \frac{1}{x\left(x-1\right)\left(9x-1\right)R^{(-1)}_{11}}.
\eq
Thus the matrix $R_2^{(-1)}$ reads
\bq
 R_2^{(-1)}
 & = &
 \left( \begin{array}{cc|c}
  R^{(-1)}_{11} & 0 & 0  \\
  0 & 1 & 0 \\
 \hline
  \frac{1}{6\eps} \frac{d}{dx} R^{(-1)}_{11} & 0 & \frac{1}{x\left(x-1\right)\left(9x-1\right)R^{(-1)}_{11}} \\
 \end{array} \right),
\eq
where $R^{(-1)}_{11}$ is determined by eq.~(\ref{Picard_Fuchs}).

Rotating by $R_2^{(-1)}$ we obtain a matrix $\tilde{A}^{(0)}$, which contains terms of $B$-order $0$ and $1$.
The terms of $B$-order $0$ are then removed by the rotation with the matrix $R_2^{(0)}$.
The matrix $R_2^{(0)}$ contains two unknown functions 
$R^{(0)}_{31}$ and $R^{(0)}_{32}$. The differential equation for the latter is trivial
\bq
 \frac{d}{dx} R^{(0)}_{32} & = & 0
\eq
and we may set $R^{(0)}_{32}=0$.
The former satisfies the differential equation
\bq
\label{eq_R0_31}
 \frac{d}{dx} R^{(0)}_{31}
 & = &
 - \frac{1}{6} R^{(-1)}_{11} \left[ \left(63x-15\right) R^{(-1)}_{11} + \left(63x^2-30x-1\right) \frac{d}{dx} R^{(-1)}_{11} \right],
\eq
which can be solved by direct integration.
One obtains
\bq
\label{sol_R0_31}
 R^{(0)}_{31}
 & = &
 - \frac{1}{12} \left(63x^2-30x-1\right) \left( R^{(-1)}_{11} \right)^2.
\eq
Note that we may freely choose the integration constant.
Thus $R_2^{(0)}$ reads
\bq
 R_2^{(0)}
 & = &
 \left( \begin{array}{cc|c}
  1 & 0 & 0 \\
  0 & 1 & 0 \\
 \hline
  R^{(0)}_{31} & 0 & 1 \\
 \end{array} \right),
\eq
with $R^{(0)}_{31}$ given by eq.~(\ref{sol_R0_31}).
After rotation by $R_2^{(0)}$ we then obtain the $\eps$-factorised form.
In this example there is nothing to be done for the tadpole sub-sector.

The tadpole integral is rather simple. We set
\bq
 K_0 & = & \eps^3 I_{011100000}.
\eq
The master integrals integrals $K_1$, $K_2$ and $K_3$ for the sector $15$ are obtained 
by the rotating $R_2= R_2^{(-1)}R_2^{(0)}$ from the basis $J$.
This yields
\bq
 K_1
 & = &
 \frac{J_1}{R^{(-1)}_{11}},
 \nonumber \\
 K_2
 & = &
 J_2,
 \nonumber \\
 K_3
 & = &
 \left[ 
  \frac{1}{12} \left(63x^2-30x-1\right) R^{(-1)}_{11}
  - \frac{1}{6\eps} x \left(x-1\right) \left(9x-1\right) \frac{d}{dx} R^{(-1)}_{11}
 \right] J_1
 \nonumber \\
 & & 
 + x \left(x-1\right) \left(9x-1\right) R^{(-1)}_{11} J_3.
\eq
$K=(K_0,K_1,K_2,K_3)^T$ is a basis of the full system with an $\eps$-factorised differential equation
\bq
 d_B K & = & \eps A K,
\eq
where
\bq
 A
 =
 \left( \begin{array}{cccc}
 -\frac{3}{x} & 0 & 0 & 0 \\
 0 & -\frac{\left(63x^2-30x-1\right)}{2x \left(x-1\right) \left(9x-1\right)} & 0 & \frac{6}{x \left(x-1\right) \left(9x-1\right)\left( R^{(-1)}_{11} \right)^2} \\
 \frac{4}{x} & -\frac{4R^{(-1)}_{11}}{x} & -\frac{4}{x} & 0 \\
 \frac{2 R^{(-1)}_{11}}{x} & \frac{\left(81x^4+1188x^3-594x^2+372x-23\right)\left( R^{(-1)}_{11} \right)^2}{24x \left(x-1\right) \left(9x-1\right)} & -\frac{R^{(-1)}_{11}}{x} & -\frac{\left(63x^2-30x-1\right)}{2x \left(x-1\right) \left(9x-1\right)}  \\
 \end{array} \right).
 \;\;\;
\eq

\subsection{The four-loop equal-mass banana integral}
\label{sect:four_loop_banana}

\begin{figure}[!htp]
\begin{center}
\includegraphics[scale=1.0]{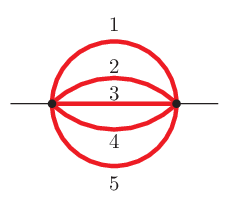}
\end{center}
\caption{
The four-loop equal mass banana integral. Red lines are massive propagators.
}
\label{fig:four_loop_banana}
\end{figure}

As our final example, we consider the four-loop equal mass banana integral~\cite{Pogel:2022ken}.
This is an example with a three-dimensional Baikov representation.
The associated geometry is a Calabi-Yau threefold.
The inverse propagators are defined by
\begin{align}
 \sigma_1 & = -k_1^2 + m^2,
 &
 \sigma_2 & = -k_2^2 + m^2,
 &
 \sigma_3 & = -k_3^2 + m^2,
 \nonumber \\
 \sigma_4 & = -k_4^2 + m^2,
 &
 \sigma_5 & = -\left(k_{1234}-p\right)^2 + m^2,
 &
 \sigma_6 & = -\left(k_{123}-p\right)^2,
 \nonumber \\
 \sigma_7 & = -\left(k_{12}-p\right)^2,
 &
 \sigma_8 & = -\left(k_1-p\right)^2,
 &
 \sigma_9 & = -k_{12}^2,
 \nonumber \\
 \sigma_{10} & = -k_{13}^2,
 &
 \sigma_{11} & = -k_{23}^2,
 &
 \sigma_{12} & = -k_{14}^2,
 \nonumber \\
 \sigma_{13} & = -k_{24}^2,
 &
 \sigma_{14} & = -k_{34}^2,
\end{align}
with $k_{ij}=k_i+k_j$, $k_{ijk}=k_{ij}+k_k$ and $k_{ijkl}=k_{ijk}+k_l$. The graph is shown in fig.~\ref{fig:four_loop_banana}.
It will be convenient to discuss this example in $D=2-2\eps$ space-time dimensions.

\begin{figure}[!htp]
\begin{center}
\includegraphics[scale=0.88]{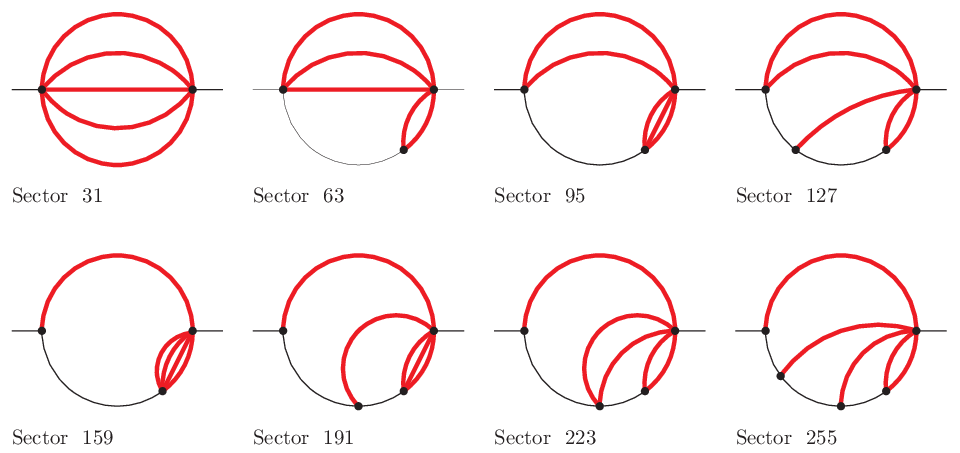}
\end{center}
\caption{
The sector of the equal-mass banana integral at four loops together with the relevant super-sectors.
}
\label{fig:banana4_super_sectors}
\end{figure}

Following the notation of sunrise and three-loop banana integrals, we define $x=m^2/(-p^2),\,\sigma_6=-p^2 z_1,\,\sigma_7= -p^2 z_2, \sigma_8=-p^2 z_3$. The loop-by-loop Baikov representation under the maximal cut of sector 31 (with inverse propagator $\{\sigma_1, \sigma_2, \sigma_3, \sigma_4, \sigma_5\}$) reads
\begin{equation}
	\begin{aligned}
		e^{4\eps\gamma_E} (-s)^{1+4\eps}\int\limits_{\mathcal{C}_{\rm maxcut}} \left(\prod_{i=1}^4\frac{d^D k_i}{i\pi^{1-\eps}}\right) \frac{1}{\sigma_1\sigma_2\sigma_3\sigma_4\sigma_5} = \prebaikov \int \left(\prod_{i=1}^3\frac{d z_i}{2\pi i}\right) \left. U\left(z\right)\right|_{z_0=1}, 
	\end{aligned}
\end{equation}
with 
\begin{equation}
	\prebaikov = e^{4\eps\gamma_E}\frac{2^{8+8\eps}\pi^6}{\Gamma\left(\frac{1}{2}-\eps\right)^{4}}. 
\end{equation}
Setting $\preabs=\eps^4$ ensures that $\prebaikov \cdot \preabs$ is of uniform transcendental weight zero. 

The minimal twist function $U$ for four loops reads
\begin{equation}
    \label{eq:banana4Ltwist}
	U(z_0, z_1, z_2, z_3) = P_0^{\,5\eps}\, P_2^{\,\eps}\, P_3^{\,\eps}\,  P_1^{-\frac{1}{2}}\,  P_4^{-\frac{1}{2}-\eps}\,  P_5^{-\frac{1}{2}-\eps}\,  P_6^{-\frac{1}{2}-\eps}\,  P_7^{-\frac{1}{2}-\eps},
\end{equation}
where
\begin{equation}
    \label{eq:banana4Ldivisors}
	\begin{aligned}
		P_0 &= z_0,\quad P_1 = z_1,\quad P_2=z_2,\quad P_3=z_3,\quad P_4 = z_1 + 4x\,z_0,\\
		P_5 &= \big(z_2 - z_1\big)^2 + 2x (z_2+z_1)z_0 +x^2 z_0^2,\\
		P_6 &= \big(z_3 - z_2\big)^2 + 2x (z_3+z_2)z_0 +x^2 z_0^2,\\
		P_7 &= z_3^2 -2(1-x)z_3 z_0 + (1+x)^2 z_0^2.
	\end{aligned}
\end{equation}
One can generalise the minimal twist function to $l$ loops as
\begin{equation}
    \label{eq:bananalloopstwist}
	U_l(z_0, z_1,\cdots, z_{l-1}) = P_0^{(l+1)\eps}\, P_1^{-\frac{1}{2}}\, P_l^{-\frac{1}{2}-\eps}\left( \,\prod_{i= 2}^{l-1}  P_{i}^{\eps}\, P_{l-1+i}^{-\frac{1}{2}-\eps}\,\right) P_{2l-1}^{-\frac{1}{2}-\eps},
\end{equation}
where
\begin{equation}
	\begin{aligned}
		P_0 &= z_0,\quad P_1 = z_1,\quad P_l = z_1 + 4x z_0,\\
		P_{2l-1} &= z_{l-1}^2 - 2(1-x)z_{l-1} z_0 +(1+x)^2z_0^2,
	\end{aligned}
\end{equation}
and
\begin{equation}
	\begin{aligned}
	    P_i &= z_i, \\
	    P_{l-1+i} &= \left(z_i-z_{i-1}\right)^2+2x\left(z_i + z_{i-1}\right)z_0 + x^2 z_0^2,\qquad 2\leq i\leq l-1.
	\end{aligned}
\end{equation}
The subscript in the twist in \eqref{eq:bananalloopstwist} manifests the loop number. One can find a generalisation for the unequal-mass cases in a similar way. 

In the equal-mass four-loop case, $d_U = -4$, hence $\hat{\Phi}$ must be homogeneous of degree 0.  We can read off from eq.~\eqref{eq:banana4Ltwist} that there are three even polynomials $I_{\rm even}^0 = \{0, 2, 3\}$ and five odd polynomials $I_{\rm odd}^0 = \{1, 4, 5, 6, 7\}$. The dimensions of $V^3$ and $H_\omega^3$ do not match, and we find that 
\begin{equation}
	\dim V^3 = 4,\quad \dim H_\omega^3 = 25.
\end{equation}
The mismatch is due to super-sectors and symmetry relations.
Figure~\ref{fig:banana4_super_sectors} shows all relevant super-sectors.
We find that the super-sector 95 has 3 extra master integrals, and that super-sector 159 has one extra master integral. 
All the other super-sectors have no extra master integrals. 

The four-loop equal-mass banana integral has a high degree of symmetry 
and, as a consequence, there are many symmetry relations (to be concrete, seventeen in this example).
This explains the high dimension of $H_\omega^3$.
As an outlook, let us note that one may incorporate symmetries directly in twisted cohomology following the ideas of \cite{Gasparotto:2023roh,Duhr:2025xyy}.
This will improve the efficiency in examples which have a high degree of symmetry.

The decomposition of $V^3$ with respect to $(W_\bullet,\Fgeom^\bullet)$ is
\begin{center}
\begin{axopicture}(310,160)(0,0)
\Text(70,60)[c]{$1$}
\Text(110,60)[c]{$1$}
\Text(150,60)[c]{$1$}
\Text(190,60)[c]{$1$}
\Text(90,80)[c]{$0$}
\Text(130,80)[c]{$0$}
\Text(170,80)[c]{$0$}
\Text(110,100)[c]{$0$}
\Text(150,100)[c]{$0$}
\Text(130,120)[c]{$0$}
\DashLine(30,70)(300,70){6}
\DashLine(30,90)(300,90){6}
\DashLine(30,110)(300,110){6}
\DashLine(30,130)(300,130){6}
\Text(293,60)[c]{$W_3$}
\Text(293,80)[c]{$W_4$}
\Text(293,100)[c]{$W_5$}
\Text(293,120)[c]{$W_6$}
\Line(300,70)(300,64)
\Line(300,90)(300,84)
\Line(300,110)(300,104)
\Line(300,130)(300,124)
\DashLine(160,10)(290,140){3}
\DashLine(120,10)(250,140){3}
\DashLine(80,10)(210,140){3}
\DashLine(40,10)(170,140){3}
\Text(150,20)[c]{$\Fgeom^0$}
\Text(110,20)[c]{$\Fgeom^1$}
\Text(70,20)[c]{$\Fgeom^2$}
\Text(30,20)[c]{$\Fgeom^3$}
\Line(160,10)(157,10)
\Line(120,10)(117,10)
\Line(80,10)(77,10)
\Line(40,10)(37,10)
\end{axopicture}
\end{center}
Step $2$ of the algorithm proceeds similarly to the calculation of ref.~\cite{Pogel:2022ken}.


\section{Conclusions}
\label{sect:conclusions}

In this work, we provided a detailed description of an algorithm outlined in ref.~\cite{e-collaboration:2025frv}.
This algorithm identifies a basis of master integrals whose differential equations are $\eps$-factorised. 
The approach has the advantage of not requiring any previous knowledge of the underlying geometry associated with the maximally cut Feynman integral. 
The algorithm consists of two steps: first, the selection of master integrands that satisfy differential equations in Laurent polynomial form in $\eps$ and compatible with a filtration, 
and second, the rotation of these differential equations to $\eps$-factorised form. 

In the first step, we perform integration-by-parts reduction within the framework of twisted cohomology, supplemented with a specific ordering criterion in the Laporta algorithm. 
We observe that for all tested cases, the corresponding system of differential equations exhibits a Laurent polynomial form with respect to the regulator $\eps$. 
Even more constraining, the deepest pole in the differential equation is dictated by a certain combinatorial filtration of the basis elements, 
whereas the maximum exponent in the $\eps$ expansion is bounded to be one. We say that such a differential equation is $F^{\bullet}$-compatible.

In the second step of the algorithm, we employ a recursive procedure to find a transformation that leads to $\eps$-factorised form. 
We prove that such a transformation can be systematically found if the starting differential equation is $F^{\bullet}$-compatible. 
At each step of this procedure, we transform the differential equations such that the deepest pole in $\eps$ is removed until all terms are  $\eps$-factorised.

We demonstrate the applicability of the method on a range of examples of varying complexity. 
In particular, we first show that the algorithm correctly produces a basis with desired properties in polylogarithmic cases, 
such as the massless double-box integral and the massless two-loop pentabox.
More importantly, the algorithm is designed for cutting-edge Feynman integrals, characterised by non-trivial geometries with sub-structures, 
like an elliptic curve with additional punctures or a K3-surface with punctures.
For the elliptic case we present two non-trivial examples in section~\ref{sect:sector_93_moeller} and section~\ref{sect:electron_self_energy}.
The case of a K3-surface appeared in a separate publication \cite{Pogel:2025bca}.

Among possible future avenues, it would be interesting to establish, for step one, mathematical proofs of our findings. 
In particular, one can attempt to provide a proof of existence for our candidate basis following the dlog-filtration of twisted cohomology introduced in \cite{Aomoto2015}.

From a practical standpoint, it will be interesting to stress-test the algorithm against even more demanding examples with more scales, more loops, and more Baikov variables. 
Furthermore, the suggested algorithm provides an alternative setup for integration-by-parts reduction, both in the generation of equations and ordering criteria. 
It would be interesting to combine those with other efficiency improvements, e.g.~the ones proposed in refs.~\cite{Peraro:2019svx,Guan:2024byi,Lange:2025fba,Smith:2025xes}. 

As a potential improvement of the algorithm, we can exploit the remaining degeneracy among the differential forms that are equal under the ordering criteria on the maximal cut, but differ by lower-sector terms.
Presently, we do not take advantage of this. 
This could yield additional gains in the performance of the algorithm or provide a more compact lower sector coupling coefficients. 
The latter seems to be hinted at in ref.~\cite{Pogel:2025bca}, where projecting onto the dot basis has resulted in the same expression on the maximal cut, while offering a significant simplification of lower sector couplings. 

An additional improvement can be the incorporation of symmetries directly in twisted cohomology following the ideas of \cite{Gasparotto:2023roh,Duhr:2025xyy}.
This will improve the efficiency in examples which have a high degree of symmetry.
 
Furthermore, in ref.~\cite{Pogel:2024sdi}, it has proven useful to find a self-dual form of the $\eps$-factorised differential equations. 
This was demonstrated to reduce the number of implicitly defined functions in the system. 
It would be interesting to see whether this operation can be performed independently of the underlying geometry, as it would improve the performance of the second step of the algorithm.

\subsection*{Acknowledgements}

This work has been supported by the Research Unit ``Modern Foundations of Scattering Amplitudes'' (FOR 5582)
funded by the German Research Foundation (DFG).
X.W. is supported by the University Development Fund of The Chinese University of Hong Kong, Shenzhen, under the Grant No. UDF01003912.
X.W. is also supported in part by the National Natural Science Foundation of China with Grant No. 12535006.
The work of F.G. is supported by the European Research Council (ERC) under the European Union’s Horizon 2022
Research and Innovation Program (ERC Consolidator Grant LoCoMotive 101043686).
X.X. has received funding from the European Research Council (ERC) under the European Union’s Horizon 2022
Research and Innovation Program (ERC Advanced Grant No.~101097780, EFT4jets). 
Views and opinions expressed
are however those of the authors only and do not necessarily reflect those of the European Union or the European
Research Council Executive Agency. Neither the European Union nor the granting authority can be held responsible for
them.

{\footnotesize
\bibliography{biblio}
\bibliographystyle{h-physrev5}
}

\end{document}